\documentclass[12pt]{article}

\def\anon{0}
\newif\ifanon \ifnum\anon=1 \anontrue\else\anonfalse\fi

\usepackage[margin=1in]{geometry}
\usepackage{amsmath,amssymb,amsthm}
\usepackage{natbib}
\usepackage{booktabs}
\usepackage{graphicx}
\usepackage{float}
\usepackage{xcolor}
\usepackage{setspace}
\usepackage{enumitem}
\usepackage[T1]{fontenc}
\usepackage[utf8]{inputenc}
\usepackage{lmodern}
\usepackage{microtype}
\usepackage{multirow}
\usepackage{placeins}
\usepackage{caption}
\usepackage{pdflscape}
\usepackage{hyperref}
\usepackage{titletoc}
\usepackage{pdflscape}
\usepackage{rotating}
\usepackage{afterpage}
\titlecontents{section}[1.5em]
  {\addvspace{0.5em}}
  {\contentslabel{1.5em}}{\hspace*{-1.5em}}
  {\titlerule*[0.5pc]{.}\contentspage}
\titlecontents{subsection}[3.8em]
  {}
  {\contentslabel{2.3em}}{\hspace*{-2.3em}}
  {\titlerule*[0.5pc]{.}\contentspage}
\titlecontents{paragraph}[6em]
  {\small}{}{}
  {\titlerule*[0.5pc]{.}\contentspage}

\newcommand{\E}{\mathbb{E}}
\newcommand{\Var}{\mathrm{Var}}
\newcommand{\Cov}{\mathrm{Cov}}
\renewcommand{\Pr}{\mathrm{Pr}}
\newcommand{\bfbeta}{\boldsymbol{\beta}}
\newcommand{\bfX}{\mathbf{X}}
\newcommand{\bfZ}{\mathbf{Z}}
\newcommand{\bfDelta}{\boldsymbol{\Delta}}
\newcommand{\bfLambda}{\boldsymbol{\Lambda}}

\newcommand{\bfW}{\mathbf{W}}
\newcommand{\bfb}{\mathbf{b}}
\newcommand{\bfh}{\mathbf{h}}
\newtheorem{proposition}{Proposition}[section]

\newcommand{\bfeta}{\boldsymbol{\eta}}
\newcommand{\bfSigma}{\boldsymbol{\Sigma}}
\newcommand{\bfS}{\mathbf{S}}
\newcommand{\bfJ}{\mathbf{J}}
\newcommand{\bfc}{\mathbf{c}}
\newcommand{\IF}{\mathrm{IF}}

\makeatletter
\g@addto@macro\normalsize{
  \setlength{\abovedisplayskip}{6pt plus 2pt minus 2pt}
  \setlength{\belowdisplayskip}{6pt plus 2pt minus 2pt}
  \setlength{\abovedisplayshortskip}{3pt plus 1pt minus 1pt}
  \setlength{\belowdisplayshortskip}{4pt plus 1pt minus 2pt}
}
\makeatother

\hypersetup{
  colorlinks=true,
  linkcolor=blue!60!black,
  citecolor=blue!60!black,
  urlcolor=blue!60!black
}

\title{\Large \textbf{Learning Preferences from Conjoint Data:\\ A Hybrid Structural Deep Learning Approach}\ifanon\else\thanks{We thank P. M. Aronow, Cameron Ballard-Rosa, Kosuke Imai, Will Marble, Mingcong Pan, Sparsha Saha, Milan Svolik, and Ana Weeks for helpful comments.}\fi}
\author{
  \ifanon
  Author names withheld for review
  \else
  Avidit Acharya\footnote{Department of Political Science, Stanford University, Stanford, CA. Emails: \url{avidit@stanford.edu}, \url{jhain@stanford.edu}, and \url{yiqingxu@stanford.edu}.}\\Stanford
  \and
  Jens Hainmueller$^\dagger$\\Stanford
  \and
  Yiqing Xu$^\dagger$\\Stanford
  \fi
}
\date{\today}

\begin{document}
\maketitle

\begin{abstract}
\begin{singlespace}
\noindent Conjoint experiments randomize multidimensional profiles, yet political science applications typically report only nonparametric averages that do not recover counterfactual choices or individual tradeoffs. We develop a hybrid structural approach for recovering individual preferences from conjoint data. The estimator combines a flexible machine-learning mean preference function, via a deep neural network in our applications, with respondent-level empirical-Bayes updating in a logistic random utility model, allowing preferences to vary with observed characteristics while learning residual heterogeneity from repeated choices. Double/debiased machine learning delivers valid inference for population-average preference parameters with any sufficiently accurate first-stage learner. Across three applications, the method reveals heterogeneity reduced-form averages obscure: opposition to undemocratic behavior is broad but uneven in intensity, progressive tax preferences are widespread across partisan subgroups, and partisan polarization offsets the average gender effect in candidate choice. The framework opens the door to core theoretical questions in political science by recovering substantively interpretable structural parameters.

\end{singlespace}
\end{abstract}

\thispagestyle{empty}
\newpage
\setcounter{page}{1}
\doublespacing

\section{Introduction}\label{sec:intro}

Conjoint experiments help researchers understand how individuals trade off different dimensions of preference: by presenting respondents with multidimensional profiles and asking them to choose, they directly elicit the multi-attribute tradeoffs inherent in real-world decisions.  This insight was recognized early: \citet{greenhalgh1981conjoint} used conjoint analysis to study negotiator preferences over contract terms, and \citet{shamir1995competing} explicitly framed conjoint designs as a way to measure value tradeoffs and interactions in mass opinion.  Since \citet{hainmueller2014causal} introduced conjoint experiments to political science, applications have proliferated across a number of areas from immigration policy \citep{hainmueller2015hidden, bansak2016europeans} to candidate evaluation \citep{saha2022ambitious}, democratic accountability \citep{graham2020democracy}, and tax policy preferences \citep{ballard2017structure}.

The ability of conjoints to elicit multi-attribute tradeoffs connects tightly to theoretical quantities that arise from specifying an underlying utility function.  Under such a model, conjoint data can recover not only marginal rates of substitution (MRS), willingness to pay (WTP), and the distribution of preferences across a population, but also vote shares and win probabilities in counterfactual electoral contests.  These quantities are central to theories of voting, representation, and political economy, and conjoint experiments are ideally suited to estimate them.  But because estimating them requires imposing parametric structural assumptions on preferences, and because the recent turn toward design-based causal inference in political methodology has been skeptical of such structural assumptions, the development of conjoint analysis in political science focused primarily on nonparametric causal estimands.  The average marginal component effect (AMCE) introduced by \citet{hainmueller2014causal} has dominated applied work in the field.  A small but growing set of recent papers broadens this tradition: Bayesian Additive Regression Trees model effect heterogeneity \citep{robinson2024detect} and design-based tests probe for interactions \citep{ham2024machine}, while \citet{goplerud2025estimating} recover respondent-level heterogeneity with treatment interactions and covariate-modeled group membership and \citet{zhirkov2022individual} estimates individual marginal component effects.  These approaches enrich estimation and testing, yet they do not target the structural tradeoff, distributional, and electoral quantities that motivate our approach.

Skepticism of structural assumptions is well-motivated---such models can be wrong and misspecification costly.  We address this concern by introducing a method that combines the conjoint design with a flexible machine-learning first stage---a deep neural network (DNN) \citep{farrell2021deep, farrell2025deep} in our main analyses, with elastic net and generalized random forests (GRF) as comparable alternatives---together with respondent-level empirical-Bayes updating and double/debiased machine learning (DML) \citep{chernozhukov2018double}.  We adopt the standard random utility framework: respondent $i$'s utility for profile $j$ in task $t$ is $U_{ijt} = \bfX_{ijt}^{\top}\bfbeta_i + \varepsilon_{ijt}$, where $\bfX_{ijt} \in \mathbb{R}^p$ is the profile's vector of attribute levels and $\bfbeta_i \in \mathbb{R}^p$ is respondent $i$'s vector of marginal utilities---one per attribute level---governing how much that respondent rewards or penalizes each attribute when choosing between profiles.  Respondent $i$ chooses the higher-utility profile, so the binary choice probability is a logit in the attribute contrast $\bfDelta\bfX_{it} = \bfX_{i,1,t} - \bfX_{i,2,t}$ with utility difference $\bfDelta\bfX_{it}^{\top}\bfbeta_i$.

In departure from the rigid parametric models that have dominated structural conjoint analysis, we write the latent preference vector as $\bfbeta_i = f(\bfZ_i) + \boldsymbol{\eta}_i$, where a flexible first-stage learner---a DNN in our applications---estimates the systematic component $f(\bfZ_i)$ from the full sample and each respondent's own repeated conjoint choices update the residual component $\boldsymbol{\eta}_i$.  This hybrid estimator preserves the full logit structure of the random utility model: the predicted choice probability for any contest is determined jointly by the randomly assigned profile contrast and the respondent's estimated preference vector, and every tradeoff, distributional, and electoral quantity described above becomes computable from $\hat{\bfbeta}_i$.  This distinguishes our approach from standard mixed-logit and hierarchical Bayes models \citep{lenk1996hierarchical, rossi2005bayesian, train2009discrete}, which typically impose a linear mean structure, $\bfbeta_i \sim \mathcal{N}(\boldsymbol{B}\bfZ_i, \boldsymbol{\Sigma})$, or drop $\bfZ_i$ entirely when it is rich and the model becomes difficult to fit.  Relative to a pure DNN that uses only $f(\bfZ_i)$, the respondent-level updating recovers residual heterogeneity within covariate strata.  For inference on population-average preference parameters, we use cross-fitting and the influence-function correction of DML, which provides valid confidence intervals (CIs) despite the flexibility of the first-stage estimator.

The power of our approach comes from combining three sources of leverage: (a) the flexible machine-learning first stage, (b) the identification power of the conjoint's randomized profile contrasts, and (c) the within-respondent information in repeated tasks that updates residual preferences beyond what observables explain. We view this approach as \textit{lean structure}: it preserves the advantages of a structural model while relaxing its usual parametric restrictions.

The payoff is substantial, and it speaks directly to an active debate about what conjoint experiments can and cannot deliver. \citet{bansak2023amce} review what conjoint experiments can---and cannot---aggregate, showing that different quantities of interest carry different substantive meanings; a central implication is that many theoretically natural quantities---in particular, the fraction of voters who prefer an attribute level---cannot be point-identified from the AMCE and related design-based averages, which recover only marginal population means and so leave individual-level distributional features of preferences unidentified without additional structure \citep{abramson2022voter}.  Our paper develops a flexible structural framework that recovers individual-level preferences directly. It therefore delivers the quantities of interest that \citeauthor{bansak2023amce} and others identify as substantively important but beyond the reach of existing reduced-form approaches. At the same time, it recovers standard reduced-form quantities, such as the AMCE, as special cases under correct specification---a condition the misspecification diagnostic of Section~\ref{sec:quantities} is designed to probe. Among the structural quantities it delivers are vote shares and win probabilities in counterfactual electoral contests over arbitrary policy positionings---the structural-model analog of the questions that have driven a half-century of work on candidate moderation and spatial voting \citep{black1948rationale,downs1957economic}.

We apply our method to three prominent conjoint studies, showing how the structural approach reveals deep preference heterogeneity.  In the \citet{graham2020democracy} democracy conjoint, the structural model reveals that almost all voters \textit{oppose} undemocratic behavior, but many weight party and policy more heavily---the disagreement is about intensity, not direction.  Translated to electoral counterfactuals, a co-partisan candidate retains majority support among the respondents on most of the seven tested undemocratic actions; prosecuting journalists is the clear exception---and once the candidate has endorsed prosecuting journalists, they hold both liberal and conservative majorities only near the center of the social-policy spectrum.  In the \citet{ballard2017structure} tax-plan experiment, we recover an individual-level preferred rate schedule for each respondent and find that progressive revealed preferences are the clear majority pattern in every partisan, income, and ideological subgroup, while Democrats and Republicans differ primarily in \textit{which} brackets drive their choices.  The same application provides individual-level external validation. The hybrid-recovered progressivity slope correlates with each respondent's self-reported ideal tax rates, with a correlation coefficient of $r = 0.43$. This measure was collected independently of the conjoint task and was never observed by the model during training. It therefore provides a validation check that reduced-form estimators cannot perform. The \citet{saha2022ambitious} candidate-choice experiment shows what the model can describe even in a deliberately sparse design of three tasks per respondent: a near-zero average gender effect coexists with offsetting partisan camps, with a majority of Democrats estimated on the female side and a majority of Republicans on the male side, so that gender registers as a leading source of individual-level variance even though its population mean is null. Because the democracy and tax applications use richer task counts, they carry the greatest evidentiary weight. Taken together, these results show that the structural approach goes meaningfully beyond AMCE-style analysis, recovering individual-level preference distributions, electoral counterfactuals, and externally-validated estimates that nonparametric estimands leave on the table---opening a new agenda for studying candidate positioning, electoral competition, and the structure of mass ideology.

The marketing literature has long recognized that conjoint studies can identify the structural parameters of an underlying utility model. The conjoint measurement foundations laid by \citet{luce1964simultaneous} and \citet{green1971conjoint} were linked to random utility models by \citet{mcfadden1974conditional}, and the connection between conjoint experiments and structural discrete choice has been a staple of marketing research \citep{green1978conjoint, green1990conjoint, louviere2000stated, train2009discrete}.\footnote{The introduction of conjoint analysis to political science by \citet{hainmueller2014causal} situates it within the older traditions of conjoint measurement \citep{luce1964simultaneous} and marketing \citep{green1971conjoint} and relates their nonparametric estimator to model-based discrete choice such as conditional logit \citep{mcfadden1974conditional}, though they remain agnostic about the underlying behavioral model; see \citet{bansak2021conjoint} for a comprehensive handbook treatment.  Our contribution is to develop this structural approach and the quantities it enables.}  Random utility models have also been used extensively in political science---notably the probabilistic spatial-utility frameworks of \citet{poole1985spatial} and \citet{palfrey1987relationship}, the heterogeneous voter-choice models of \citet{rivers1988heterogeneity} and \citet{alvarez1998multiparty}, and the Bayesian ideal-point models of \citet{martin2002dynamic} and \citet{clinton2004statistical}---and the formal literature on spatial voting and preference aggregation \citep{enelow1984spatial, hinich1997analytical, austensmith1999positive} provides the theoretical foundations for the structural quantities we recover.  But these traditions have not been linked to conjoint experiments in a way that combines flexible systematic heterogeneity with respondent-level updating.  Our paper makes this connection and, by combining flexible DNN estimation of the mean structure with a structural updating step for residual heterogeneity, addresses the main worry about parametric identification: that the specified model may not capture the underlying data generating process.

\section{Theoretical Framework}\label{sec:framework}

This section develops the structural model that underlies our approach: the conjoint setup, the random utility model connecting attributes to choices, the structural quantities of interest, and their identification by the conjoint randomization design.

\subsection{Setup}\label{sec:setup}

Consider a conjoint experiment in which each of $N$ respondents evaluates $T$ choice tasks.  In each task $t$, respondent~$i$ views two profiles (alternatives) and selects one.\footnote{We focus on the binary forced-choice case because it is the most common conjoint design in political science, but the framework extends naturally to multinomial choice ($J \geq 3$ alternatives), rankings, and rating-scale outcomes; see Section~\ref{sec:limitations}.}  Each profile $j \in \{1, 2\}$ is characterized by a vector of randomly assigned attribute levels, encoded as dummy variables relative to a reference category: $\bfX_{ijt} \in \mathbb{R}^p$, where $p$ is the total number of non-reference attribute levels across all attributes.  To fix ideas, consider the candidate-choice experiment of \citet{saha2022ambitious}, where respondents evaluate hypothetical candidates described by five attributes (policy agenda, talent, number of children, gender, and progressive ambition), yielding $p = 13$ dummy-coded attribute levels.

We define the \textit{profile-pair difference}
\begin{equation}\label{eq:deltaX}
  \bfDelta\bfX_{it} \;=\; \bfX_{i1t} - \bfX_{i2t} \;\in\; \mathbb{R}^p,
\end{equation}
which captures the contrast between the two profiles on all attribute dimensions simultaneously.  Let $Y_{it} \in \{0,1\}$ denote whether respondent~$i$ chose profile~1 in task $t$.  The total number of choice observations is $NT$.

Respondent characteristics are collected in a vector $\bfZ_i \in \mathbb{R}^{p_Z}$, which may include demographics, attitudinal measures, and contextual variables.  In the candidate-choice experiment, $\bfZ_i$ comprises 19 variables, including party identification, ideology, demographics, employment status, region, 2016 vote choice, and gender attitudes.  Crucially, $\bfZ_i$ is constant across all tasks for respondent~$i$.

\subsection{The Random Utility Model}\label{sec:rum}

We ground our framework in the canonical random utility model of \citet{mcfadden1974conditional}. Respondent~$i$ derives utility from profile $j$ in task $t$ according to
\begin{equation}\label{eq:utility}
  U_{ijt} = \bfX_{ijt}^\top \bfbeta_i + \varepsilon_{ijt},
\end{equation}
where $\bfbeta_i \in \mathbb{R}^p$ is respondent~$i$'s latent vector of \textit{marginal utilities}, and $\varepsilon_{ijt}$ are independent and identically distributed Type~I Extreme Value (Gumbel) taste shocks. Because the Gumbel shocks normalize the choice scale, $\bfbeta_i$ is the scale-normalized preference vector. If respondent~$i$'s latent utility has an unnormalized coefficient vector $\bfbeta_i^\ast$ and error scale $\sigma_i$, the choice likelihood identifies $\bfbeta_i^\ast/\sigma_i$, and this ratio is the $\bfbeta_i$ we estimate. Our hybrid model decomposes these preferences as
\begin{equation}\label{eq:beta_decomp}
  \bfbeta_i = f(\bfZ_i) + \boldsymbol{\eta}_i.
\end{equation}
We maintain this mean-logit random utility model throughout. Under that maintained specification, $f(\bfZ_i)=\E[\bfbeta_i \mid \bfZ_i]$ is the systematic component of the scale-normalized structural preference vector and $\boldsymbol{\eta}_i$ is a respondent-specific residual with $\E[\boldsymbol{\eta}_i \mid \bfZ_i]=\mathbf{0}$. Thus, when we refer to ``preferences,'' ``mean preferences,'' or ``preference parameters,'' we mean these structural objects. If the same estimating equations are instead read only as a working conditional-logit approximation, the first-stage target is the conditional-logit projection $f_0$ characterized in Supplementary Materials~\ref{app:asymptotic_results} condition (H2).\footnote{Because a mixture of logits is not generally a logit, $f_0(\bfZ_i)$ need not equal $\E[\bfbeta_i \mid \bfZ_i]$. So in the weaker conditional-logit projection interpretation, the corresponding average quantities are projection preferences rather than exact latent preference means.} This formulation nests two familiar special cases.  Setting $\boldsymbol{\eta}_i=\mathbf{0}$ yields the pure DNN model, in which preferences depend only on observed covariates.  Restricting $f(\bfZ_i)=\bfZ_i\mathbf{B}$ yields the linear mean structure used in standard hierarchical logit and mixed-logit models.

Each coefficient $\beta_{ik}$ is the marginal utility that respondent~$i$ assigns to attribute level $k$ relative to the reference level of the corresponding attribute.  For instance, in the candidate-choice experiment covered below, $\beta_{i,\text{Female}}$ measures how much respondent~$i$ values a female candidate relative to a male candidate.  The key modeling choice is that the systematic mean $f(\cdot)$ is nonparametric, as we impose no parametric restriction on how respondent characteristics relate to preferences. This distinguishes our approach from standard mixed logit or hierarchical Bayes.

The respondent chooses profile~1 if $U_{i1t} > U_{i2t}$.  Since the difference of two independent and identically distributed Gumbel random variables follows a logistic distribution, the forced-choice probability conditional on respondent~$i$'s preferences is
\begin{equation}\label{eq:choiceprob}
  \Pr(Y_{it} = 1 \mid \bfDelta\bfX_{it}, \bfbeta_i) = G\!\left(\bfDelta\bfX_{it}^\top \bfbeta_i\right),
\end{equation}
where $G(v) = (1 + e^{-v})^{-1}$ is the logistic cumulative distribution function.  The profile-pair differencing \eqref{eq:deltaX} eliminates any alternative-specific constant, so we need only estimate respondent-specific preference vectors.  In practice, our estimator uses the cross-sectional information to learn $f(\bfZ_i)$ and then uses each respondent's repeated choices to update the residual component $\boldsymbol{\eta}_i$.

The key assumption of this model is an additive utility that is linear in attribute levels, which rules out attribute interactions unless explicitly included.  The logistic link from the Gumbel error distribution is a secondary modeling choice---common alternatives such as probit are nearly indistinguishable in practice.  While $f(\bfZ)$ is left fully nonparametric, the additive utility structure is what distinguishes the approach from purely design-based methods and is the source of its additional identifying power.  More fundamentally, the random utility framework makes the recovered $\hat{\bfbeta}_i$ interpretable as preferences rather than coefficients in a flexible classifier. Every structural quantity we report, including MRS, WTP, compensating differentials, and vote shares in counterfactual contests, is a parameter of an economic model. Without this framework, the DNN-plus-empirical-Bayes approach would remain a useful flexible prediction tool. It would still recover heterogeneous predicted choice probabilities, direction-versus-intensity decompositions, polarization, attribute-importance shares, and subgroup heterogeneity as features of a flexible classifier.  What the random utility model specifically provides is a utility-cardinality interpretation that captures tradeoffs: the same $\hat{\bfbeta}_i$ can be compared across attributes (MRS, WTP), aggregated into vote shares, and used to compute compensating differentials. We discuss its limitations and possible extensions in Section~\ref{sec:limitations}.

\subsection{Quantities of Interest}\label{sec:quantities}

The structural model \eqref{eq:utility}--\eqref{eq:choiceprob} enables quantities inaccessible to reduced-form analysis.  Some depend only on the population mean of preferences; others depend on the respondent-specific $\bfbeta_i$, which our hybrid estimator treats as empirical-Bayes posterior summaries under \eqref{eq:beta_decomp}.  We focus on the quantities estimated in our applications.

\begin{enumerate}[leftmargin=*]

\item \textit{Average preference parameters.}  The primary estimand is the population-average structural preference for attribute level $k$:
\[
  \theta_k = \E[\beta_{ik}] = \E[f_k(\bfZ_i)], \quad k = 1, \ldots, p.
\]
On the logit scale, $\theta_k$ summarizes how much the average voter rewards or penalizes attribute $k$. Under the working-model interpretation described above, the same estimating procedure targets $\E[f_{0k}(\bfZ_i)]$, an average conditional-logit projection preference.

\item \textit{Average marginal effects.}  The average preference parameter $\theta_k$ lives on the logit scale.  Its probability-scale counterpart is the average marginal effect (AME):
\[
  \text{AME}_k = \E_{\bfbeta_i, \bfX_{-k}}\!\left[\,G\!\left(\beta_{ik} + \bfX_{-k}^\top \bfbeta_{i,-k}\right) - G\!\left(\bfX_{-k}^\top \bfbeta_{i,-k}\right)\right],
\]
the average change in the probability of choosing a profile when level $k$ is switched on, averaging over respondents and the randomization of all other attributes.  Under correct specification the AME equals the AMCE of \citet{hainmueller2014causal}; we use any discrepancy beyond sampling noise as a misspecification diagnostic, and in the democracy application the two agree closely across all 30 attribute levels (Figure~\ref{fig:gs_beta_ridgelines}B).

\item \textit{Individual preference vectors.}  The estimator recovers $\hat{\bfbeta}_i \in \mathbb{R}^p$ for each respondent---the complete vector of marginal utilities across all attribute levels simultaneously.  In our implementation,
\[
  \hat{\bfbeta}_i = \hat f(\bfZ_i) + \hat{\boldsymbol{\eta}}_i,
\]
combining the cross-fitted DNN mean $\hat f(\bfZ_i)$ with the empirical-Bayes update $\hat{\boldsymbol{\eta}}_i$ from the respondent's repeated choices.  This is the central structural object: every quantity below is a functional of it, and reduced-form methods cannot deliver it.

\item \textit{Counterfactual choice probabilities and electoral competition.}  For any pair of hypothetical profiles $A$ and $B$ with attribute vectors $\bfX_A$ and $\bfX_B$:
\[
  \Pr(\text{choose } A \text{ over } B \mid \bfbeta_i) = G\!\left((\bfX_A - \bfX_B)^\top \hat{\bfbeta}_i\right).
\]
This predicts any head-to-head contest between fully specified profiles for any respondent or subgroup; aggregated across the electorate it yields counterfactual vote shares as candidates move through policy positions.  Recovering it requires the joint preference vector, because the nonlinear logistic link makes the combined contrast more than the sum of marginal effects.

\item \textit{Preference polarization.}  The sign split of $\beta_{ik}$ across respondents separates the \textit{direction} of a preference from its \textit{intensity}.  We summarize it by the polarization fraction, the population share who favor level $k$:
\[
  \pi_k = \Pr(\beta_{ik} > 0) = \E\!\left[\mathbf{1}\{\beta_{ik} > 0\}\right],
\]
so that $1 - \pi_k = \Pr(\beta_{ik} < 0)$ is the share who oppose it.  A value of $\pi_k$ near $\tfrac{1}{2}$ marks an attribute on which the electorate divides into offsetting camps of strong supporters and opponents.  This distinguishes genuine consensus from a near-zero average $\theta_k$ that masks such a split---a distinction the AMCE cannot make.

\item \textit{Attribute importance.}  Under conjoint randomization, the variance of utility decomposes additively across attributes:
\[
  \Var\!\left(\bfX^\top \bfbeta_i\right) = \sum_{a=1}^A \Var\!\left(\bfX_a^\top \bfbeta_{i,a}\right),
\]
where $a$ indexes attributes.\footnote{This decomposition is exact when attributes are randomized independently. We compute each attribute's contribution as the exact block variance $\Var(\bfX_a^\top\bfbeta_{i,a})$ over the design distribution, which for multi-level categorical attributes retains the cross-level covariances among its mutually exclusive level dummies. The naive sum-of-squared-effects measure $\sum_{k \in a}\hat\beta_{ik}^2\,\Var(X_k)$ omits these covariances and overstates co-signed multi-level blocks; for the seven undemocratic actions in Section~\ref{sec:graham_svolik} it would report a $24\%$ share against the $9\%$ of the exact decomposition. We report the exact shares throughout.}  Normalizing to shares gives each respondent an importance ranking over attributes.  The resulting heterogeneity in \textit{what} voters weight---some focused on one dimension, others spread across many---is invisible to average-effect analysis.

\item \textit{Marginal rates of substitution (MRS).}  The tradeoff between attributes $j$ and $k$ for respondent~$i$ is:
\[
  \text{MRS}_{ijk} = -\frac{\beta_{ij}}{\beta_{ik}},
\]
the units of attribute $k$ needed to compensate a one-unit change in attribute $j$, expressed in the denominator attribute's units; with a monetary denominator it reduces to WTP.  At the population level, the corresponding quantity is the negative ratio of \emph{average} preference parameters,
\[
  \text{MRS}_{jk} = -\frac{\theta_j}{\theta_k}, \qquad \text{WTP}_j = -\frac{\theta_j}{\theta_{\text{money}}},
\]
which admits valid $\sqrt{N}$ inference by the delta method, with a Fieller interval near a vanishing denominator (Supplementary Materials~\ref{app:other_estimands}).  The population negative ratio is the more stable target; individual MRS is sensitive to near-zero denominators, so we trim in practice.

\item \textit{Compensating differentials.}  For a penalty attribute ($\beta_{ij} < 0$), the compensating differential is the fraction of respondents for whom some benefit $k$ satisfies $\beta_{ij} + \beta_{ik} \geq 0$---the discrete ``would you take the deal?'' counterpart to the MRS, requiring the joint within-respondent distribution of preferences.

\end{enumerate}

The preference vectors also support consumer surplus \citep{mcfadden1981econometric}, preference clustering, inequality measures, and the \textit{majority preference function}---for any pair $(A, B)$, the fraction of respondents with $(\bfX_A - \bfX_B)^\top \hat{\bfbeta}_i > 0$.\footnote{Unlike logit choice probabilities, the majority preference function abstracts from idiosyncratic taste shocks and directly addresses the concern raised by \citet{abramson2022voter} that the AMCE can indicate the opposite of the true majority preference.}  See \citet{train2009discrete}, \citet{louviere2000stated}, and \citet{green1990conjoint} for detailed treatments.

\paragraph{Relations to the AMCE.}  The AMCE, which \citet{bansak2023amce} show maps to aggregate vote shares, marginalizes one attribute over all others and so cannot recover the joint vector $\bfbeta_i$ or its functionals---MRS, counterfactual choice probabilities, compensating differentials, importance, polarization; conditional AMCEs \citep{hainmueller2014causal, leeper2020measuring} still average within subgroups.  \citet{delacuesta2022improving} show the AMCE is design-dependent under heterogeneity, and \citet{abramson2022voter} that it can reverse the majority preference because it reflects preference intensity, not direction alone.  Section~\ref{sec:applications} illustrates these distinctions.

\subsection{Identification}\label{sec:ident}

All of the quantities of interest defined above are derived from either the systematic mean function $f(\bfZ)$ or the respondent-specific preference vector $\bfbeta_i = f(\bfZ_i) + \boldsymbol{\eta}_i$.  It is therefore useful to separate what is identified by randomization and the structural model from what is learned by respondent-level updating.

By construction, $\bfDelta\bfX_{it}$ is randomly assigned and therefore independent of $\bfZ_i$:
\[
  \bfDelta\bfX_{it} \perp \!\!\! \perp \bfZ_i.
\]
This exogeneity, guaranteed by the design, identifies the conditional choice distribution without the omitted-variable bias or endogenous sorting that plague observational discrete choice.

To see what this buys us, consider a simplified version of the candidate-choice experiment with two binary attributes (candidate gender and policy agenda, reform versus status quo) and a binary respondent characteristic (college degree).  Randomization means that within each education group, differences in choice probabilities across profile contrasts identify the systematic mean coefficients directly: a comparison of female versus male candidates (holding the agenda fixed) isolates $f_{\text{Female}}(z)$, and a comparison of reform versus status-quo candidates (holding gender fixed) isolates $f_{\text{Agenda}}(z)$.  Once these coordinates are identified, the structural model combines these marginal utilities into joint quantities of interest.

We emphasize that identification here has two components.  First, experimental randomization nonparametrically identifies the conditional choice probability $\Pr(Y_{it}=1 \mid \bfDelta\bfX_{it}, \bfZ_i)$, eliminating the endogeneity concerns that dominate observational discrete choice.  Second, the logit functional form links those probabilities to latent utilities, which are linear in $\bfX_{it}$, and thereby identifies the systematic preference function $f(\bfZ)=\E[\bfbeta_i\mid\bfZ_i]$ and the population-average structural preferences $\theta_k=\E[\beta_{ik}]=\E[f_k(\bfZ_i)]$ under the maintained mean-logit model. If the same equations are treated only as a working conditional-logit approximation, the identified mean-stage object is instead the projection $f_0(\bfZ)$ and the average target is $\E[f_{0k}(\bfZ_i)]$. The respondent-specific residual $\boldsymbol{\eta}_i$ is different: with finite tasks per respondent it is not point-identified from the design alone, and is instead updated by respondent-level likelihood information under the hybrid model.  The individual-level quantities we report should therefore be understood as empirical-Bayes posterior summaries under \eqref{eq:beta_decomp}, not as design-based point identification without additional structure.

Two features make the structural component relatively modest. First, as the random utility model makes clear, the parametric restriction applies only to the mapping from preferences to choice probabilities, not to the flexible (DNN-estimated) heterogeneity itself. Second, as noted above, the choice of link is secondary in practice---logit and probit are nearly indistinguishable in the interior of the probability space, and the DNN can partially compensate for link misspecification by rescaling the preference vector---so the binding structural assumption is additive utility.

\medskip
The key distinction is between population-average quantities, which admit design-based DML inference, and distributional or individual-level quantities, which rely more heavily on the respondent-level structural update.  Population averages and smooth functionals of the mean stage inherit the experimental design's inferential guarantee, whereas population distributional summaries and individual-level quantities are model-based at fixed $T$ and become more reliable as the number of tasks per respondent grows.

\section{Estimation}\label{sec:estimation}

We estimate the hybrid model in three stages: (i)~estimate the systematic mean $f(\bfZ)$ with a cross-fitted flexible learner (a DNN throughout), (ii)~recover $\hat{\bfbeta}_i$ by respondent-level empirical-Bayes updating, and (iii)~apply the DML correction for valid inference on population averages.  Detailed architecture, tuning, variance construction, and diagnostics are in Supplementary Materials~\ref{app:estimation_details}.

The first stage is learner-agnostic: any supervised learner that estimates $f(\bfZ)$ at a fast enough rate enters the same debiased score and delivers the same valid inference.  We use a deep neural network by default; elastic-net and GRF first stages perform comparably (Supplementary Materials~\ref{app:learners}, \ref{app:thetak_mc}).  The DNN maps covariates to a mean preference vector, $\hat f(\bfZ_i)\in\mathbb{R}^p$, which enters the choice probability through the logit index $\bfDelta\bfX_{it}^\top \hat f(\bfZ_i)$. The map from $\bfZ_i$ to preferences can be arbitrarily nonlinear, but the resulting preference vector still enters utility linearly. We train on the heterogeneous-logit likelihood with respondent-level $K$-fold cross-fitting \citep{chernozhukov2018double}---all of a respondent's tasks in one fold---so out-of-fold estimates never reuse a respondent's data for both training and evaluation; production runs average two independently seeded cross-fits (Supplementary Materials~\ref{app:estimation_details}).

The natural fully modeled second stage is a DNN-offset mixed logit or hierarchical logit:
\[
  \bfbeta_i = f(\bfZ_i) + \boldsymbol{\eta}_i, \qquad \boldsymbol{\eta}_i \sim N(\mathbf{0}, \boldsymbol{\Sigma}_\eta).
\]
In typical conjoint designs the tasks per respondent are few relative to the preference coefficients, so fully estimating $\boldsymbol{\Sigma}_\eta$ is unstable.  We therefore use a simpler empirical-Bayes implementation with prior mean the two-seed cross-fitted ensemble from the first stage:
\[
  \hat f_{\text{ens}}(\bfZ_i) = \frac{1}{2}\left(\hat f^{(1)}(\bfZ_i) + \hat f^{(2)}(\bfZ_i)\right).
\]
We use the score-based scale heuristic $\hat\sigma_{\text{score},k}^2$ as a coordinate-specific scale for residual heterogeneity.
At low $T$, using this raw scale leaves the prior too diffuse, so a few binary choices can move $\hat{\bfbeta}_i$ excessively off the DNN mean.  We therefore add a fixed prior-precision constant $c_\eta$ and form the diagonal working covariance $\hat{\boldsymbol{\Sigma}}_\eta(c_\eta)=\operatorname{diag}(\hat\sigma_{\eta,1}^2(c_\eta),\ldots,\hat\sigma_{\eta,p}^2(c_\eta))$ as
\[
  \hat{\sigma}_{\eta,k}^2(c_\eta)
  =
  \frac{\hat{\sigma}_{\text{score},k}^2}{c_\eta},
  \qquad k = 1,\ldots,p.
\]
Our default uses $c_\eta=5$, selected ex ante from the simulation diagnostics in Supplementary Materials~\ref{app:estimation_details}---a regularization calibration, not a claim about the true variance of $\eta_{ik}$.\footnote{On a large targeted simulation grid, $c_\eta=5$ gave the best average individual-level recovery and was statistically indistinguishable from a more adaptive MAP calibration, so we use the simpler fixed rule (Supplementary Materials~\ref{app:estimation_details}).}  The respondent-level update is the maximum-a-posteriori (MAP) problem
\begin{equation*}
\begin{aligned}
  \hat{\bfbeta}_i
  =
  \operatorname*{arg\,max}_{\bfbeta \in \mathbb{R}^p}
  \Biggl\{
    \sum_{t=1}^{T_i}
    \Bigl[
      Y_{it}\log G(\bfDelta\bfX_{it}^\top \bfbeta)
      +
      (1-Y_{it})\log\!\bigl(1-G(\bfDelta\bfX_{it}^\top \bfbeta)\bigr)
    \Bigr]
    \\
    -
    \frac{1}{2}
    (\bfbeta-\hat f_{\text{ens}}(\bfZ_i))^\top
    \hat{\boldsymbol{\Sigma}}_\eta^{-1}
    (\bfbeta-\hat f_{\text{ens}}(\bfZ_i))
  \Biggr\}.
\end{aligned}
\end{equation*}
Equivalently, because the working covariance is diagonal, this update solves
\[
  \hat{\bfbeta}_i
  =
  \operatorname*{arg\,max}_{\bfbeta}
  \left\{
    \ell_i(\bfbeta)
    -
    \frac{c_\eta}{2}
    \sum_{k=1}^{p}
    \frac{
      \left(\beta_k-\hat f_{\text{ens},k}(\bfZ_i)\right)^2
    }{
      \hat\sigma_{\text{score},k}^2
    }
  \right\},
\]
where $\ell_i(\bfbeta)$ is the respondent-level log-likelihood above.  So $c_\eta$ is a ridge-like prior-precision multiplier---larger values shrink more toward the DNN mean---and the default $c_\eta=5$ favors stable recovery over fitting the few choices per respondent.  We call the resulting estimator the \textit{empirical-Bayes hybrid DNN}, labeled \texttt{EnsC5} in the software.  The calibration affects only the individual-level and distributional quantities based on $\hat{\bfbeta}_i$, not the population-average DML inference, which uses the orthogonal score and is independent of the MAP update.

For population-average parameters, cross-fitting still leaves a first-order bias from the estimation error of $\hat f$, which DML corrects via an orthogonal moment (the score uses the same two-seed ensemble).  For attribute $k$, respondent $i$, and task $t$, define the debiased signal:
\begin{equation}\label{eq:psi}
  \hat\psi_{ikt} = \hat{f}_k(\bfZ_i) + \left[\hat{\bfLambda}^{-1}(\bfZ_i) \cdot \bfDelta\bfX_{it} \cdot \left(Y_{it} - \hat{G}_{it}\right)\right]_k,
\end{equation}
where $\hat{G}_{it} = G(\bfDelta\bfX_{it}^\top \hat f(\bfZ_i))$ and $\hat{\bfLambda}(\bfZ_i)$ is the estimated local information matrix: the first term is the plug-in DNN mean, the second the orthogonal correction from the logit residual.  The debiased estimator of the average preference parameter $\theta_k=\E[f_k(\bfZ_i)]$ under the maintained structural model is then:
\begin{equation}\label{eq:theta_hat}
  \hat{\theta}_k
  =
  \frac{1}{N}
  \sum_{i=1}^N
  \frac{1}{T_i}
  \sum_{t=1}^{T_i}
  \hat\psi_{ikt}.
\end{equation}
When tasks per respondent are equal, this respondent-weighted estimator equals the task average.  Its mean-zero moment $\hat\varphi_{ikt}(\theta_k)=\hat\psi_{ikt}-\theta_k$ is Neyman-orthogonal, so first-stage errors affect $\hat{\theta}_k$ only at second order; clustering standard errors at the respondent level, the regularity conditions of \citet{farrell2025deep} and \citet{chernozhukov2018double} give $\sqrt{N}$-consistent, asymptotically normal inference.

\paragraph{Implications for the quantities of interest.}  The three stages serve different quantities.  Population averages---$\theta_k$, AMEs, counterfactual vote shares, the population MRS and WTP, and the between-respondent importance share---use only the first and third stages: the MAP update never enters, and precision is governed primarily by the number of respondents $N$.  Individual-level quantities---$\hat{\bfbeta}_i$ itself, respondent-specific ratios and rankings, and threshold conditions such as compensating differentials---rely on the second stage.  Their quality is governed by the tasks per respondent $T$ and by how much heterogeneity the covariates explain ($R_Z^2$); at low $T$ they are shrunk toward the mean stage, and their inference uses the respondent-cluster wild bootstrap.  Distributional summaries built from the $\hat{\bfbeta}_i$, such as polarization fractions, inherit both stages.  The result is a hierarchy: averages are easiest, respondent-level coefficients harder, and individual ratios hardest, because a near-zero denominator amplifies noise.  Supplementary Materials~\ref{sec:asymptotics} develops the corresponding asymptotic intuition.

Table~\ref{tab:estimands} summarizes the quantities of interest, their identifying information, and the type of inference they support. Population averages and smooth functionals of the mean stage are covered by the design-based DML guarantee. Population distributional summaries rely on the empirical-Bayes update and are model-based at fixed $T$. Individual, per-respondent quantities, including the recovered preference vector $\hat{\bfbeta}_i$ and its functionals such as individual MRS, become reliable as $T$ grows.
\begin{table}[t!]
\centering\small
\caption{Quantities of Interest, Source of Identification, and Inference}
\label{tab:estimands}\vspace{-0.5em}
\setlength{\tabcolsep}{5pt}
\begin{tabular}{@{}p{0.47\textwidth} p{0.27\textwidth} p{0.21\textwidth}@{}}
\toprule
Quantity & Identified by & Inference \\
\midrule
\multicolumn{3}{@{}l}{\textit{1. Mean-stage averages and smooth functionals (design-identified)}}\\[2pt]
Average preference $\theta_k = \E[f_k(\bfZ_i)]$ & mean stage $f(\bfZ)$ & DML, $\sqrt{N}$ CI \\
Average marginal effect & smooth functional of $f$ & DML, $\sqrt{N}$ CI \\
Counterfactual vote share (mean stage) & smooth functional of $f$ & DML, $\sqrt{N}$ CI \\
Between-respondent importance share & smooth functional of $f$ & DML, $\sqrt{N}$ CI \\
Population MRS, willingness to pay & smooth ratio of $\theta$ & DML/delta-method \\
\addlinespace
\multicolumn{3}{@{}l}{\textit{2. Population distributional summaries (empirical-Bayes, fixed $T$)}}\\[2pt]
Polarization fraction $\Pr(\beta_{ik} > 0)$ & residual law $F_{\eta \mid \bfZ}$ & cluster bootstrap$^{\dagger}$ \\
Compensating-differential fractions & residual law (threshold) & cluster bootstrap$^{\dagger}$ \\
\addlinespace
\multicolumn{3}{@{}l}{\textit{3. Individual, per-respondent quantities (reliable at large $T$)}}\\[2pt]
Recovered individual preferences $\hat{\bfbeta}_i$ & individual posterior mean & model-based$^{\dagger}$ \\
Individual MRS, within-respondent importance, rankings & individual $\hat{\bfbeta}_i$ & model-based$^{\dagger}$ \\
\bottomrule
\end{tabular}
\begin{minipage}[]{1\textwidth}\footnotesize
    \textit{Note:} $^{\dagger}$Quantities that depend on the residual law $F_{\eta\mid\bfZ}$ are not point-identified with a fixed number of tasks. The respondent-cluster wild bootstrap quantifies sampling variability, but fixed-$T$ shrinkage biases the plug-in estimates toward consensus. We therefore interpret these quantities as model-based descriptions and report the individual-level quantities as point summaries without confidence intervals.
\end{minipage}
\end{table}

\paragraph{Monte Carlo evidence.}  Three sets of checks verify this strategy before we take it to data.  First, the mean stage reproduces the reduced form it nests: on the \citet{bansak2016europeans} immigration conjoint it matches a pooled homogeneous logit almost exactly (Pearson correlation $r = 0.997$ across 28 levels; $r = 0.973$ across $15 \times 28$ country-specific coefficients; Supplementary Materials~\ref{app:validation}).

Second, coverage studies for the debiased averages compare the three first-stage learners (Supplementary Materials~\ref{app:learners} and~\ref{app:thetak_mc}).  With the DNN first stage, the intervals are close to nominal across the estimand suite at $N = 1{,}000$: $0.94$--$0.96$ for the average parameter, the AME, and counterfactual vote shares, with ratio quantities hardest ($0.85$ for the MRS).  The elastic-net and GRF first stages enter the same orthogonal score and are valid in the same way, but run a few points lower in finite samples.  The learner matters most when preferences are nonlinear in the moderators: an elastic net on raw moderators under-covers badly there ($0.68$), expanding its basis with splines restores nominal coverage ($0.97$), and the DNN and GRF are robust to both surfaces ($0.94$ and $0.89$).  Because calibration is not uniformly nominal across learners and designs, the DNN---most robust to the unknown shape of heterogeneity---is our default.

Third, a $1{,}872$-cell factorial grid varies $N$, $T$, $p$, covariate informativeness $R_Z^2$, and functional form ($18{,}720$ runs; Supplementary Materials~\ref{app:factorial}).  Population point recovery is strong throughout (mean $|\hat\theta - \theta| = 0.046$ on the logit scale, vote shares within $2.1$ points), and grid-wide coverage is $92\%$ against the $95\%$ target, the worst single cell ($71\%$) traced to mean-stage capacity in the most nonlinear designs and closed by widening the network.  Individual-level recovery is governed first by $R_Z^2$ (a $56\%$ ANOVA share) and is robust to functional form, and the MAP update improves it over the raw mean stage in $96\%$ of cells, supporting the \texttt{EnsC5} default.

\section{Applications}\label{sec:applications}

We illustrate the hybrid structural estimator with three published conjoint experiments, ordered from the most information-rich design to the sparsest.  The \citet{graham2020democracy} democracy conjoint ($T \approx 13$) presents the widest range of recoverable quantities; the \citet{ballard2017structure} tax-plan conjoint ($T = 8$) shows how continuous attributes yield individual-level rate schedules and the sharpest external validation; and the \citet{saha2022ambitious} candidate conjoint tests recovery in the sparsest design ($T = 3$).

\subsection{The Democracy Tradeoff}\label{sec:graham_svolik}

\citet{graham2020democracy} study how American voters trade off democratic principles against policy and partisan considerations.  Their conjoint pairs hypothetical state legislative candidates described by eight attribute groups: party (co-partisan vs.\ not), economic and social policy positions, seven good-governance positions, seven undemocratic actions, two valence violations, sex, race, and profession---$p = 30$ attribute levels, one good-governance position serving as the reference.\footnote{The original fields eight undemocratic wordings, two of which are protest-ban variants that \citet{graham2020democracy} pool into a single action; we follow their pooling.  The gerrymander action appears in two versions (a two-seat and a ten-seat variant).  Candidate age and political experience, also randomized in the design, are excluded here as in the original's main specification.}  The original fielded 1,691 respondents for 16 tasks; the public replication files release 13, with 1.2\% missing outcomes, and dropping respondents with missing covariates yields our sample of 1,605 respondents and 20,657 tasks.  We use the default implementation from Section~\ref{sec:estimation}; the original's estimates are survey-weighted, ours unweighted (reweighting moves the headline penalty by less than 0.1 points).  The covariate vector $\bfZ_i$ has 16 variables (ideology, party ID, Trump approval, age, education, income, authoritarianism, political knowledge, gender, race, and four policy ideal-point scales).  The original also collected each respondent's direct, pre-conjoint ratings of how undemocratic each practice is; we exclude these from $\bfZ_i$ so the recovered preferences are identified from choice behavior alone, reserving the ratings for out-of-sample validation (below).

Panel~A of Figure~\ref{fig:gs_amce_fraction} reports the average preference parameters $\hat{\theta}_k$ with their 95\% DML confidence intervals.  Co-partisanship has the largest positive effect ($\hat{\theta} = 0.72$, $p < 0.001$).  All seven undemocratic actions carry significant negative coefficients from $-0.57$ to $-0.93$, prosecuting journalists most opposed; the two valence violations are larger still---an affair ($-0.98$) and tax evasion ($-1.21$).  Clustered and unclustered standard errors are essentially identical (ratio $1.04$), as expected when profile contrasts are independently randomized across tasks, leaving little within-respondent correlation in the orthogonal scores; we read this as a consistency check, not evidence of correctness.

\citet{graham2020democracy} themselves go beyond the average penalty, specifying a random-utility logit (see Equation~6 and Table~2, ``Structural Estimates'') and summarizing support for democracy through its \emph{relative weight} ($\approx 18\%$, against $33\%$ for social policy, $27\%$ for economic, and $22\%$ for party) and an implied \emph{marginal rate of substitution}---a value for democracy of $\approx 0.82$ in co-partisanship units, below one.  Our contribution is to \emph{individualize} this account: in place of a single weight and price, we recover the full distribution of $\hat{\bfbeta}_i$---importance shares, marginal rates of substitution, compensating differentials, and counterfactual vote shares---with a flexible first stage and debiased inference for the averages.

A central \citet{graham2020democracy} finding is that an undemocratic position costs a candidate only about 11.7\% of vote share on average---roughly one in nine voters punishes an otherwise-preferred candidate for violating democratic principles.  Our model reframes this finding: Panel~B of Figure~\ref{fig:gs_amce_fraction} reports, for each attribute level, the fraction of respondents whose recovered $\hat{\beta}_{ik}$ is positive versus negative, and for every undemocratic action only a small minority have $\hat{\beta}_{ik} > 0$---at most about $6\%$ overall for the ten-seat gerrymander, a tail estimate that is less precisely identified than the average, and effectively zero for the most-opposed actions.  Rejection of undemocratic behavior is broad, if not universal.

These two facts are complementary, not contradictory.  The 11.7\% penalty---the share who defect to punish a violation---measures how often democratic preferences prove decisive against competing considerations, not how widely they are held.  Almost all voters dislike undemocratic behavior (direction), but many weight party and policy more heavily (magnitude): the disagreement is not whether democracy matters but how much.

\begin{figure}[t!]
  \centering
  \caption{Average Preference Estimates and Favor--Oppose Fractions:\\ Graham \& Svolik (2020)}
  \label{fig:gs_amce_fraction}
  \includegraphics[width=\textwidth]{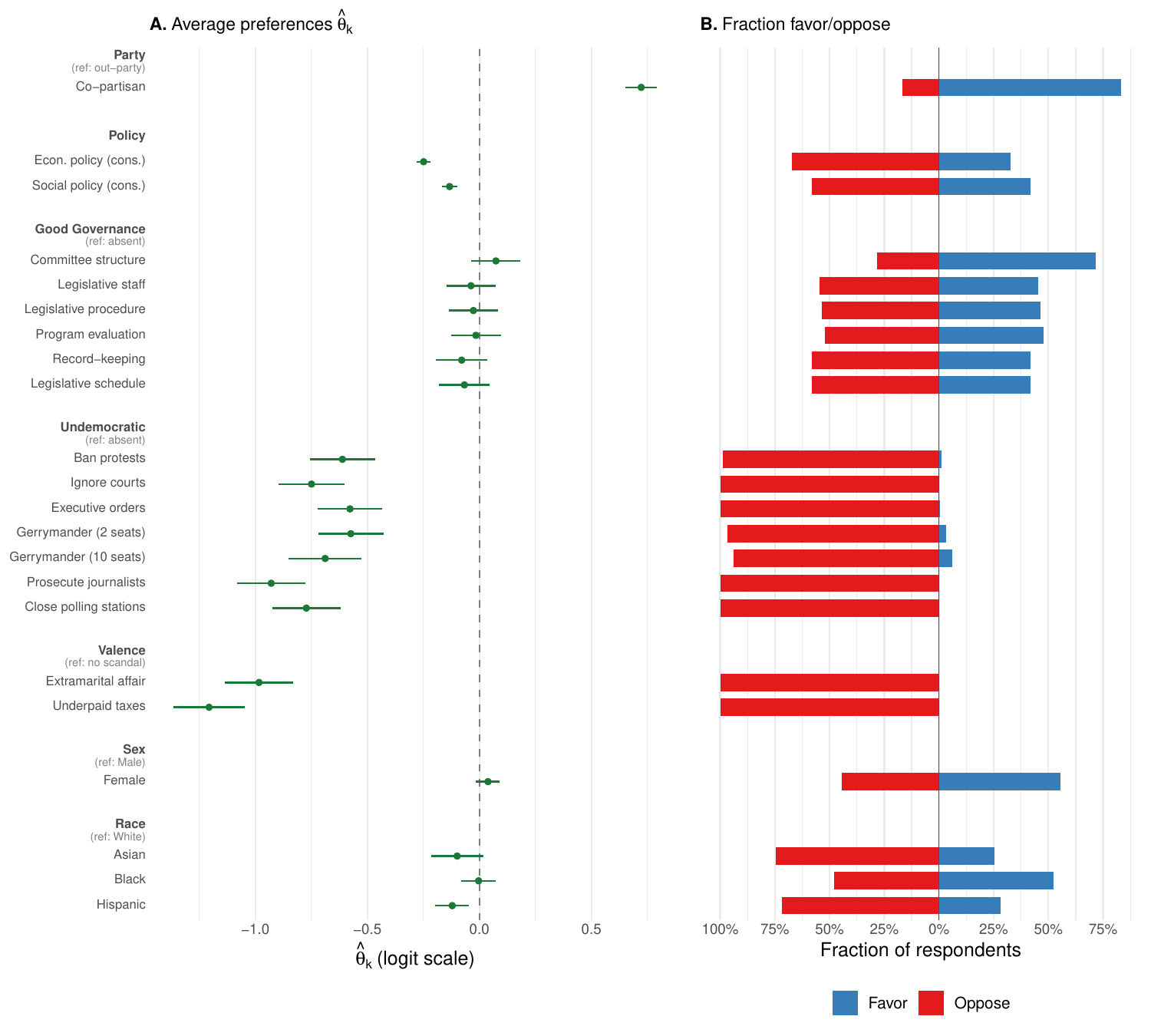}
  \begin{minipage}[]{1\textwidth}\footnotesize
      \textit{Note:} Profession dummies are omitted. \textbf{A}: $\hat{\theta}_k$ with 95\% DML CIs.  \textbf{B}: Fraction favoring vs.\ opposing each level; all undemocratic actions are opposed by $>$93\%.
  \end{minipage}
\end{figure}

\begin{figure}[p]
  \centering
  \caption{Respondent-level Preferences and the Misspecification Diagnostic:\\Graham \& Svolik (2020)}
  \label{fig:gs_beta_ridgelines}
  \includegraphics[width=\textwidth]{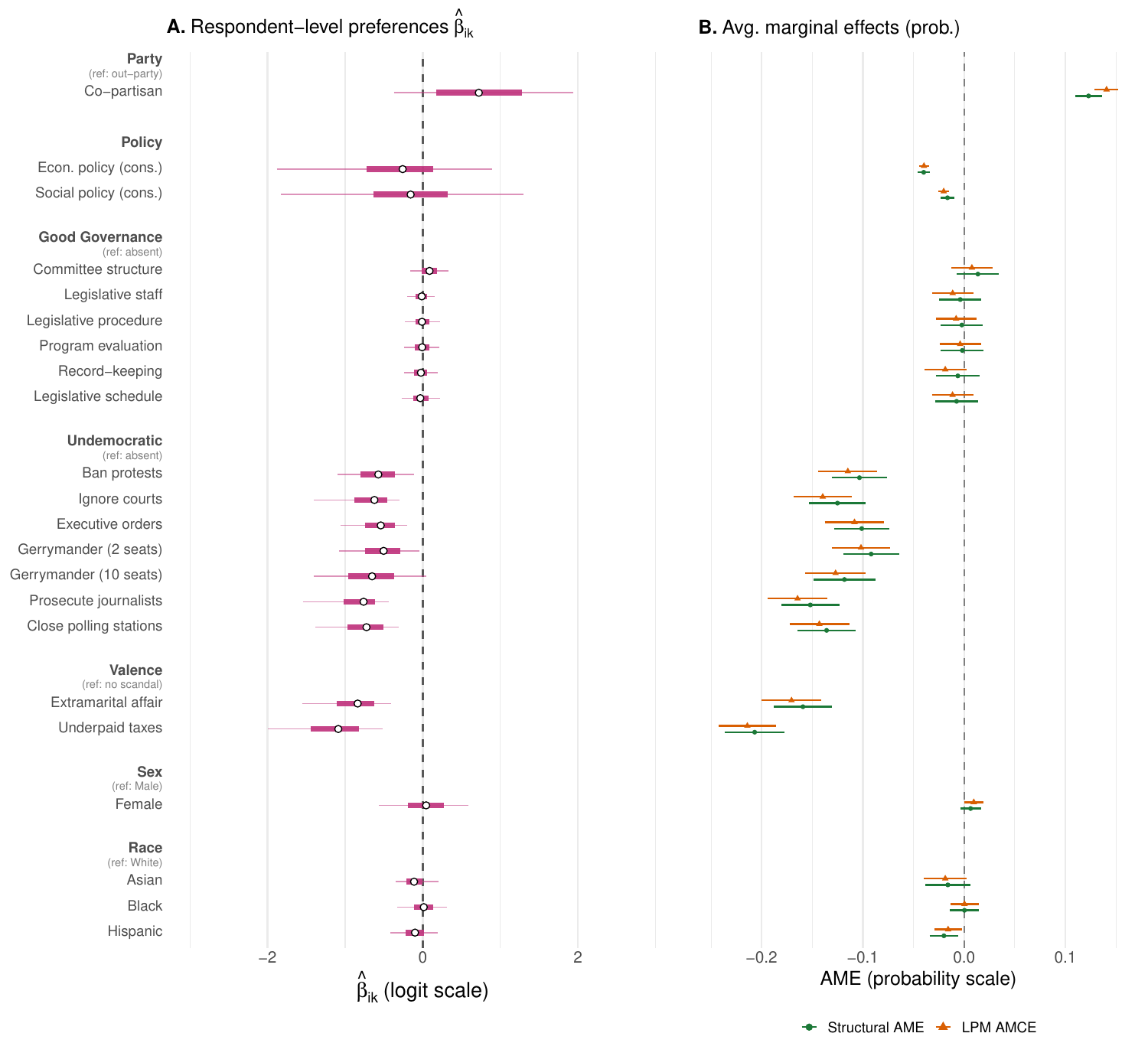}
  \begin{minipage}[]{1\textwidth}\footnotesize
      \textit{Note:} Profession dummies are omitted.  \textbf{A}: $\hat{\beta}_{ik}$ distributions by attribute family; points mark medians, thick bars the IQR, thin bars the 5th--95th percentiles.  Undemocratic actions lie almost entirely below zero; spread reflects intensity heterogeneity.  \textbf{B}: structural AME (probability scale) vs.\ the linear-probability AMCE, with 95\% CIs.  The two series agree in sign on all 30 levels and differ by 0.6 percentage points on average.
  \end{minipage}
\end{figure}

Panel~A of Figure~\ref{fig:gs_beta_ridgelines} displays the $\hat{\beta}_{ik}$ distributions by attribute level: good-governance estimates concentrate near zero, the undemocratic actions sit in negative territory with substantial spread---virtually everyone opposes them, but intensity varies.  Average sensitivity is fairly uniform across ideology (the correlation between mean $|\hat{\beta}_{i,u}|$ and 7-point ideology is weak, $r = -0.20$); the ideological disagreement below is about the weight placed on competing considerations, not the intensity of opposition to erosion.

Panel~B provides a first validation of these estimates: the misspecification diagnostic of Section~\ref{sec:quantities}, comparing the structural AME on the probability scale with the design-based AMCE from a linear probability model with respondent-clustered standard errors.  In this information-rich design the two series agree closely: correlation $0.998$ across the 30 levels, signs coincide everywhere, mean absolute difference $0.6$ points (at most $1.8$, on co-partisanship), all 95\% CIs overlap.  With $T \approx 13$, the structural model reproduces the reduced-form benchmark almost exactly; the quantities that follow are what it adds.

How much weight do voters actually place on each attribute group when deciding between candidates?  We answer this with a variance decomposition: for each respondent~$i$ and attribute group $g$, the importance share $\text{Imp}_{i,g} = \mathrm{Var}_X\!\big(\sum_{k \in g} \hat{\beta}_{i,k} X_k\big) \big/ \sum_{g'} \mathrm{Var}_X\!\big(\sum_{k \in g'} \hat{\beta}_{i,k} X_k\big)$ is the share of utility variance the group accounts for, with the variance taken over the design distribution of attribute levels so that within-attribute level covariances enter exactly rather than through a sum-of-squares approximation.  Figure~\ref{fig:gs_importance} shows the distribution of these shares across respondents.  Averaged across all respondents, valence violations account for $29\%$ of utility variance, policy for $26\%$, party for $23\%$, the seven undemocratic actions together for $9\%$, and sex/race/profession/good-governance together for about $12\%$.  So while voters virtually unanimously oppose undemocratic behavior in \emph{direction}, it commands only $\sim$$9\%$ of variance---comparable to sex, race, and profession combined.  Party, policy, and valence each command between roughly a fifth and a third of the variance, which is consistent with \citeauthor{graham2020democracy}'s finding that voters tolerate undemocratic actions from co-partisans: democratic norms are a real but distinctly secondary constraint.

\begin{figure}[t!]
  \centering
  \caption{Distribution of Individual-level Attribute-importance Shares:\\Graham \& Svolik (2020)}
  \label{fig:gs_importance}
  \includegraphics[width=0.85\textwidth]{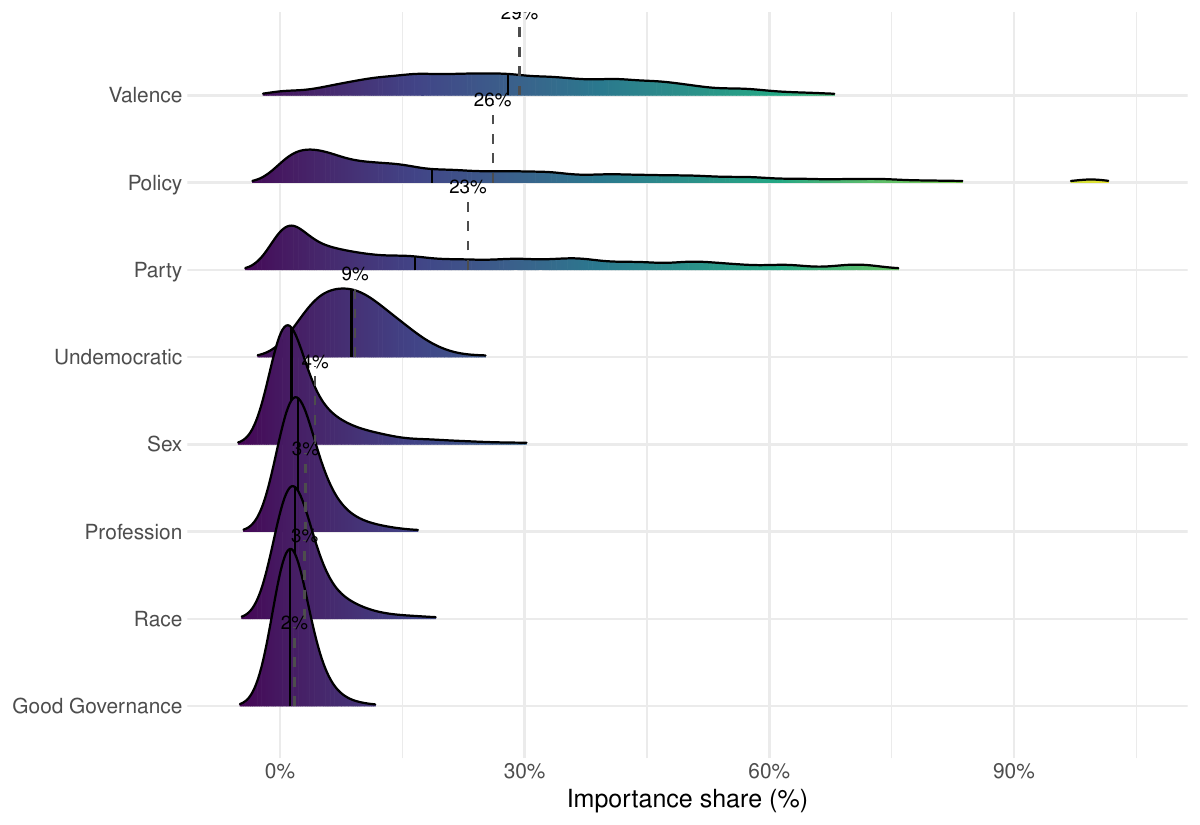}
  \begin{minipage}[]{1\textwidth}\footnotesize
      \textit{Note:} Dashed lines mark the labeled means.  Valence ($\sim$29\%), policy ($\sim$26\%), and party ($\sim$23\%) are roughly co-equal at the top; the seven undemocratic actions together account for $\sim$9\%.
  \end{minipage}
\end{figure}

How does a violation weigh against the pull of co-partisanship?  The population marginal rate of substitution, $-\hat{\theta}_{u}/\hat{\theta}_{\text{party}}$, is the share of the co-partisan benefit that just offsets a violation.  For prosecuting journalists it is $1.29$ ($95\%$ CI $[1.04,\,1.53]$), above the $\approx 0.82$ implied by \citeauthor{graham2020democracy}'s pooled structural estimate; the interval lies entirely above one, so co-partisanship alone does not, on average, compensate for it.  Ignoring courts ($1.04$, CI $[0.81,\,1.27]$) and the ten-seat gerrymander ($0.95$, CI $[0.71,\,1.20]$) sit near one, with the average voter roughly breaking even.  But these population prices average over very different voters.

A more discrete question---closer to \citeauthor{graham2020democracy}'s headline defection rates---asks whether a compensator is \textit{sufficient} for a voter to accept an undemocratic action.  For each voter $i$ and action $j$ we compute $\mathbf{1}\{\hat{\beta}_{i,u_j} + \text{benefit}_{ik} \geq 0\}$ under three compensators: co-partisanship ($+\hat{\beta}_{i,\text{party}}$), a full-range policy swing ($3(|\hat{\beta}_{i,p_1}| + |\hat{\beta}_{i,p_2}|)$, where three units span each policy scale's coded range), and the voter's favorite good-governance feature ($\max_k \hat{\beta}_{i,g_k}$).  Figure~\ref{fig:gs_compdiff} reports the fraction for whom compensation holds, by ideology tercile.

\begin{figure}[t!]
  \centering
  \caption{Fraction Accepting Undemocratic Actions under Compensating Benefits:\\Graham \& Svolik (2020)}
  \label{fig:gs_compdiff}
  \includegraphics[width=\textwidth]{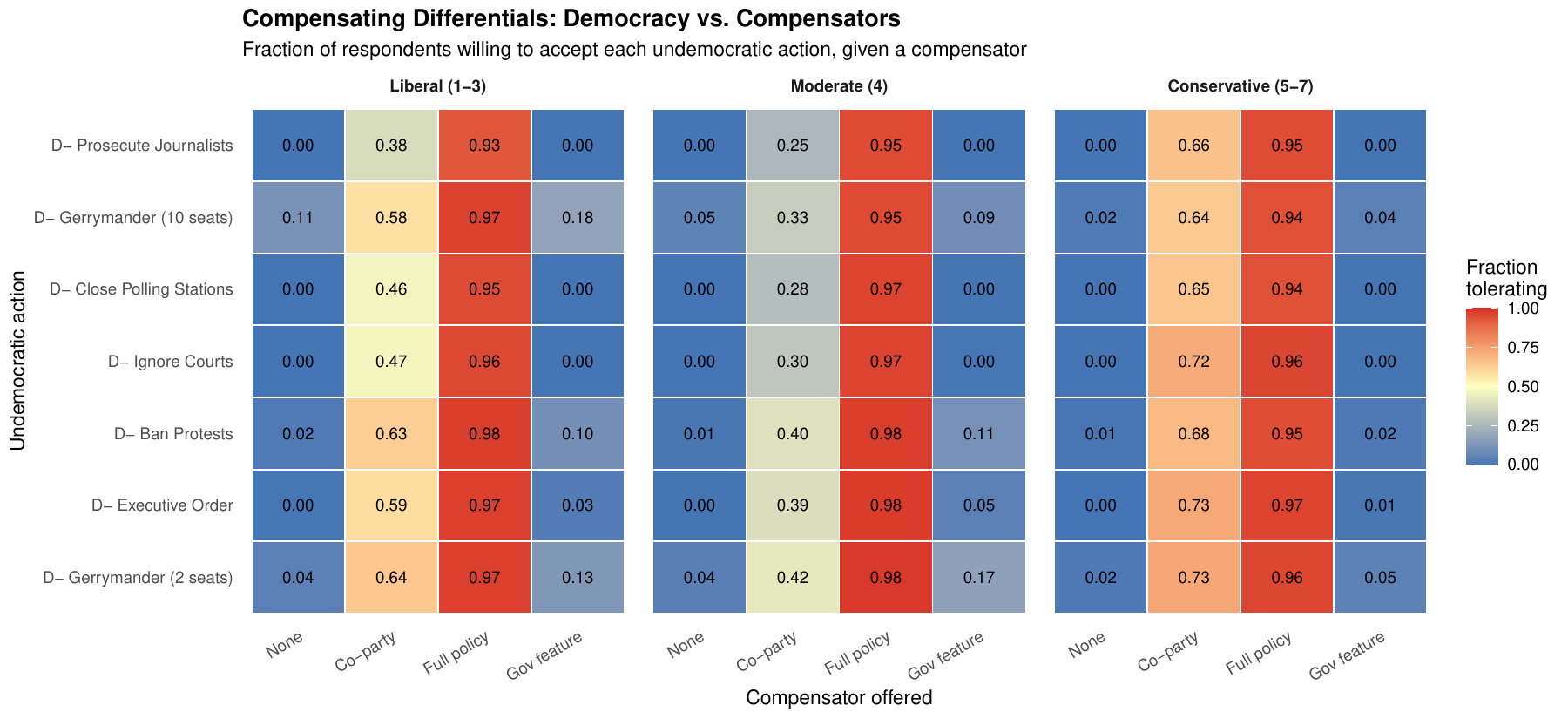}
  \begin{minipage}[]{1\textwidth}\footnotesize
      \textit{Note:} Each cell is the fraction of respondents (by ideology tercile) with $\hat{\beta}_{i,u_j} + \text{benefit}_{ik} \geq 0$---the share for whom the compensator induces acceptance of the undemocratic action.  Rows sorted by severity; the ``None'' column is the baseline fraction with $\hat{\beta}_{i,u_j} \geq 0$.
  \end{minipage}
\end{figure}

Three features stand out.  First, the ``None'' column confirms near-universal resistance: the share with a positive coefficient on any action is essentially zero, rising above a few percent only for gerrymandering ($\sim$$11\%$ of liberals for the ten-seat version).  Second, co-partisanship alone compensates a majority on most actions (ban-protests, ignore-courts, executive order, gerrymandering).  Third, it reveals a clear ideological gradient on the hardest cases: for prosecuting journalists $38\%$ (CI $[34,42]$) of liberals would accept the violation for a co-partisan versus $66\%$ ($[62,70]$) of conservatives---a 28-point gap excluding zero---with closing polling stations a 19-point gap ($46$ vs.\ $65$) and ignoring courts 25 ($47$ vs.\ $72$).  These three---tied to free elections, a free press, and judicial independence---are where liberals most cross party lines to defend democracy.  A full-range policy swing, by contrast, compensates over $94\%$ for any violation, underscoring why policy alignment dominates democratic principles in observational voting.

The compensating-differentials analysis above operates at the individual level: each voter accepts or rejects the candidate depending on whether the compensator clears the threshold.  A complementary structural counterfactual works at the aggregate level, computing predicted support in any two-candidate contest the model can simulate.  Figure~\ref{fig:gs_moderation} reports three such contests, each isolating one strategic dimension with the rest of the bundle fixed; in every panel Candidate A is the respondent's co-partisan and B the opposing party at the design's standard-conservative positions, so A begins with co-partisanship's $+0.72$ logit advantage.\footnote{Here and in the later applications, the displayed contest probabilities aggregate the recovered $\hat{\bfbeta}_i$, so they reflect both covariate-driven and residual heterogeneity---a distributional summary in the hierarchy of Section~\ref{sec:quantities}.  The debiased observations of Supplementary Materials~\ref{app:other_estimands} target the mean-stage population share $\E[G(\bfDelta\bfX^\top f(\bfZ))]$, which differs from these aggregates by roughly one to five percentage points on the candidate contests and by more on the most lopsided tax-plan bundles.}

\begin{figure}[t!]
  \centering
  \caption{Co-partisan Win Probability under Electoral-competition\\ Counterfactuals: Graham \& Svolik (2020)}
  \label{fig:gs_moderation}
  \includegraphics[width=\textwidth, trim={2cm 0.5cm 2cm 0cm}, clip]{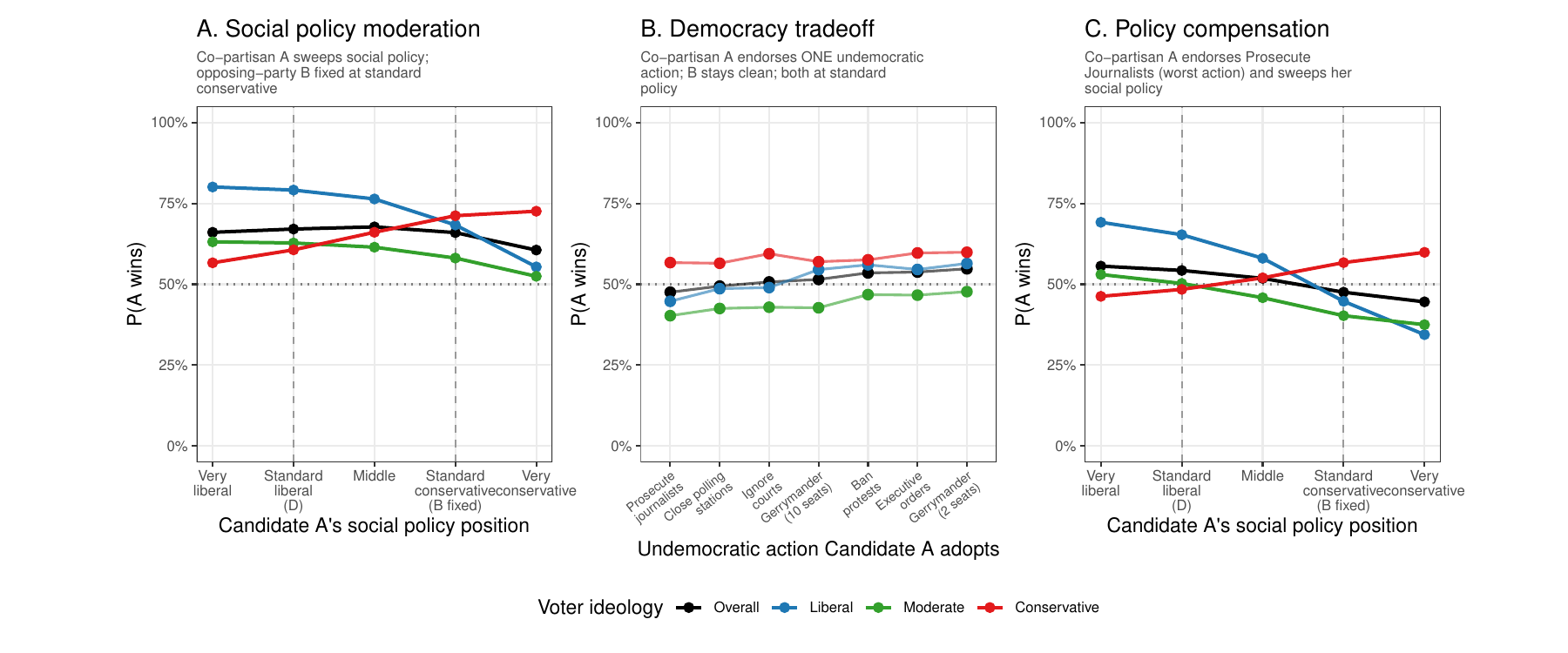}
  \begin{minipage}[]{1\textwidth}\footnotesize
      \textit{Note:} In every panel, Candidate A is the respondent's co-partisan and B the opposing party at standard-conservative positions.  \textbf{A}: A sweeps social policy from very liberal to very conservative, B fixed.  \textbf{B}: A endorses one of the seven undemocratic actions, B stays clean.  \textbf{C}: A endorses Prosecute Journalists and sweeps social policy.  $\Pr(A\text{ wins})$ aggregates the recovered $\hat{\bfbeta}_i$ within ideology subgroups (1--3, 4, 5--7 on the 7-point scale).
  \end{minipage}
\end{figure}

Panel A sweeps A's social policy: Liberals reward staying liberal ($80\%$ at very liberal, $55\%$ at very conservative), Conservatives move oppositely ($57\%$ to $73\%$), and the lines cross near the middle.  Panel B holds policy fixed while A endorses one undemocratic action; most leave A a small overall majority, because the co-partisan benefit absorbs the cost, but Prosecute Journalists is worst (overall $48\%$, Liberal $45\%$).  Panel C combines both: co-partisan A endorses Prosecute Journalists and sweeps social policy.  Liberal support runs $69\%$ (very liberal) to $34\%$ (very conservative); Conservatives mirror it, reaching $60\%$ at very-conservative policy.  The Overall line stays between $45\%$ and $56\%$, masking the polarization---only near the middle does A hold a slim majority in both camps at once; away from it, moving on social policy trades one base for the other one-for-one.

\textbf{External validation.} Are these recovered preferences real or model artifacts?  \citet{graham2020democracy}'s first-wave survey separately asked each respondent to rate how undemocratic each practice is.  Because these ratings are excluded from $\bfZ_i$, they are an out-of-sample check: we correlate each recovered $\hat{\beta}_{i,k}$ against the respondent's own direct rating of that practice.  Pooled, the rank correlation is $0.37$ (Spearman, $n = 11{,}229$ respondent--action pairs), per-action $0.32$ (executive order) to $0.47$ (closing polling stations) and positive throughout---convergent validity against a measure the model never saw, which reduced-form AMCEs cannot provide.  It is strongest among liberals (prosecuting journalists, $0.60$) and weakest among conservatives ($0.20$), consistent with conservatives' more uniformly low direct ratings leaving less individual variation to recover.  Adding the items back to $\bfZ_i$ leaves the averages essentially unchanged (max $|\Delta\hat{\theta}_k| = 0.05$; Supplementary Materials~\ref{app:gs_covariates}), paralleling \citeauthor{graham2020democracy}'s own finding with respondents' severity ratings (their Table~2, columns 5--6).  Figure~\ref{fig:gs_validation} in Supplementary Materials~\ref{app:more_figures} plots the relationship by action and ideology.

\FloatBarrier
\subsection{The Structure of Tax Policy Preferences}\label{sec:ballard_rosa}

\citet{ballard2017structure} study American mass preferences over federal income-tax policy, asking 2{,}000 US adults to choose between pairs of hypothetical tax plans in which each plan specifies the marginal rate on six income brackets together with a five-level revenue indicator.  Their central reduced-form finding is that Americans have generally progressive preferences---opposing higher rates on the poor and supporting higher rates on the rich---but that support for a plan responds much more elastically to rates on the poor than to rates on the rich.  The seven attributes are all continuous: marginal rates on six brackets (in percentage points) and a revenue indicator rescaled to $\{-2,-1,0,1,2\}$.\footnote{The distributed replication file's derived rate variable for the \$175--375k bracket miscodes its 45\% level as 5; we rebuild all six bracket rates from the underlying coded variables and value labels, leaving the original's coded-level analyses unaffected.  As in the other applications, the original's published estimates are survey-weighted while ours are unweighted and describe the analysis sample.}  The hybrid estimator therefore recovers an individual-level \textit{slope} per bracket and respondent, producing a full preferred tax schedule $\hat{\bfbeta}_i \in \mathbb{R}^7$ per person rather than a handful of level effects.  The covariate vector $\bfZ_i$ contains 23 respondent characteristics including age, gender, party ID, education, race, income, ideology, and a battery of attitudinal and economic-belief measures.  Design dimensions: $N = 2{,}000$, $T = 8$, $p = 7$, $|\bfZ| = 23$, $NT = 16{,}000$.

Figure~\ref{fig:br_amce_fraction} reports the DML-corrected average effects.  A one-percentage-point increase in the bottom-bracket rate lowers the predicted probability of plan support by about 1.6 points (at a 50\% baseline); the same increase on the top bracket raises support by only about 0.5 points.  This more than threefold absolute asymmetry is the quantitative form of the central \citeauthor{ballard2017structure} finding.  The \$85--175k bracket is indistinguishable from zero on average---the ``dead zone''---and the revenue indicator carries a clear positive effect ($\hat{\theta}=0.14$, $p<0.01$): respondents penalize plans that reduce federal revenue and reward plans that raise it.  Panel~B reports the fraction favoring versus opposing a rate increase on each bracket, as in the democracy application.

\begin{figure}[t!]
  \centering
  \caption{Average Preference Estimates and Favor--Oppose Fractions:\\ Ballard-Rosa et al. (2017)}
  \label{fig:br_amce_fraction}
  \includegraphics[width=\textwidth]{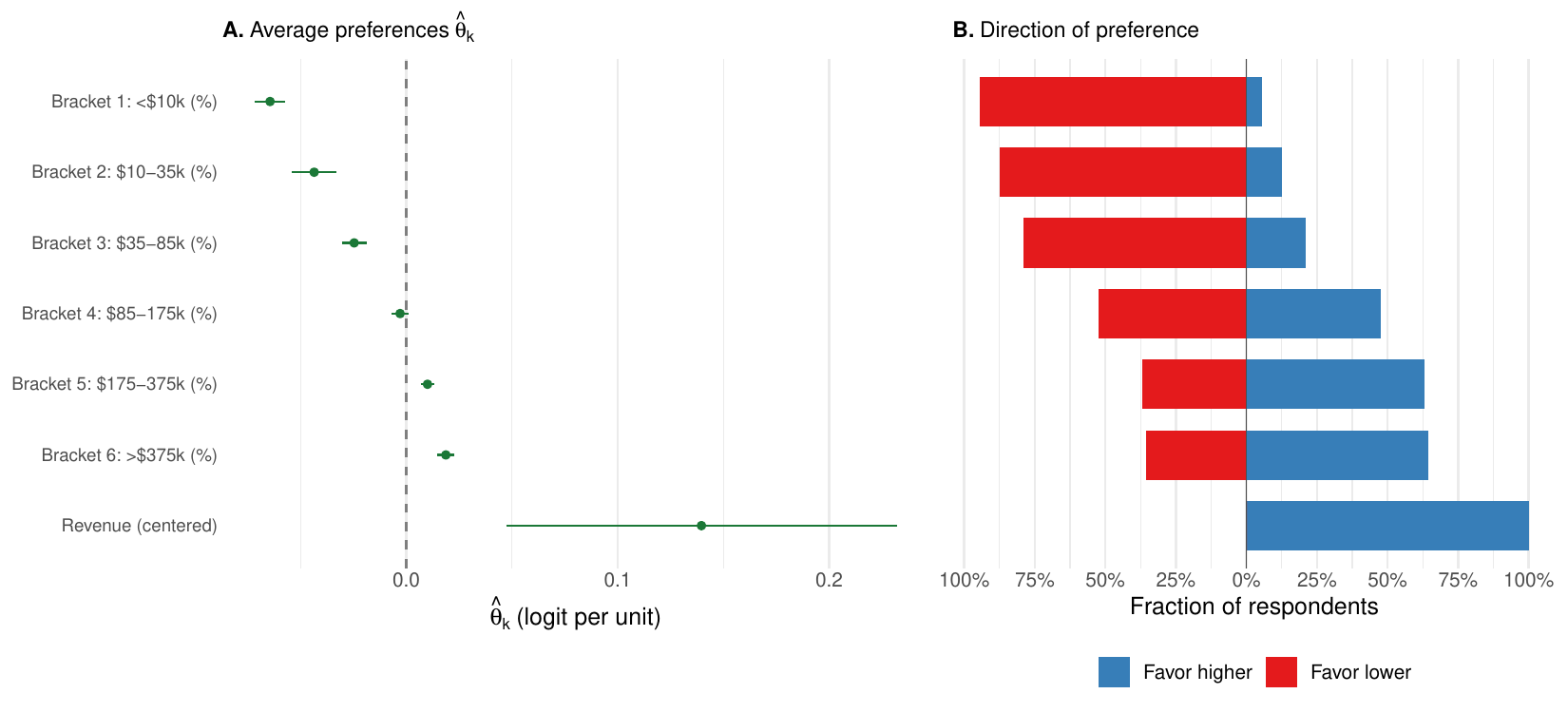}
  \begin{minipage}[]{1\textwidth}\footnotesize
      \textit{Note:} \textbf{A}: $\hat{\theta}_k$ on the logit scale (per percentage point of rate, or per unit of revenue).  \textbf{B}: Fraction with $\hat{\beta}_{i,k} > 0$ (favor raising) vs.\ $< 0$ (oppose).  The \$85--175k bracket is the most polarized---roughly half each way---despite a near-zero average.
  \end{minipage}
\end{figure}

Because each respondent has a full 7-element vector $\hat{\bfbeta}_i$, we can compute an \textit{individual-level progressivity slope}---the within-respondent regression of bracket-specific marginal utilities on the log of the bracket midpoint:
\[
s_i \;=\; \frac{\sum_{k=1}^{6}(\log m_k - \overline{\log m}) \cdot \hat{\beta}_{i,k}}{\sum_{k=1}^{6}(\log m_k - \overline{\log m})^2}.
\]
A positive $s_i$ means respondent~$i$ prefers higher rates on higher incomes; zero means a flat tax; negative means regressive.  In this sample $93\%$ of respondents have a positive slope, and $\hat{\beta}_{i,\text{top}} > \hat{\beta}_{i,\text{bottom}}$ for $91\%$ (distributional summaries of the recovered individual slopes, model-based at $T = 8$).  The progressive direction is dominant: Democrats' mean slope is $+0.0210$, Republicans' is $+0.0142$---a 48\% gap---and $90\%$ of Republican respondents still reveal progressive preferences.  This is a strong individual-level form of the \citeauthor{ballard2017structure} finding that progressive preferences cut across partisan lines, while also revealing a small minority (about $7\%$) with flat or weakly regressive revealed preferences.

\begin{figure}[!ht]
  \caption{Individual-level $\hat{\bfbeta}_i$ Tax Schedules, by Party:\\ Ballard-Rosa et al. (2017)}
  \label{fig:br_schedule_by_party}
  \centering
  \includegraphics[width=0.8\textwidth]{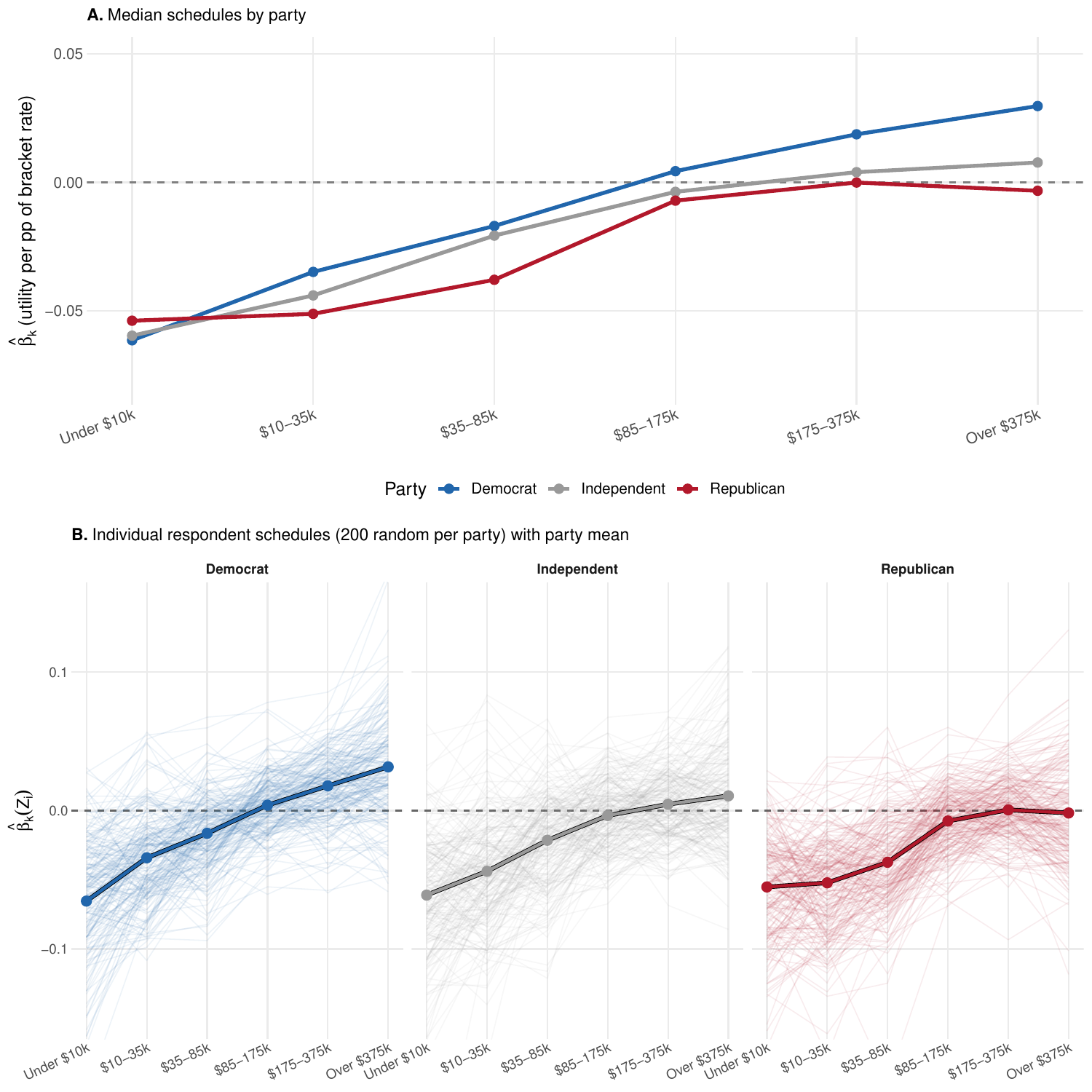}
  \begin{minipage}[]{1\textwidth}\footnotesize
        \textit{Note:} \textbf{A}: Median schedule by party.  \textbf{B}: 200 random individual schedules per party (semi-transparent), party mean in bold.  Almost every line slopes upward, but levels vary widely---within-party heterogeneity far exceeds the between-party gap.
    \end{minipage}
\end{figure}

Figure~\ref{fig:br_schedule_by_party} exposes the within-party heterogeneity the model recovers.  The party-level median schedules (Panel~A) are indistinguishable at the bottom brackets but fan out at the top: Democrats prefer higher top-bracket rates (where Republicans are essentially flat) and lower rates on the lower-middle brackets, while Republicans want smaller gaps between bottom and top---their opposition to higher rates is more uniform across the income distribution.  But the 200 individual schedules per party (Panel~B) almost all slope upward with widely varying levels---some Republicans steeper than the Democratic median, some Democrats nearly as flat as the Republican median.  Within-party heterogeneity far exceeds the between-party gap: a shift in central tendency, not a separation of populations.

The structural model also enables a variance decomposition by subgroup---the importance shares of Section~\ref{sec:graham_svolik}, here computed by party---that is invisible to reduced-form subgroup AMCEs.  Figure~\ref{fig:br_importance_party} (Supplementary Materials~\ref{app:more_figures}) reports each attribute's share of plan-choice variance: Democrats' choices are driven most by the top bracket ($29.2\%$) and the very bottom bracket, while Republicans weight the working- and middle-class brackets most heavily and place the least weight on the top bracket ($18.4\%$).  This reframes the usual partisan account: Republican opposition to progressive taxation reflects substantial weight on rates for the working and middle class, not solely protection of the rich.

Alongside the conjoint, \citet{ballard2017structure} collected self-reported ideal marginal tax rates---every respondent for the top bracket, each for one randomly assigned lower bracket---elicited outside the conjoint and never used in training, so they provide an independent benchmark.  No reduced-form AMCE estimator can be validated at the individual level, because it does not produce an individual-level preference parameter.  Among the 401 respondents answering both the bottom- and top-bracket questions, the revealed progressivity slope $s_i$ correlates $r = 0.43$ with its self-reported analog and the revealed top-minus-bottom gap $r = 0.47$; across all $N = 2{,}000$, the individual top-bracket coefficient correlates $r = 0.42$ with the self-reported ideal top rate, positive within every partisan subgroup (Democrats $0.29$, Independents $0.36$, Republicans $0.39$)---the model recovers progressivity within, not just across, parties.  An $r = 0.42$ between a recovered parameter and an independent self-report---comparable to test--retest reliabilities for such items \citep{ansolabehere2008strength}---demonstrates the value of the empirical-Bayes update with only 8 tasks per respondent.

\begin{figure}[t!]
  \centering
  \caption{Counterfactual Evaluation of Four Stylized Tax Plans from  Recovered~$\hat{\bfbeta}_i$: Ballard-Rosa et al. (2017)}
  \label{fig:br_plans_counterfactual}
  \includegraphics[width=\textwidth]{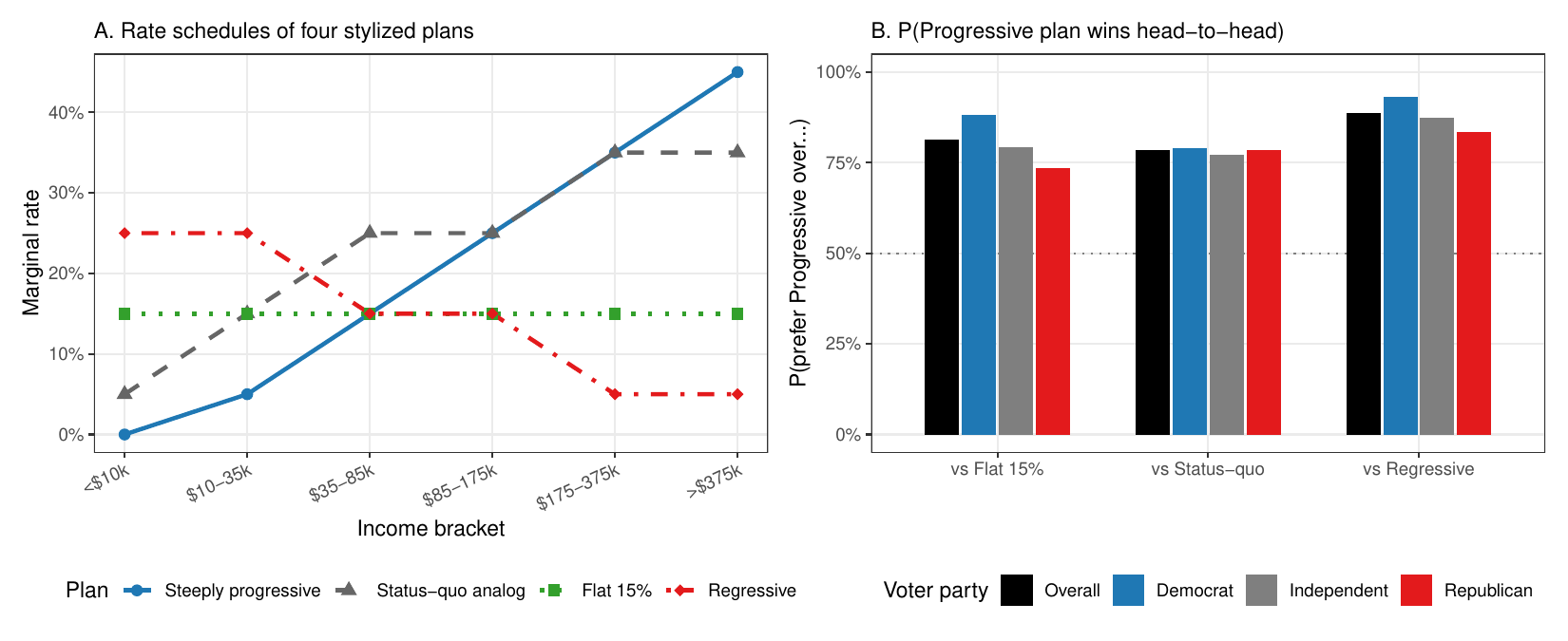}
  \begin{minipage}[]{1\textwidth}\footnotesize
      \textit{Note:} \textbf{A}: Marginal-rate schedules across the six brackets.  \textbf{B}: Predicted probability the progressive plan wins each matchup, by party.  All bars sit well above 50\%---Republicans pick the progressive over the flat plan $74\%$ of the time---so progressive preferences cut across party lines.
  \end{minipage}
\end{figure}

The individual-level preference vector $\hat{\bfbeta}_i$ also enables evaluation of any hypothetical tax schedule.  Figure~\ref{fig:br_plans_counterfactual} compares four stylized plans (Panel A): a steeply progressive plan with rates rising $0/5/15/25/35/45$ and ``much more'' revenue, a revenue-neutral status-quo analog with rates $5/15/25/25/35/35$, a revenue-neutral flat 15\% tax, and a regressive plan with rates falling $25/25/15/15/5/5$ and ``much less'' revenue.  The progressive plan yields the only positive mean systematic utility in the set.  Panel B shows head-to-head probabilistic comparisons by party.  The progressive plan wins every matchup---$81\%$ against the flat plan ($74\%$ among Republicans), $78\%$ against the status quo, $89\%$ against the regressive plan---and delivers higher systematic utility than the flat plan for $89\%$ of respondents.  Americans in this sample prefer the status quo to a flat or regressive plan, but they prefer an explicitly more progressive alternative to the status quo even more strongly.  The \citeauthor{ballard2017structure} conclusion that Americans' preferences lie ``quite close to existing policy'' is consistent with the status-quo plan ranking second-best, but it understates how much the public would prefer a more steeply progressive schedule if offered one.  A purely partisan model---``Republicans want flat, Democrats want progressive''---fits this sample poorly: $83\%$ of Republicans' recovered preferences favor a steeply progressive, higher-revenue schedule over a flat, revenue-neutral one.

\FloatBarrier
\subsection{Political Candidate Preferences}\label{sec:saha_weeks}

\citet{saha2022ambitious} study how voters evaluate candidates for public office---with attention to how gender interacts with perceived ambition---across a series of conjoint experiments.\footnote{The core finding of \citet{saha2022ambitious} is about voter taste for \emph{ambitious women} specifically, not gender in isolation; \citet{teele2018ties} provide a closer reference point for the unconditional effect of candidate gender.}  Each respondent evaluated three pairs of hypothetical candidates described by five attributes: policy agenda (3 levels), talent (7 levels), number of children (4 levels), gender (2 levels), and progressive ambition---whether the candidate has shown interest in running for higher office (2 levels), yielding $p = 13$ dummy-coded attribute levels.  We again use the default implementation from Section~\ref{sec:estimation}.  The covariate vector $\bfZ_i$ includes 19 respondent characteristics (party identification, ideology, gender, age, education, income, employment status, region, 2016 vote choice, and gender attitudes).  We use their main U.S.\ survey, originally fielded by Survey Sampling International (SSI) to 1,249 respondents; after dropping incomplete cases, this is the sparsest of our three applications: $N = 1{,}191$, $T = 3$, $p = 13$, $|\bfZ| = 19$, and $NT = 3{,}573$.

\begin{figure}[t!]
  \centering
  \caption{Average Preference Estimates and Favor--Oppose Fractions:\\Saha \& Weeks (2022)}
  \label{fig:sw_amce_fraction}
  \includegraphics[width=\textwidth]{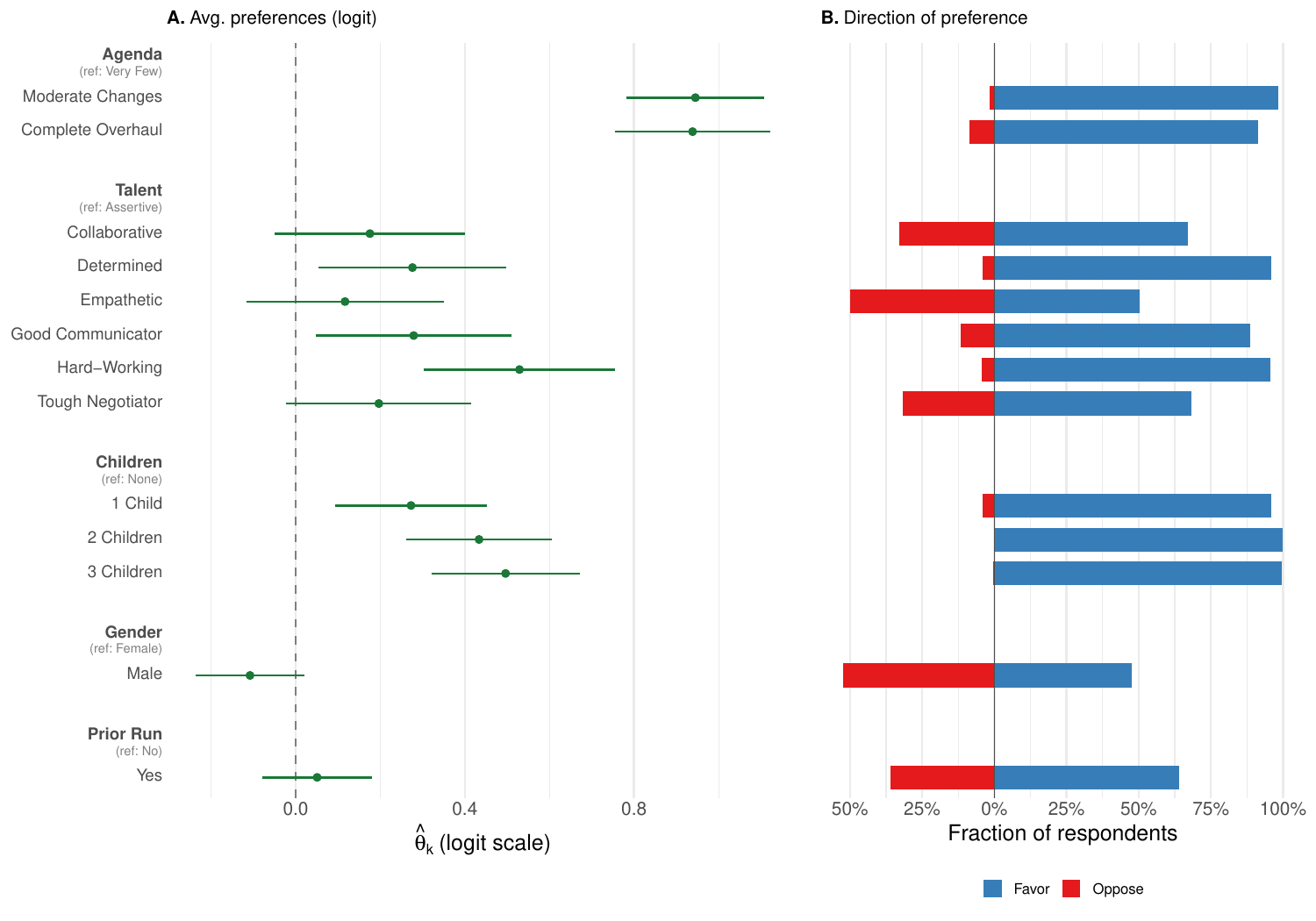}
  \begin{minipage}[]{1\textwidth}\footnotesize
      \textit{Note:} \textbf{A}: Average preference parameters $\hat{\theta}_k$ (logit scale, 95\% DML CIs).  \textbf{B}: Fraction favoring (blue) vs.\ opposing (red) each level, by the sign of $\hat{\beta}_{ik}$.  Gender is near-zero on average but split across the population, masking the partisan polarization of Figure~\ref{fig:sw_partisan_gender}.
  \end{minipage}
\end{figure}

Panel~A of Figure~\ref{fig:sw_amce_fraction} reports the DML-corrected average estimates $\hat{\theta}_k$ on the logit scale.  Policy agenda dominates---Moderate Changes ($\hat{\theta} = 0.95$) and Complete Overhaul ($\hat{\theta} = 0.94$) are both strongly preferred to Very Few Changes ($p < 0.001$)---followed by talent (Hard-Working $0.53$, Good Communicator and Determined to Succeed $0.28$).  The gender coefficient (Male relative to Female) is negative but insignificant ($\hat{\theta} = -0.11$, $[-0.24, 0.02]$), consistent with the original study's slight edge for women candidates.  Clustered and unclustered standard errors again essentially match (ratio $1.03$).

Panel~B reports the fraction favoring versus opposing each level, as in the first two applications; Figure~\ref{fig:sw_beta_ridgelines} (Supplementary Materials~\ref{app:more_figures}) shows the full $\hat{\beta}_{ik}$ densities for all 13 levels, ordered by variance.

Two patterns stand out.  First, most attributes are polarized in \emph{intensity} but not direction---nearly all respondents favor Moderate Changes and Complete Overhaul, with intensities ranging up to over $1.5$ logit units---so voters agree on the sign and the average $\hat{\theta}_k$ tracks individual preferences well.  Second, gender is the exception: its density is wide and centered near zero, the only attribute split in \emph{direction}, so the near-zero AMCE reflects cancellation rather than indifference.  What predicts which side of zero a voter falls on?

Figure~\ref{fig:sw_partisan_gender}, the central display of this application, answers it.  Panel~A reports DML estimates of the average gender effect overall and within party: indistinguishable from zero in the full sample ($\hat{\theta} = -0.11$, $[-0.24, 0.02]$), clearly negative among Democrats ($-0.32$), and positive among Republicans ($+0.15$, interval including zero).  Panel~B shows the full $\hat{\beta}_{i,\text{Male}}$ distributions behind these averages, whose party means ($-0.33$ for Democrats, $+0.23$ for Republicans, $-0.12$ for Independents) differ modestly from the debiased Panel~A subgroup estimates: Democrats prefer female ($68\%$), Republicans male ($69\%$), Independents in between.  Their density-weighted average is the null AMCE---the population mean is zero not from indifference but because two large groups disagree in opposite directions.\footnote{With only three tasks per respondent, these partisan splits are model-based posterior summaries that lean substantially on the covariate mean stage rather than design-identified quantities; Supplementary Materials~\ref{app:factorial_guidance} places distributional targets like these beyond a $T = 3$ design. We read the candidate application as illustrative of what the model can describe in a deliberately sparse design, and rest the paper's distributional conclusions primarily on the richer democracy and tax conjoints.}

\begin{figure}[t!]
  \centering
  \caption{Average and Individual-level Gender Preference Estimates:\\ Saha \& Weeks (2022)}
  \label{fig:sw_partisan_gender}
  \includegraphics[width=\textwidth]{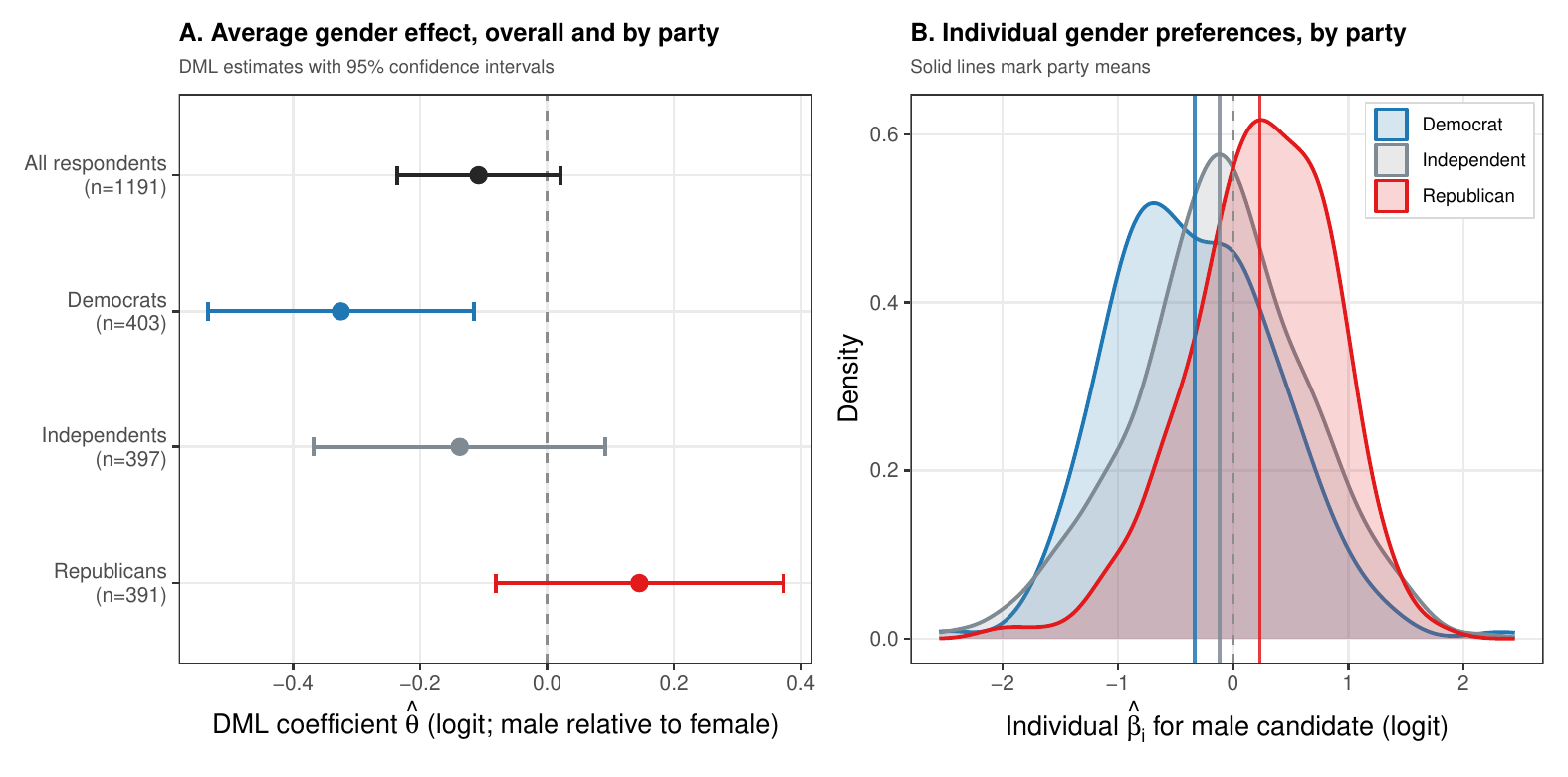}
  \begin{minipage}[]{1\textwidth}\footnotesize
        \textit{Note:} \textbf{A}: DML estimates of the average gender effect (male relative to female) with 95\% confidence intervals, overall and by respondent party.  \textbf{B}: the full distributions of the recovered $\hat{\beta}_{i,\text{Male}}$ by party; solid lines mark party means, the dashed line zero.  The near-zero population AMCE ($\hat{\theta} = -0.11$, $[-0.24, 0.02]$) masks near-mirror-image partisan preferences: Democrats average $-0.33$ ($68\%$ prefer female), Republicans $+0.23$ ($69\%$ prefer male), Independents $-0.12$ between.
  \end{minipage}
\end{figure}

Sign and spread tell us which way voters lean, but not how much each attribute drives the vote: two attributes can carry equally large coefficients yet matter very differently once we account for how far apart their levels are spread.  The importance share $\text{Imp}_{i,g}$ of Section~\ref{sec:graham_svolik} quantifies this; Figure~\ref{fig:sw_importance} (Supplementary Materials~\ref{app:more_figures}) shows its distribution.  Policy agenda is the largest share (mean $52\%$)---more than its $\hat{\theta}$ alone would suggest, since its three highly differentiated levels compound---followed by talent ($21\%$), gender ($16.5\%$), children ($8\%$), and progressive ambition ($2\%$).  Gender's share is the variance-decomposition form of the polarization finding: though $\hat{\theta}_{\text{Male}}$ is near zero, the opposed-partisan spread of $\hat{\beta}_{i,\text{Male}}$ makes it the third-largest driver of decision variance.  A reduced-form AMCE registers gender as a non-issue; the structural model reveals it as a major source of partisan disagreement.

Finally, an electoral-competition counterfactual in the style of Section~\ref{sec:graham_svolik}---Candidate A, an Empathetic Female, scaling back her agenda against a Tough Negotiator Male---shows two mechanisms at once: a near-uniform agenda penalty of about 12 points that no partisan benefit offsets, and a steady $\sim$25-point partisan gap the AMCE cannot deliver (Figure~\ref{fig:sw_agenda_counterfactual}, Supplementary Materials~\ref{app:more_figures}).  The same construction extends to any contest, including ones where multiple attributes move together.

\section{Limitations and Extensions}\label{sec:limitations}

Our framework yields interpretable quantities of interest, but at the cost of additional structural assumptions. In this section, we discuss the resulting limitations and, for each, describe an extension that relaxes the corresponding assumption. Most of these extensions are already implemented in the software or developed in the Supplementary Materials; the main open issue is formal inference for the full preference distribution.

First, the recovered preference distribution can understate its true dispersion.  The decomposition $\bfbeta_i = f(\bfZ_i) + \boldsymbol{\eta}_i$ allows residual heterogeneity within covariate strata, and the empirical-Bayes update uses a respondent's repeated choices to recover posterior summaries of $\boldsymbol{\eta}_i$.  But with finite $T$, $\hat{\bfbeta}_i$ remains a shrinkage estimator: when respondent-level data are sparse it is pulled toward $\hat{f}(\bfZ_i)$, so the recovered distribution can be too narrow when $T$ is small or the covariates are weak.  How much this matters depends on the quantity.  Distributional summaries such as polarization, and individual-level ratios such as a respondent's MRS, inherit the quality of both stages; the population averages do not, since their debiased estimators rest on the mean stage alone (Section~\ref{sec:estimation}).  In practice this means designing for both stages---rich covariates for the mean stage, enough tasks per respondent for the update; Supplementary Materials~\ref{app:factorial_guidance} translates this into design guidance by quantity of interest.

Second, utility is modeled as additive and linear in $\bfX$, which rules out attribute interactions and noncompensatory decision rules unless they are entered explicitly.  Pairwise interactions are straightforward to add but high-dimensional ($p(p+1)/2$ terms); a more systematic extension is a low-rank interaction layer in the DNN, where the network outputs both a main-effect vector $f(\bfZ) \in \mathbb{R}^p$ and a low-rank factor $\mathbf{V} \in \mathbb{R}^{p \times r}$ with $r \ll p$, so each profile's utility gains the quadratic term $\|\mathbf{V}^\top \bfX\|^2$ and the choice index becomes $\bfDelta\bfX^\top f(\bfZ) + \|\mathbf{V}^\top \bfX_A\|^2 - \|\mathbf{V}^\top \bfX_B\|^2$.  The software implements this extension alongside the explicit regularized pairwise-interaction model (Supplementary Materials~\ref{app:interactions}), and an application-based check (Table~\ref{tab:sw_interactions}) leaves the main findings unchanged.  A further restriction comes from linearity in attribute \emph{levels}, which rules out interior ideal points on ordered attributes: a respondent whose most-preferred tax rate is a middle bracket, or who votes by spatial proximity over a policy scale, is represented only by a monotone slope.  On these dimensions the additive index is more restrictive than the parametric spatial form it otherwise generalizes \citep{fowler2023moderates}.  Dummy-coding the ordered attributes or adding curvature relaxes this at the cost of dimensionality. We use the linear coding for comparability with the original studies and flag where it binds.

Relatedly, the Gumbel distribution of the taste shocks keeps the error scale fixed. Preferences are therefore identified only up to this normalization. If true utility has a respondent-specific scale $\sigma_i$, the data identify $\bfbeta_i / \sigma_i$ rather than $\bfbeta_i$ itself. As a result, apparent taste heterogeneity may partly reflect scale heterogeneity. Not all quantities are affected by this scale normalization. Signs, direction-based shares, individual marginal rates of substitution, compensating-differential thresholds, and within-respondent importance shares are invariant to a respondent-specific scale. By contrast, cross-respondent comparisons of coefficient \emph{magnitudes} are not. This includes population MRS and willingness to pay in level units, as well as displays of preference intensity. Supplementary Materials~\ref{app:scale} provides a quantity-by-quantity invariance table and a heteroskedastic-logit diagnostic. We reword the cross-group intensity comparisons flagged by this diagnostic and leave joint modeling of scale and preference heterogeneity for future research.

Third, the importance shares and counterfactual vote shares are design-dependent estimands.  They are defined relative to the uniform randomization distribution of the design and to the unweighted analysis sample; under a realistic correlated distribution of candidate profiles, or a survey-reweighted population, they would differ \citep{delacuesta2022improving, bansak2023amce}.  We report design-distribution estimands throughout, present survey-weighted versions of the headline quantities in Supplementary Materials~\ref{app:gs_covariates}, and treat sensitivity to a specific target profile pool as an application-specific extension.

Fourth, we take binary forced choice as the canonical outcome format, but the hybrid model extends naturally to others---ratings (a Gaussian observation model with a closed-form respondent-level update), multinomial choice (a softmax likelihood), and rankings (Plackett--Luce)---changing only the observation model while the preference decomposition $\bfbeta_i = f(\bfZ_i) + \boldsymbol{\eta}_i$ is unchanged.  Supplementary Materials~\ref{app:extensions} works through each.  Binary forced choice is canonical here because most consequential decisions---voting, accepting or rejecting an offer, hiring---are themselves binary, and it has accordingly been the workhorse format in applied conjoint work; the appropriate format ultimately depends on the substantive context \citep{bansak2021conjoint}.

Finally, formal inference reaches the average parameters but not yet the full preference distribution.  The DML guarantees apply to the average parameters $\theta_k$ because the orthogonal score depends on the cross-fitted mean stage $f(\bfZ)$, not on the respondent-level MAP refinement.  The same orthogonal score extends to every \emph{smooth} functional of the mean stage: composing the score with a gradient (Proposition~\ref{prop:app_master}) delivers $\sqrt{N}$-consistent, asymptotically normal inference for the average marginal effect, counterfactual win probabilities and vote shares and their comparative statics, the between-respondent importance share, and the population-mean MRS and willingness to pay; these are developed in Supplementary Materials~\ref{app:other_estimands}.

What remains genuinely open, by contrast, is formal inference for the quantities that depend on the full shape of the latent coefficient distribution---the polarization fraction $\Pr(\beta_{ik}>0)$, the compensating-differential fractions of the same threshold form, the total (within-respondent) importance share, and individual-level preference rankings.  These hinge on the residual law $F_{\eta\mid\bfZ}$, which is not point-identified at a fixed number of tasks (Proposition~\ref{prop:app_fixedT}), and in the polarization case on a non-smooth indicator.  The accompanying software implements a respondent-cluster wild bootstrap for them, whose intervals attain near-nominal coverage in simulation (around $94\%$ against a $95\%$ target), with the caveat that the bootstrap quantifies sampling variability and does not undo the bias from fixed-$T$ shrinkage toward consensus.  Simulations show these quantities are recovered well in favorable design regimes (Supplementary Materials~\ref{app:factorial_guidance}); formal inferential tools for the residual-law quantities remain for future work.

\section{Conclusion}\label{sec:conclusion}

This paper develops a hybrid structural estimator for conjoint experiments that preserves the benefits of structural modeling without its usual parametric restrictions. Combining a random utility model with a flexible DNN mean function and respondent-level empirical-Bayes updating, it recovers an individual preference vector $\hat{\bfbeta}_i$ for every respondent while retaining valid debiased inference on population averages.  It recovers reduced-form AMCE analysis as a special case under correct specification, and adds the full distribution of preferences and the structural quantities classical theories of voting and electoral competition have long invoked: MRS, WTP, compensating differentials, polarization, and predicted vote shares in head-to-head contests between arbitrary candidate position bundles.

The applications give a sense of what becomes visible.  Strong baseline opposition to undemocratic behavior coexists with substantial heterogeneity in how voters weigh democracy against policy and party, letting us rank the violations by electoral cost and identify the counterfactual position bundles where co-partisanship preserves majority support and where it does not.  Individual-level progressivity preferences are widespread across every subgroup---a clear majority of Republicans prefer a steeply progressive plan to a flat one---and the recovered slopes correlate with respondents' self-reported ideal rates at $r = 0.43$, an external validation reduced-form approaches cannot deliver.  An average gender effect near zero masks sharp partisan polarization---Democrats and Republicans hold roughly equal-and-opposite preferences whose density-weighted average is the null AMCE---yet in that sparse design gender is a leading driver of decision variance.  Each is a pattern invisible to AMCE-style aggregation.

For practice, the simulations and applications yield four recommendations:
\begin{itemize}[leftmargin=1.5em,noitemsep]
  \item \textbf{Invest in covariates first.}  How much of the preference heterogeneity the covariates explain matters more for recovering individual preferences than adding respondents or tasks, regardless of how preferences depend on the covariates, so the design need not guess the functional form.  Before fielding a conjoint, collect the respondent characteristics most likely to predict the preferences under study.
  \item \textbf{Use enough tasks per respondent.}  Averages such as AMCEs are recovered accurately even with a handful of tasks.  Claims about how preferences are spread across respondents---the share favoring a policy, or which attribute matters most---need roughly ten or more tasks plus reasonably informative covariates.  Clear-cut conclusions are dependable then; close calls, such as a 52--48 split or two nearly tied attributes, need more tasks or stronger covariates.
  \item \textbf{Treat individual-level claims as the most demanding.}  Statements about a particular respondent---their own attribute ranking, or choice probabilities computed from their own coefficients---require the recovered vector to track the truth person by person.  This is reachable when $T$ is large or when covariates are highly informative ($R_Z^2$ around $0.5$--$0.75$); when they explain only about a third of the heterogeneity, even the largest designs fall short.  Aggregate and distributional claims succeed with much less, so researchers planning individual-level claims should prioritize covariate quality and a sufficiently large $T$.
  \item \textbf{Design for tradeoff quantities.}  Ratios such as the MRS and WTP are unstable for respondents whose denominator coefficient is near zero, so prefer population-level versions, which average first and carry valid debiased confidence intervals.  Compensating differentials compare sums rather than ratios and are recovered as well as the coefficients themselves.  When tradeoffs are the goal, draw the key attributes from continuous ranges, as in the tax-policy application, and plan on large samples ($N \geq 5{,}000$).
\end{itemize}
Supplementary Materials~\ref{app:factorial_guidance} maps these recommendations quantity by quantity.

This approach lets empirical work catch up with theory.  Classical models of electoral competition---spatial voting, candidate moderation, valence tradeoffs, median-voter dynamics---have for decades predicted how voters respond when candidates move through concrete policy positions, but those predictions have lived mostly in formal theory or in observational studies with limited identification.  With individual preference vectors recovered from randomized conjoint contrasts, the predictions become directly testable: how do predicted vote shares move as candidates moderate, who gains and who loses, where do voters cross party lines, and which position bundles maximize support for which constituencies.  It also sharpens inferences about cross-pressure, polarization, and preference-distribution shape in domains where conjoint designs are standard but analyses have stopped at marginal effects. We implement the method in \ifanon an R package\else the R package \texttt{sconjoint}\fi, which provides functions for estimation, inference, and visualization.\footnote{\ifanon The software package and a user tutorial will be made available upon publication.\else A user tutorial is available at \url{https://yiqingxu.org/packages/sconjoint/}. The package was developed using \href{https://statsclaw.ai/}{\texttt{StatsClaw}}, an AI-collaborative workflow described in \citet{qinxu2026statsclaw}.\fi}

\vspace{1em}

\onehalfspacing
\bibliographystyle{apsr}
\bibliography{refs}

@article{abramson2022voter,
  author    = {Abramson, Scott F. and Ko\c{c}ak, Korhan and Magazinnik, Asya},
  title     = {What Do We Learn about Voter Preferences from Conjoint Experiments?},
  journal   = {American Journal of Political Science},
  year      = {2022},
  volume    = {66},
  number    = {4},
  pages     = {1008--1020}
}

@article{alvarez1998multiparty,
  author    = {Alvarez, R. Michael and Nagler, Jonathan},
  title     = {When Politics and Models Collide: Estimating Models of Multiparty Elections},
  journal   = {American Journal of Political Science},
  year      = {1998},
  volume    = {42},
  number    = {1},
  pages     = {55--96}
}

@book{austensmith1999positive,
  author    = {Austen-Smith, David and Banks, Jeffrey S.},
  title     = {Positive Political Theory I: Collective Preference},
  publisher = {University of Michigan Press},
  year      = {1999}
}

@article{bansak2016europeans,
  author    = {Bansak, Kirk and Hainmueller, Jens and Hangartner, Dominik},
  title     = {How Economic, Humanitarian, and Religious Concerns Shape European Attitudes Toward Asylum Seekers},
  journal   = {Science},
  year      = {2016},
  volume    = {354},
  number    = {6309},
  pages     = {217--222}
}

@article{bansak2023amce,
  author    = {Bansak, Kirk and Hainmueller, Jens and Hopkins, Daniel J. and Yamamoto, Teppei},
  title     = {Using Conjoint Experiments to Analyze Election Outcomes: The Essential Role of the Average Marginal Component Effect},
  journal   = {Political Analysis},
  year      = {2023},
  volume    = {31},
  number    = {4},
  pages     = {500--518}
}

@incollection{bansak2021conjoint,
  author    = {Bansak, Kirk and Hainmueller, Jens and Hopkins, Daniel J. and Yamamoto, Teppei},
  title     = {Conjoint Survey Experiments},
  booktitle = {Advances in Experimental Political Science},
  editor    = {Druckman, James N. and Green, Donald P.},
  publisher = {Cambridge University Press},
  year      = {2021},
  pages     = {19--41}
}

@article{ballard2017structure,
  author    = {Ballard-Rosa, Cameron and Martin, Lucy and Scheve, Kenneth},
  title     = {The Structure of {A}merican Income Tax Policy Preferences},
  journal   = {Journal of Politics},
  year      = {2017},
  volume    = {79},
  number    = {1},
  pages     = {1--16}
}

@article{chernozhukov2018double,
  author    = {Chernozhukov, Victor and Chetverikov, Denis and Demirer, Mert and Duflo, Esther and Hansen, Christian and Newey, Whitney and Robins, James},
  title     = {Double/Debiased Machine Learning for Treatment and Structural Parameters},
  journal   = {The Econometrics Journal},
  year      = {2018},
  volume    = {21},
  number    = {1},
  pages     = {C1--C68}
}

@article{clinton2004statistical,
  author    = {Clinton, Joshua and Jackman, Simon and Rivers, Douglas},
  title     = {The Statistical Analysis of Roll Call Data},
  journal   = {American Political Science Review},
  year      = {2004},
  volume    = {98},
  number    = {2},
  pages     = {355--370}
}

@article{delacuesta2022improving,
  author    = {de la Cuesta, Brandon and Egami, Naoki and Imai, Kosuke},
  title     = {Improving the External Validity of Conjoint Analysis: The Essential Role of Profile Distribution},
  journal   = {Political Analysis},
  year      = {2022},
  volume    = {30},
  number    = {1},
  pages     = {19--45}
}

@book{enelow1984spatial,
  author    = {Enelow, James M. and Hinich, Melvin J.},
  title     = {The Spatial Theory of Voting: An Introduction},
  publisher = {Cambridge University Press},
  year      = {1984}
}

@article{farrell2021deep,
  author    = {Farrell, Max H. and Liang, Tengyuan and Misra, Sanjog},
  title     = {Deep Neural Networks for Estimation and Inference},
  journal   = {Econometrica},
  year      = {2021},
  volume    = {89},
  number    = {1},
  pages     = {181--213}
}

@unpublished{farrell2025deep,
  author    = {Farrell, Max H. and Liang, Tengyuan and Misra, Sanjog},
  title     = {Deep Learning for Individual Heterogeneity: An Automatic Inference Framework},
  year      = {2025},
  note      = {Working paper, arXiv:2010.14694}
}

@article{goplerud2025estimating,
  author    = {Goplerud, Max and Imai, Kosuke and Pashley, Nicole E.},
  title     = {Estimating Heterogeneous Causal Effects of High-Dimensional Treatments: Application to Conjoint Analysis},
  journal   = {Annals of Applied Statistics},
  year      = {2025},
  volume    = {19},
  number    = {2},
  pages     = {866--888}
}

@article{graham2020democracy,
  author    = {Graham, Matthew H. and Svolik, Milan W.},
  title     = {Democracy in {A}merica? {P}artisanship, Polarization, and the Robustness of Support for Democracy in the {U}nited {S}tates},
  journal   = {American Political Science Review},
  year      = {2020},
  volume    = {114},
  number    = {2},
  pages     = {392--409}
}

@article{green1971conjoint,
  author    = {Green, Paul E. and Rao, Vithala R.},
  title     = {Conjoint Measurement for Quantifying Judgmental Data},
  journal   = {Journal of Marketing Research},
  year      = {1971},
  volume    = {8},
  number    = {3},
  pages     = {355--363}
}

@article{green1978conjoint,
  author    = {Green, Paul E. and Srinivasan, V.},
  title     = {Conjoint Analysis in Consumer Research: Issues and Outlook},
  journal   = {Journal of Consumer Research},
  year      = {1978},
  volume    = {5},
  number    = {2},
  pages     = {103--123}
}

@article{green1990conjoint,
  author    = {Green, Paul E. and Srinivasan, V.},
  title     = {Conjoint Analysis in Marketing: New Developments with Implications for Research and Practice},
  journal   = {Journal of Marketing},
  year      = {1990},
  volume    = {54},
  number    = {4},
  pages     = {3--19}
}

@article{greenhalgh1981conjoint,
  author    = {Greenhalgh, Leonard and Neslin, Scott A.},
  title     = {Conjoint Analysis of Negotiator Preferences},
  journal   = {Journal of Conflict Resolution},
  year      = {1981},
  volume    = {25},
  number    = {2},
  pages     = {301--327}
}

@book{hinich1997analytical,
  author    = {Hinich, Melvin J. and Munger, Michael C.},
  title     = {Analytical Politics},
  publisher = {Cambridge University Press},
  year      = {1997}
}

@article{hainmueller2014causal,
  author    = {Hainmueller, Jens and Hopkins, Daniel J. and Yamamoto, Teppei},
  title     = {Causal Inference in Conjoint Analysis: Understanding Multidimensional Choices via Stated Preference Experiments},
  journal   = {Political Analysis},
  year      = {2014},
  volume    = {22},
  number    = {1},
  pages     = {1--30}
}

@article{hainmueller2015hidden,
  author    = {Hainmueller, Jens and Hopkins, Daniel J.},
  title     = {The Hidden {A}merican Immigration Consensus: A Conjoint Analysis of Attitudes toward Immigrants},
  journal   = {American Journal of Political Science},
  year      = {2015},
  volume    = {59},
  number    = {3},
  pages     = {529--548}
}

@article{ham2024machine,
  author    = {Ham, Dae Woong and Imai, Kosuke and Janson, Lucas},
  title     = {Using Machine Learning to Test Causal Hypotheses in Conjoint Analysis},
  journal   = {Political Analysis},
  year      = {2024},
  volume    = {32},
  number    = {3},
  pages     = {329--344}
}

@article{leeper2020measuring,
  author    = {Leeper, Thomas J. and Hobolt, Sara B. and Tilley, James},
  title     = {Measuring Subgroup Preferences in Conjoint Experiments},
  journal   = {Political Analysis},
  year      = {2020},
  volume    = {28},
  number    = {2},
  pages     = {207--221}
}

@book{louviere2000stated,
  author    = {Louviere, Jordan J. and Hensher, David A. and Swait, Joffre D.},
  title     = {Stated Choice Methods: Analysis and Applications},
  publisher = {Cambridge University Press},
  year      = {2000}
}

@article{luce1964simultaneous,
  author    = {Luce, R. Duncan and Tukey, John W.},
  title     = {Simultaneous Conjoint Measurement: A New Type of Fundamental Measurement},
  journal   = {Journal of Mathematical Psychology},
  year      = {1964},
  volume    = {1},
  number    = {1},
  pages     = {1--27}
}

@article{martin2002dynamic,
  author    = {Martin, Andrew D. and Quinn, Kevin M.},
  title     = {Dynamic Ideal Point Estimation via {M}arkov Chain {M}onte {C}arlo for the {U.S.} {S}upreme {C}ourt, 1953--1999},
  journal   = {Political Analysis},
  year      = {2002},
  volume    = {10},
  number    = {2},
  pages     = {134--153}
}

@incollection{mcfadden1974conditional,
  author    = {McFadden, Daniel},
  title     = {Conditional Logit Analysis of Qualitative Choice Behavior},
  booktitle = {Frontiers in Econometrics},
  editor    = {Zarembka, Paul},
  publisher = {Academic Press},
  year      = {1974},
  pages     = {105--142}
}

@incollection{mcfadden1981econometric,
  author    = {McFadden, Daniel},
  title     = {Econometric Models of Probabilistic Choice},
  booktitle = {Structural Analysis of Discrete Data with Econometric Applications},
  editor    = {Manski, Charles F. and McFadden, Daniel},
  publisher = {MIT Press},
  year      = {1981},
  pages     = {198--272}
}

@article{palfrey1987relationship,
  author    = {Palfrey, Thomas R. and Poole, Keith T.},
  title     = {The Relationship between Information, Ideology, and Voting Behavior},
  journal   = {American Journal of Political Science},
  year      = {1987},
  volume    = {31},
  number    = {3},
  pages     = {511--530}
}

@article{poole1985spatial,
  author    = {Poole, Keith T. and Rosenthal, Howard},
  title     = {A Spatial Model for Legislative Roll Call Analysis},
  journal   = {American Journal of Political Science},
  year      = {1985},
  volume    = {29},
  number    = {2},
  pages     = {357--384}
}

@article{rivers1988heterogeneity,
  author    = {Rivers, Douglas},
  title     = {Heterogeneity in Models of Electoral Choice},
  journal   = {American Journal of Political Science},
  year      = {1988},
  volume    = {32},
  number    = {3},
  pages     = {737--757}
}

@article{robinson2024detect,
  author    = {Robinson, Thomas S. and Duch, Raymond},
  title     = {How to Detect Heterogeneity in Conjoint Experiments},
  journal   = {Journal of Politics},
  year      = {2024},
  volume    = {86},
  number    = {2},
  pages     = {412--427}
}

@article{saha2022ambitious,
  author    = {Saha, Sparsha and Weeks, Ana Catalano},
  title     = {Ambitious Women: Gender and Voter Perceptions of Candidate Ambition},
  journal   = {Political Behavior},
  year      = {2022},
  volume    = {44},
  pages     = {779--805},
  number    = {4},
  doi       = {10.1007/s11109-020-09636-z}
}

@article{teele2018ties,
  author    = {Teele, Dawn Langan and Kalla, Joshua L. and Rosenbluth, Frances},
  title     = {The Ties That Double Bind: Social Roles and Women's Underrepresentation in Politics},
  journal   = {American Political Science Review},
  year      = {2018},
  volume    = {112},
  number    = {3},
  pages     = {525--541}
}

@article{shamir1995competing,
  author    = {Shamir, Michal and Shamir, Jacob},
  title     = {Competing Values in Public Opinion: A Conjoint Analysis},
  journal   = {Political Behavior},
  year      = {1995},
  volume    = {17},
  number    = {1},
  pages     = {107--133}
}

@book{train2009discrete,
  author    = {Train, Kenneth E.},
  title     = {Discrete Choice Methods with Simulation},
  edition   = {2nd},
  publisher = {Cambridge University Press},
  year      = {2009}
}

@misc{qinxu2026statsclaw,
  title={StatsClaw: An AI-Collaborative Workflow for Statistical Software Development},
  author={Qin, Tianzhu and Xu, Yiqing},
  year={2026},
  note={arXiv:2604.04871},
  eprint={2604.04871},
  archivePrefix={arXiv},
  primaryClass={cs.SE},
  doi={10.48550/arXiv.2604.04871},
  url={https://arxiv.org/abs/2604.04871}
}

@article{black1948rationale,
  author    = {Black, Duncan},
  title     = {On the Rationale of Group Decision-making},
  journal   = {Journal of Political Economy},
  year      = {1948},
  volume    = {56},
  number    = {1},
  pages     = {23--34}
}

@book{downs1957economic,
  author    = {Downs, Anthony},
  title     = {An Economic Theory of Democracy},
  publisher = {Harper and Row},
  year      = {1957},
  address   = {New York}
}

@article{zhirkov2022individual,
  author  = {Zhirkov, Kirill},
  title   = {Estimating and Using Individual Marginal Component Effects from Conjoint Experiments},
  journal = {Political Analysis},
  year    = {2022}, volume = {30}, number = {2}, pages = {236--249}
}

@article{lenk1996hierarchical,
  author  = {Lenk, Peter J. and DeSarbo, Wayne S. and Green, Paul E. and Young, Martin R.},
  title   = {Hierarchical {B}ayes Conjoint Analysis: Recovery of Partworth Heterogeneity from Reduced Experimental Designs},
  journal = {Marketing Science}, year = {1996}, volume = {15}, number = {2}, pages = {173--191}
}

@book{rossi2005bayesian,
  author    = {Rossi, Peter E. and Allenby, Greg M. and McCulloch, Robert},
  title     = {Bayesian Statistics and Marketing}, publisher = {Wiley}, year = {2005}
}

@article{fowler2023moderates,
  author  = {Fowler, Anthony and Hill, Seth J. and Lewis, Jeffrey B. and Tausanovitch, Chris and Vavreck, Lynn and Warshaw, Christopher},
  title   = {Moderates}, journal = {American Political Science Review},
  year    = {2023}, volume = {117}, number = {2}, pages = {643--660}
}

@article{ansolabehere2008strength,
  author  = {Ansolabehere, Stephen and Rodden, Jonathan and Snyder, James M.},
  title   = {The Strength of Issues: Using Multiple Measures to Gauge Preference Stability, Ideological Constraint, and Issue Voting},
  journal = {American Political Science Review}, year = {2008}, volume = {102}, number = {2}, pages = {215--232}
}

\onehalfspacing
\newpage
\appendix
\setcounter{page}{1}
\setcounter{table}{0}
\setcounter{figure}{0}
\setcounter{equation}{0}
\setcounter{footnote}{0}
\renewcommand\thetable{A\arabic{table}}
\renewcommand\thefigure{A\arabic{figure}}
\renewcommand{\thepage}{A-\arabic{page}}
\renewcommand{\theequation}{A\arabic{equation}}
\renewcommand{\thefootnote}{A\arabic{footnote}}

\begin{center}
{\LARGE\bfseries Supplementary Materials}
\end{center}
\vspace{0.5em}

\noindent These supplementary materials contain the estimation details (\ref{app:estimation_details}), the asymptotic theory and debiased-inference results (\ref{sec:asymptotics} and~\ref{app:other_estimands}), the benchmark validation and the design-guidance simulation behind the checks of Section~\ref{sec:estimation} (\ref{app:validation} and~\ref{app:factorial}), additional results for the three applications (\ref{app:more_figures}), and extensions to other outcome formats (\ref{app:extensions}).\medskip

\startcontents[supplementary]
\printcontents[supplementary]{}{1}{\setcounter{tocdepth}{2}}
\vspace{1em}

\counterwithin{figure}{section}
\counterwithin{table}{section}
\vspace{3em}

\section{Estimation Details and Implementation}\label{app:estimation_details}

This section collects the technical details behind Section~\ref{sec:estimation}: the DNN mean stage, the respondent-level empirical-Bayes update, and the main alternative stage-2 estimators used in the simulations.  Table~\ref{tab:comparison} situates the hybrid estimator relative to existing approaches.

\begin{table}[!h]
\centering
\footnotesize
\caption{Comparison of Approaches to Conjoint Analysis}\label{tab:comparison}
\resizebox*{\textwidth}{!}{
{\setlength{\tabcolsep}{4pt}
\begin{tabular}{@{}l p{0.19\textwidth} cccc cc@{}}
\toprule
& \multicolumn{4}{c}{\textit{Structural}} & \multicolumn{2}{c}{\textit{Reduced-form}} \\
\cmidrule(lr){2-5} \cmidrule(lr){6-7}
& Hybrid DNN-EB & Homog. & Mixed & Hier. & & BART/ \\
& (ours)$^\dagger$ & Logit$^\dagger$ & Logit$^\dagger$ & Bayes$^\dagger$ & AMCE & CF \\[\medskipamount]
\midrule
Utility model              & Yes & Yes & Yes & Yes & Implicit & No \\[\medskipamount]
Preference heterogeneity   & Flexible mean + EB residual & None & Parametric & Parametric & None & Nonparam. \\[\medskipamount]
Systematic by $\bfZ$       & Flexible & --- & Linear & Linear & --- & Flexible \\[\medskipamount]
Individual $\bfbeta_i$     & \shortstack[l]{EB posterior\\mode} & --- & Posterior$^*$ & Posterior$^*$ & --- & --- \\[\medskipamount]
Structural quantities            & Yes & Yes & Yes & Yes & No & No \\[\medskipamount]
Distributional assumption  & \shortstack[l]{Working Gaussian\\prior} & --- & Normal & Normal & None & None \\[\medskipamount]
Inference on average quantities    & DML & MLE & MLE & MCMC & OLS & --- \\
\bottomrule
\end{tabular}
}}
\begin{minipage}[]{1\textwidth}\footnotesize
    \textit{Note:} Structural quantities include MRS, WTP, counterfactual choice probabilities, and compensating differentials---all of which require the utility model and individual preference vectors.  Our method combines a flexible DNN mean stage with respondent-level empirical-Bayes (EB) updating.  $^*$Individual-level posteriors from Mixed Logit and hierarchical Bayes are heavily shrunk toward the population mean when the number of tasks per respondent is small relative to the number of parameters.  BART/CF denotes tree-based machine-learning estimators such as Bayesian additive regression trees (BART) \citep{robinson2024detect} and causal forests (CF).  ``---'' indicates the feature is not applicable or not provided by the method.  Inference-row abbreviations: DML (double machine learning), MLE (maximum likelihood estimation), MCMC (Markov chain Monte Carlo), and OLS (ordinary least squares).
\end{minipage}
\end{table}

\subsection{Mean-Stage DNN}

The DNN mean stage consists of a \textit{feature network} that maps respondent characteristics to preference parameters, and a \textit{model layer} that embeds these parameters in the structural logit model.  The feature network takes $\bfZ \in \mathbb{R}^{p_Z}$ as input and applies $L$ hidden layers with ReLU activations:
\[
  \bfh_\ell = \text{ReLU}(\bfW_\ell \bfh_{\ell-1} + \bfb_\ell), \qquad \ell = 1,\ldots,L,
\]
with $\bfh_0 = \bfZ$ and $\text{ReLU}(x)=\max(0,x)$ applied elementwise.  The final hidden layer maps to the $p$-dimensional mean preference vector
\[
  f(\bfZ) = \bfW_{L+1}\bfh_L + \bfb_{L+1} \in \mathbb{R}^p,
\]
with no output activation, since the preference coefficients are unrestricted in sign and magnitude.  The model layer then computes the structural logit index $\bfDelta\bfX^\top f(\bfZ)$.

The network is trained by minimizing the mean-stage conditional-logit binary cross-entropy loss
\[
  \mathcal{L}
  =
  -\frac{1}{n}
  \sum_{i=1}^{N}\sum_{t=1}^{T_i}
  \left[
    Y_{it}\log G\!\left(\bfDelta\bfX_{it}^\top f(\bfZ_i)\right)
    +
    (1-Y_{it})\log\!\left(1-G\!\left(\bfDelta\bfX_{it}^\top f(\bfZ_i)\right)\right)
  \right],
  \quad n=\sum_{i=1}^{N}T_i.
\]
This is the negative mean-stage conditional-logit log-likelihood for the model $\Pr(Y_{it}=1\mid \bfDelta\bfX_{it},\bfZ_i)=G(\bfDelta\bfX_{it}^\top f(\bfZ_i))$.  In our implementation, optimization uses Adam in \texttt{torch}, all cross-fitting is performed at the respondent level so that every respondent's tasks remain in the same fold, and the binary-choice production runs use no inner-split early stopping together with a fixed budget of $1{,}000$ epochs per fold.

Cross-fitting proceeds as follows.  We partition respondents into $K$ folds, train the DNN on $K-1$ folds, and predict $\hat f(\bfZ_i)$ only for respondents in the held-out fold.  Cycling through all folds yields an out-of-fold mean preference vector for every respondent.  In the locked binary implementation we repeat this full respondent-level cross-fit twice with independent seeds and average the two out-of-fold predictions before any downstream step.  This ensemble mean is the object used both in the orthogonal score for population averages and in the prior mean for the respondent-level update.

The binary-choice production rule further regularizes the DNN mean stage with an adaptive ridge penalty $\text{weight decay}=K/n$, where $K=15$ if $n/p<300$ and $K=25$ otherwise.  In the larger-$p$ corner with $p\geq 40$ and $n\geq 80{,}000$, we also widen the final hidden layer from $32$ to $64$ units.  None of the three applications in the main text trigger this architecture override, but it improves the hardest high-$p$ benchmark cells and is therefore part of the locked implementation.

\subsection{Alternative First-Stage Learners}\label{app:learners}

The debiased inference of Supplementary Materials~\ref{app:other_estimands} depends on the first stage only through the cross-fitted mean preference vector $\hat f(\bfZ_i)$ and the local information matrix $\hat{\bfLambda}(\bfZ_i)$.  The orthogonal score, the respondent-clustered variance, and every downstream quantity are computed from these two objects and are otherwise indifferent to how $\hat f$ was obtained.  Neyman orthogonality asks only that the first stage converge at the $o_p(N^{-1/4})$ rate of Supplementary Materials~\ref{app:asymptotic_results}, which many flexible learners can meet.  We use a deep network throughout the applications because it scales naturally to rich $\bfZ$ and to larger attribute spaces, but it is not the only admissible choice.  The accompanying software implements two alternatives, which we use here to show that the validity of the debiased intervals does not hinge on the particular learner.

The first is an elastic-net logit.  We expand each moderator into a natural-spline basis $\phi(\bfZ_i)$ (with pairwise products across moderators), interact it with the attribute contrasts, and regress the binary choice on $[\bfDelta\bfX_{it},\,\bfDelta\bfX_{it}\otimes\phi(\bfZ_i)]$ under an elastic-net penalty; the implied preference vector is read off the fitted index as $\hat f_k(\bfZ_i)=\hat B^{0}_k+\langle\hat B^{\mathrm{int}}_{k},\phi(\bfZ_i)\rangle$.  The basis expansion lets preferences vary nonlinearly in the moderators, and the penalty selects which expanded terms to retain, so the fit is flexible while remaining transparent and fast.  The second is a GRF.  We grow a multivariate forest of per-respondent preference coefficients on $\bfZ_i$, and for each respondent use the forest's adaptive neighbor weights to solve a locally weighted logit moment, which gives a forest-localized estimate of $\hat f(\bfZ_i)$.  Unlike the elastic net, the GRF captures nonlinear and interactive heterogeneity without a pre-specified basis, at the cost of explicit coefficients.

Both learners enter the same cross-fitting, $\hat{\bfLambda}$, orthogonal-score, and clustered-variance computations as the network; only $\hat f$ changes.  Table~\ref{tab:learner_coverage} reports Monte-Carlo coverage of the nominal $95\%$ debiased intervals on a data-generating process in which preferences vary smoothly with three moderators and each respondent completes eight tasks ($N=1{,}000$ respondents, $300$ replications).  Across the estimand suite---the average preference parameter, the AME, counterfactual vote share and its comparative static, the attribute-importance share, willingness to pay, and the marginal rate of substitution---all three learners are reasonably calibrated, with the deep network the most tightly so; the elastic-net and forest learners run a few points lower, most visibly on the hardest ratio and comparative-static quantities.  The marginal rate of substitution, a ratio of average parameters, is the hardest quantity for every learner, consistent with the ordering in Supplementary Materials~\ref{sec:asymptotics}.  Preference polarization, the open estimand of Supplementary Materials~\ref{app:other_estimands}, is reported only through the respondent-cluster wild bootstrap and is omitted here.

The choice of learner does matter when preferences vary nonlinearly with the moderators.  Table~\ref{tab:learner_nonlinear} reports coverage of the average preference parameter $\theta_k$ on the simple process above and on a nonlinear one whose mean preference surface is genuinely curved.  An elastic net using only the raw moderators is well-calibrated when preferences are near-linear ($0.95$) but under-covers sharply when they are not ($0.68$): a linear first stage cannot track the curved projection, and the resulting bias does not vanish as the sample grows.  Expanding the moderators into a spline basis restores nominal coverage ($0.97$).  The deep network and the GRF, flexible by construction, are robust to both surfaces.  This is why the elastic-net first stage expands its basis automatically, and why we adopt the deep network throughout the applications: it is the most robust to the unknown shape of preference heterogeneity.

\begin{table}[t]
\centering
\footnotesize
\caption{Monte-Carlo Coverage of Nominal $95\%$ Debiased Confidence Intervals across First-stage Learners}
\label{tab:learner_coverage}
\begin{tabular}{@{}l ccc@{}}
\toprule
Estimand & DNN & Elastic net & GRF \\
\midrule
Average preference parameter $\theta_k$ & 0.94 & 0.88 & 0.89 \\
Average marginal effect (probability)   & 0.95 & 0.91 & 0.91 \\
Counterfactual vote share               & 0.96 & 0.90 & 0.91 \\
Vote-share comparative static           & 0.92 & 0.85 & 0.89 \\
Attribute-importance share              & 0.99 & 0.95 & 0.95 \\
Willingness to pay                      & 0.94 & 0.93 & 0.94 \\
Marginal rate of substitution           & 0.85 & 0.84 & 0.84 \\
\bottomrule
\end{tabular}
\begin{minipage}[]{1\textwidth}\footnotesize
    \textit{Note:} Empirical coverage of nominal $95\%$ confidence intervals over $300$ Monte-Carlo replications at $N=1{,}000$ respondents and $T=8$ tasks, on a data-generating process in which the mean preference vector varies smoothly with three respondent moderators.  The elastic net uses its default spline-expanded moderator basis.  All three learners feed the identical cross-fitting and orthogonal-score computations; only the first-stage mean estimate differs.  The marginal rate of substitution is a ratio of average parameters and is the hardest case in the hierarchy of Supplementary Materials~\ref{sec:asymptotics}.
\end{minipage}
\end{table}

\begin{table}[t]
\centering
\footnotesize
\caption{Coverage of the Average Preference Parameter $\theta_k$ under a Simple and a Nonlinear Data-generating Process, by First-stage Learner}
\label{tab:learner_nonlinear}
\begin{tabular}{@{}l cc@{}}
\toprule
First-stage learner & Simple DGP & Nonlinear DGP \\
\midrule
Deep neural network          & 0.94 & 0.94 \\
Elastic net (linear basis)   & 0.95 & 0.68 \\
Elastic net (spline basis)   & 0.88 & 0.97 \\
GRF    & 0.89 & 0.89 \\
\bottomrule
\end{tabular}
\par\vspace{3pt}
\begin{minipage}[]{1\textwidth}\footnotesize
    \textit{Note:} Empirical coverage of nominal $95\%$ intervals for the average preference parameter $\theta_k$ over $300$ Monte-Carlo replications at $N=1{,}000$ respondents and $T=8$ tasks.  The simple process has preferences near-linear in the three moderators; the nonlinear process has a genuinely curved mean preference surface.  The linear-basis elastic net under-covers under nonlinearity because its first-stage bias does not vanish with sample size; the spline-expanded basis (the package default) restores nominal coverage.  The deep network and the GRF are flexible by construction and are robust to both surfaces.
\end{minipage}
\end{table}

\subsection{Respondent-Level Empirical-Bayes Update}

The conceptual second stage is
\[
  \bfbeta_i = f(\bfZ_i) + \boldsymbol{\eta}_i, \qquad \boldsymbol{\eta}_i \sim N(\mathbf{0}, \boldsymbol{\Sigma}_\eta),
\]
but in the low-task regime typical of conjoint studies we do not estimate a full unrestricted $\boldsymbol{\Sigma}_\eta$.  Instead, we average two independently seeded full cross-fits to form an ensemble prior mean:
\[
  \hat f_{\text{ens}}(\bfZ_i) = \frac{1}{2}\left(\hat f^{(1)}(\bfZ_i) + \hat f^{(2)}(\bfZ_i)\right).
\]
We then construct a diagonal working covariance from first-stage residual scores.  Let
\[
  \hat G_{it} = G\!\left(\bfDelta\bfX_{it}^\top \hat f_{\text{ens}}(\bfZ_i)\right), \qquad
  r_{it} = Y_{it} - \hat G_{it},
\]
and define the coordinate-specific task score
\[
  s_{itk} = \Delta X_{itk}\, r_{it}.
\]
For respondent $i$, average this score over tasks:
\[
  \bar s_{ik} = \frac{1}{T_i}\sum_{t=1}^{T_i} s_{itk}.
\]
Let the coordinate-specific average score information be
\[
  \bar v_k =
  \frac{1}{n}
  \sum_{i=1}^{N}\sum_{t=1}^{T_i}
  \hat G_{it}(1-\hat G_{it})\,\Delta X_{itk}^{2},
  \qquad n=\sum_{i=1}^{N}T_i,
\]
and let $\bar\tau=N^{-1}\sum_{i=1}^{N}T_i^{-1}$.  The score-based scale heuristic is
\[
  \hat\sigma_{\text{score},k}^2
  =
  \max\!\left\{
    \frac{\Var_i(\bar s_{ik})}{\bar v_k^2}
    -
    \frac{\bar\tau}{\bar v_k},
    \; 0.01
  \right\}.
\]
When $T_i=T$ for all respondents, the subtraction term is the balanced-design correction $1/(\bar v_k T)$.  The floor at $0.01$ prevents degenerate priors in near-separated or weak-signal cells.  We use this score-based scale heuristic as a coordinate-specific working scale proxy, not as an unbiased or consistent estimator of $\Var(\eta_{ik})$, and introduce a fixed prior-precision calibration constant $c_\eta$:
\[
  \hat{\sigma}_{\eta,k}^2(c_\eta)
  =
  \frac{\hat{\sigma}_{\text{score},k}^2}{c_\eta}.
\]
The default \texttt{EnsC5} calibration sets $c_\eta=5$, shrinking the working covariance by a factor of $5$, or equivalently multiplying the prior precision by $5$.

Given $\hat f_{\text{ens}}(\bfZ_i)$ and
\[
  \hat{\boldsymbol{\Sigma}}_\eta(c_\eta)
  =
  \operatorname{diag}\!\left(
    \hat{\sigma}_{\eta,1}^2(c_\eta),
    \ldots,
    \hat{\sigma}_{\eta,p}^2(c_\eta)
  \right),
\]
and writing $\hat{\boldsymbol{\Sigma}}_\eta=\hat{\boldsymbol{\Sigma}}_\eta(c_\eta)$, the respondent-specific update is the maximum-a-posteriori (MAP) problem of §\ref{sec:estimation}, now evaluated with the score-based working covariance $\hat{\boldsymbol{\Sigma}}_\eta(c_\eta)$ constructed above in place of a generic prior.  In implementation we solve it by Newton iterations on the residual component $\boldsymbol{\eta}_i = \bfbeta_i - \hat f_{\text{ens}}(\bfZ_i)$:
\[
  \nabla_i(\boldsymbol{\eta})
  =
  \bfDelta\bfX_i^\top(\mathbf{Y}_i-\mathbf{p}_i)
  -
  \hat{\boldsymbol{\Sigma}}_\eta^{-1}\boldsymbol{\eta},
\]
\[
  \mathbf{H}_i(\boldsymbol{\eta})
  =
  -\bfDelta\bfX_i^\top \mathbf{W}_i \bfDelta\bfX_i
  -
  \hat{\boldsymbol{\Sigma}}_\eta^{-1},
\]
where $\mathbf{p}_i$ collects the fitted choice probabilities and $\mathbf{W}_i$ is the diagonal matrix of logit weights $p_{it}(1-p_{it})$.  Because the objective is strictly concave, these updates are numerically stable and fast in practice.

\paragraph{Choosing the calibration constant $c_\eta$.}  The constant $c_\eta=5$ is the empirical default we adopt for categorical-attribute applications.  It was selected through a paired-comparison diagnostic that holds the DNN ensemble $\hat f_{\text{ens}}$ fixed within each simulation cell and applies seven respondent-level estimators on the same nuisance: the unrestricted ensemble (no MAP), the calibration constants $c_\eta \in \{5,10,20,40\}$, an adaptive variant in which $c_\eta = 5 + 35\,\hat R_Z^2$ tracks an estimated share of $\bfZ$-explained heterogeneity, and an ``oracle'' calibration that plugs in the true innovation variance $\sigma_\eta^2$.  The grid spans $3{,}071$ paired cells across four DGP families, $N \in [1{,}000,\,10{,}000]$, $T \in \{3,5,10,15\}$, $p \in \{20,30\}$, and $R_Z^2 \in \{0.10,\,0.55,\,0.75\}$, with ten replications per cell.  The pooled mean paired improvement in individual-$\bfbeta$ correlation over the unrestricted ensemble is $+0.063$ for $c_\eta=5$, $+0.061$ for the adaptive variant, $+0.056$ for $c_\eta=10$, $+0.041$ for $c_\eta=20$, $+0.036$ for the oracle, and $+0.027$ for $c_\eta=40$.  Every MAP variant beats the unrestricted ensemble on the mean in $100\%$ of cells; only $c_\eta=5$ has a (small) negative single-replication worst case, of $-0.010$, whereas the adaptive variant has a strictly positive worst case of $+0.003$.  The adaptive variant has the cleaner adaptive rationale (when the DNN ensemble explains more of the cross-respondent variance, the prior should be tighter so the respondent's own choices have less leverage), but its performance is statistically indistinguishable from the fixed $c_\eta=5$ default and we therefore ship the simpler form.  The oracle that uses the true $\sigma_\eta^2$ lies $0.027$ below $c_\eta=5$ on the pooled mean, indicating that the $\hat\sigma_{\text{score}}^2/5$ calibration is implicitly compensating for bias in the ensemble prior mean rather than merely estimating the population innovation variance.  The regime structure underneath the pooled mean is intuitive: tighter calibrations dominate at large $R_Z^2$ (where the ensemble is a reliable shrinkage target), looser calibrations dominate at small $R_Z^2$ (where the ensemble is noisy), and $c_\eta=5$ sits at the practical frontier across the design.

\paragraph{Continuous-attribute applications.}  When the attribute differences $\bfDelta\bfX_{it}$ are continuous rather than indicator-coded---as in the tax-bracket application of \citet{ballard2017structure}---the score-based scale heuristic $\hat\sigma_{\text{score},k}^2$ defined above can become numerically very large because the within-respondent score variance no longer reflects discrete level contrasts.  In such cases the $\hat\sigma_{\text{score}}^2/5$ calibration produces a near-flat prior that effectively turns off the respondent-level update.  For these applications we replace the score-based heuristic with a variance-of-references heuristic that uses the cross-respondent dispersion of the ensemble itself,
\[
  \hat\sigma_{\eta,k}^{2,\text{varref}}
  =
  \max\!\left\{
    \tfrac{1}{2}\,\Var_i\!\left(\hat f_{\text{ens},k}(\bfZ_i)\right),\;
    \tau_{\text{floor}}
  \right\},
\]
with a small floor $\tau_{\text{floor}}=10^{-3}$ to prevent degenerate priors.  This is the calibration we use for the \citet{ballard2017structure} application; on the held-out self-reported ideal-rate validation set, the choice between $\tau_{\text{floor}}=10^{-3}$ and a fixed prior variance of $10^{-3}$ for every coefficient gives identical correlation up to numerical noise.

\paragraph{Survey weights.}  The applications report unweighted structural estimates.  The accompanying software optionally accepts respondent survey weights, which reweight the respondent-level aggregation that forms the population averages and their clustered standard errors; the first-stage learner and the per-respondent empirical-Bayes update are left unweighted, so the recovered $\hat\bfbeta_i$ are unchanged.  Only the democracy application's public file carries usable respondent weights; there, reweighting moves the headline quantities by at most a few hundredths on the logit scale (Supplementary Materials~\ref{app:gs_covariates}), so we report unweighted estimates in the main text.

\subsection{Alternative Stage-2 Estimators}

The most natural alternative to \texttt{EnsC5} is a DNN-offset mixed logit or hierarchical logit.  In that approach, one treats the DNN prediction as a fixed offset,
\[
  o_{it} = \bfDelta\bfX_{it}^\top \hat f_{\text{ens}}(\bfZ_i),
\]
and estimates respondent-specific random slopes from
\[
  \Pr(Y_{it}=1 \mid \bfDelta\bfX_{it}, \bfZ_i, \mathbf{u}_i)
  =
  G\!\left(o_{it} + \bfDelta\bfX_{it}^\top \mathbf{u}_i\right),
  \qquad
  \mathbf{u}_i \sim N(\mathbf{0}, \boldsymbol{\Sigma}_u).
\]
In the code base this estimator is implemented with \texttt{glmer} as a logistic mixed model with respondent random slopes and Laplace approximation ($nAGQ=0$).  The resulting best-linear-unbiased-predictor (BLUP)-style estimate is $\hat\bfbeta_i^{\text{glmer}} = \hat f_{\text{ens}}(\bfZ_i) + \hat{\mathbf{u}}_i$.

We benchmark this DNN-offset mixed-logit estimator against \texttt{EnsC5} on a paired simulation of $980$ cells spanning the same DGP families used to calibrate $c_\eta$ above, with $N \in \{1{,}000,2{,}500,5{,}000,10{,}000\}$, $T \in \{3,5,10,15\}$, and $p \in \{20,30\}$.  Two findings undermine the mixed-logit approach as a default.  First, the \texttt{glmer} estimation of $\boldsymbol{\Sigma}_u$ goes singular in $26.5\%$ of cells, and the singularity rate is heavily concentrated at small $T$: with $T=3$ and $N \leq 2{,}500$, $93$--$100\%$ of fits are singular regardless of DGP; with $T=5$ and $N=1{,}000$, the singularity rate is still $20$--$60\%$.  The mixed-logit specification has $p(p+1)/2$ covariance parameters to identify from per-respondent task counts that are typically far too small to do so, and the failures fall in precisely the regimes where respondent-level recovery would be most valuable.  Second, even on the $73.5\%$ clean-fit subset, the average improvement of mixed logit over the unrestricted DNN ensemble is $+0.062$ in individual-$\bfbeta$ correlation, which is below \texttt{EnsC5}'s $+0.065$ on the same cells; with singular fits counted as zero gain (since the mixed-logit BLUP collapses to the offset), the overall mean drops to $+0.059$.  Mixed logit modestly leads on a secondary metric---variance-share correlation---by $+0.02$ to $+0.04$, but loses or ties on $\beta$ root-mean-square error (RMSE) and polarization mean-absolute error (MAE).  The compute cost of \texttt{glmer} is roughly $10$ to $20$ times that of the \texttt{EnsC5} update: at $(N,T,p)=(10{,}000,15,20)$, mixed logit takes about $18$ minutes per cell against under one minute for the full \texttt{EnsC5} update.  In a small appendage of the simulation at the larger $p=30$, $T \in \{10,15\}$ regime that matches the \citet{graham2020democracy} application, mixed logit ties \texttt{EnsC5} on the mean and does not dominate.  Across the full design we therefore find no regime in which mixed logit's gain over \texttt{EnsC5} is large enough to overcome its $26.5\%$ failure rate and order-of-magnitude compute cost, and we adopt \texttt{EnsC5} as the default.

We also benchmark hierarchical Bayes, ridge-style penalized updates, and tempered variants that interpolate between the precision-calibrated \texttt{EnsC5} working covariance and the tighter covariance implied by the mixed model.  The empirical pattern is consistent across these comparisons: \texttt{EnsC5} is the most robust default in the low-task regime, and alternative respondent-level estimators offer at most incremental gains in regimes where $T$ is large, the coefficient dimension is modest, and $\bfZ$ already explains a substantial share of heterogeneity.

\subsection{Attribute Interactions}\label{app:interactions}

The accompanying software implements the interaction extension of Section~\ref{sec:limitations}.  We describe the model, the two forms the software provides, how the interaction term enters estimation, and why formal inference on it requires richer designs than our applications afford.

\paragraph{Pairwise-interaction generalization.}  Write the profile utility with population-level pairwise interactions as
\begin{equation}\label{eq:int_utility}
  U_{ijt}
  =
  \bfX_{ijt}^\top \bfbeta_i
  +
  \sum_{1\leq k<\ell\leq p}
  \gamma_{k\ell} X_{ijtk}X_{ijt\ell}
  +
  \varepsilon_{ijt},
\end{equation}
where $\gamma_{k\ell}$ is the coefficient on the product of attribute levels $k$ and $\ell$.  Differencing the two profiles within a task, as in the additive model, the logit choice probability becomes
\begin{equation}\label{eq:int_index}
  \Pr(Y_{it}=1 \mid \bfbeta_i, \boldsymbol{\gamma})
  =
  G\!\left(
    \bfDelta\bfX_{it}^\top \bfbeta_i
    +
    \sum_{1\leq k<\ell\leq p}
    \gamma_{k\ell}
    \{X_{i1tk}X_{i1t\ell}-X_{i2tk}X_{i2t\ell}\}
  \right).
\end{equation}
The interaction enters as a difference of profile-level product terms, not as a product of the profile contrast $\bfDelta\bfX_{it}$. This choice is not merely notational. A forced choice requires that swapping the two profiles send the choice probability to its complement, and the difference of profile-level products has exactly this antisymmetry: the swap negates both $\bfDelta\bfX_{it}$ and the interaction term. A quadratic in the contrast alone, $\|\mathbf{V}^\top \bfDelta\bfX_{it}\|^2$, is invariant under the swap and cannot arise from any profile-separable utility.

\paragraph{Explicit and low-rank forms.}  The explicit form estimates the identified pairwise coefficients $\gamma_{k\ell}$ under a ridge penalty $\lambda\sum_{k<\ell}\gamma_{k\ell}^{2}$. With dummy-coded attributes, products of levels within the same attribute are identically zero, so only cross-attribute products survive. For the candidate application, $59$ of the nominal $\binom{13}{2}=78$ pairwise products are identified; the remaining $19$ are within-attribute products. The low-rank form parameterizes the pairwise coefficients as $\gamma_{k\ell}=2\mathbf{v}_k^\top\mathbf{v}_\ell$, with rows $\mathbf{v}_k^\top$ of $\mathbf{V}\in\mathbb{R}^{p\times r}$ and $r\ll p$. Equivalently, up to diagonal terms $\|\mathbf{v}_k\|^2X_k^2$ that are collinear with main effects for dummy-coded attributes, the profile-level interaction can be written as $\|\mathbf{V}^\top\bfX\|^2$. The matrix $\mathbf{V}$ is an additional output head of the network penalized by $\lambda_V\|\mathbf{V}\|_F^2$. Two caveats attach to this form. First, $\mathbf{V}$ is fixed only up to rotation, so substantive interpretation should be attached to the implied pairwise coefficients $\gamma_{k\ell}$ rather than to the factors themselves. Second, the low-rank representation $\gamma_{k\ell}=2\mathbf{v}_k^\top\mathbf{v}_\ell$ restricts the interaction matrix; a fully indefinite interaction structure requires a signed factorization such as $\mathbf{V}_+\mathbf{V}_+^\top-\mathbf{V}_-\mathbf{V}_-^\top$.

\paragraph{Estimation.}  The interaction coefficient vector $\boldsymbol{\gamma}$ is a population-level object. Respondent-specific interaction coefficients are not feasible in the few-task designs we study, so $\boldsymbol{\gamma}$ is shared across respondents and the individual heterogeneity stays in $\bfbeta_i$. It enters the respondent-level empirical-Bayes update of Supplementary Materials~\ref{app:estimation_details} as a known per-task offset,
\begin{equation}\label{eq:int_offset}
  o_{it}
  =
  \sum_{1\leq k<\ell\leq p}
  \gamma_{k\ell}
  \{X_{i1tk}X_{i1t\ell}-X_{i2tk}X_{i2t\ell}\},
\end{equation}
added to the index $\bfDelta\bfX_{it}^\top \bfbeta_i$ before the logit; the prior, the Laplace update, and the scale estimation are unchanged, and the recovered $\hat{\bfbeta}_i$ keep their reading as main-effect preferences at the reference profile. We use this interaction extension as a plug-in robustness check in the applications. Formal debiased inference for interaction-bearing targets would require replacing $\bfDelta\bfX_{it}$ with the expanded contrast $\widetilde{\bfDelta\bfX}_{it}=(\bfDelta\bfX_{it},\,\mathbf{q}_{i1t}-\mathbf{q}_{i2t})$, with $\mathbf{q}_{ijt}$ the vector of identified pairwise products of profile $j$, and estimating the expanded local information matrix
\begin{equation}\label{eq:int_lambda}
  \widetilde{\bfLambda}_0(\bfZ)
  = \E\!\big[G'(\cdot)\, \widetilde{\bfDelta\bfX}\, \widetilde{\bfDelta\bfX}^\top \mid \bfZ\big],
\end{equation}
of dimension $(p+q)\times(p+q)$, with $q$ the number of identified interactions ($q=59$ here). The cost of the extension is concentrated in~\eqref{eq:int_lambda}. Each cross-attribute product is nonzero in only a small share of tasks, so at small $T$ the expanded information matrix is estimated too imprecisely to invert stably, and the debiased intervals for interaction-bearing functionals are vacuous. This is a property of the design, not of the estimator, and it matches the design guidance of Supplementary Materials~\ref{app:factorial_guidance}, which places interaction-level targets well beyond a $T=3$ design.

We assess the extension on the candidate application in Supplementary Materials~\ref{app:more_figures} (Table~\ref{tab:sw_interactions}).  Under both forms the substantive findings are essentially unchanged, while formal inference from the expanded interaction specification is uninformative at $T = 3$, exactly as~\eqref{eq:int_lambda} anticipates.

\FloatBarrier

\subsection{Reproducibility}\label{app:reproducibility}

All estimates use seed $42$, with a two-seed ensemble (seeds $42$ and $7$) for the mean stage; headline quantities shift by at most $0.05$ on the logit scale across seeds.  The deep network is fit with \texttt{torch} under the configuration recorded in Supplementary Materials~\ref{app:estimation_details}, and the accompanying \ifanon R package\else \texttt{sconjoint} package\fi reproduces the reported quantities from the analysis data.  A full replication archive---data-preparation scripts, fitted objects, and figure code---will accompany the published paper.

\section{Asymptotic Theory}\label{sec:asymptotics} \label{app:asymptotic_results}

In this section we study the high-level consequences of standard orthogonal-score and M-estimation arguments. We do not verify primitive neural-network approximation or entropy conditions. Instead, we assume that the first-stage learners satisfy the rate, stability, and remainder conditions that such primitive assumptions are meant to deliver.

Let $N$ denote the number of respondents and $n=\sum_{i=1}^N T_i$ the number of choice observations.  The independent sampling units are respondents.  Unless otherwise noted, $p$ is fixed and $T_i\leq \bar T$.  The average-parameter results are stated for respondent-weighted averages.  When $T_i=T$ is constant, these coincide with task-weighted averages.  If $T_i$ varies, the task-weighted estimator targets a different population functional in general; the two targets agree only when $T_i$ is independent of $f_0(\bfZ_i)$.

Let $m_0(\bfZ_i)=\E[\bfbeta_i\mid \bfZ_i]$ denote the conditional mean of the latent random coefficients under the structural random-coefficients model. For the average-parameter results, we take the maintained mean-logit condition as a primitive:
\[
  \E[Y_{it}\mid \bfDelta\bfX_{it},\bfZ_i]
  =
  G\!\left(\bfDelta\bfX_{it}^{\top}m_0(\bfZ_i)\right).
\]
Under this maintained structural condition, the mean-stage target equals the conditional mean, $f_0=m_0$. If the same estimating equations are read only as a working conditional-logit approximation, then $f_0$ is instead the projection target defined in (H2).  For the average-parameter result, the target is
\[
  \theta_{0k} = \E[f_{0k}(\bfZ_i)],
\]
which equals $\E[\beta_{ik}]=\E[m_{0k}(\bfZ_i)]$ under the maintained structural model and is the average conditional-logit projection under the weaker working-model interpretation.

Let $\psi_{ikt}(f,\Lambda)$ denote the debiased signal whose feasible plug-in is given in~\eqref{eq:psi}, and define the mean-zero score
\[
  \varphi_{ikt}(\theta,f,\Lambda)
  =
  \psi_{ikt}(f,\Lambda)-\theta.
\]
The feasible respondent-weighted average-parameter estimator is $\hat\theta_k$ of~\eqref{eq:theta_hat}.

Define the respondent-level conditional log-likelihood
\[
  \ell_i(\bfbeta)
  =
  \sum_{t=1}^{T_i}
  \left[
    Y_{it}\log G(\bfDelta\bfX_{it}^\top \bfbeta)
    +
    (1-Y_{it})\log\!\bigl(1-G(\bfDelta\bfX_{it}^\top \bfbeta)\bigr)
  \right].
\]
For a positive-definite working covariance matrix $\boldsymbol{\Sigma}_{\eta,0}$, define the \emph{oracle} respondent-level penalized mode
\[
  \bfbeta_{i,T_i}^\star
  =
  \operatorname*{arg\,max}_{\bfbeta \in \mathbb{R}^p}
  \left\{
    \ell_i(\bfbeta)
    -
    \frac{1}{2}
    (\bfbeta-f_0(\bfZ_i))^\top
    \boldsymbol{\Sigma}_{\eta,0}^{-1}
    (\bfbeta-f_0(\bfZ_i))
  \right\}.
\]
Under the maintained random-coefficients prior interpretation, this is the oracle posterior mode.  Under a working-model interpretation of $f_0$, it is the oracle penalized target induced by the limiting first-stage mean and working covariance.

Finally, define
\[
  R_{Z,k}
  =
  \frac{\Var(m_{0k}(\bfZ_i))}{\Var(\beta_{ik})},
\]
the share of latent random-coefficient heterogeneity in coefficient $k$ explained by observed respondent characteristics. This is a structural-model object. It is defined in terms of the latent coefficient $\beta_{ik}$ and its conditional mean $m_{0k}(\bfZ_i)$, and it is not defined under the weaker projection interpretation unless one separately specifies a latent random-coefficients law.

The propositions below invoke the following high-level conditions as needed:
\begin{enumerate}[label=(H\arabic*),leftmargin=*]
  \item Respondent clusters are i.i.d.; all tasks from a respondent are kept in the same cross-fitting fold; profile contrasts are randomly assigned conditional on $\bfZ_i$; the respondent-level score contribution $T_i^{-1}\sum_{t=1}^{T_i}\varphi_{ikt}(\theta_{0k},f_0,\boldsymbol{\Lambda}_0)$ has a finite $2+\delta$ moment for some $\delta>0$; and its variance $V_k=\Var(\phi_{ik})$ is strictly positive.
  \item The mean-stage target $f_0$ is the unique square-integrable solution to
  \[
    \E\!\left[
      \bfDelta\bfX_{it}
      \left\{
        Y_{it}
        -
        G\!\left(\bfDelta\bfX_{it}^\top f_0(\bfZ_i)\right)
      \right\}
      \;\middle|\;
      \bfZ_i
    \right]
    =
    \mathbf{0}.
  \]
  Under the maintained mean-logit primitive stated above, $f_0(\bfZ_i)=m_0(\bfZ_i)=\E[\bfbeta_i\mid\bfZ_i]$; under the weaker working-model interpretation, $f_0$ is the conditional-logit projection.
  \item The local information matrix
  \[
    \boldsymbol{\Lambda}_0(\bfZ)
    =
    \E\!\left[
      G'\!\left(\bfDelta\bfX_{it}^\top f_0(\bfZ)\right)
      \bfDelta\bfX_{it}\bfDelta\bfX_{it}^\top
      \;\middle|\;
      \bfZ
    \right]
  \]
  has eigenvalues bounded away from zero and infinity uniformly on the support of $\bfZ$.
  \item The cross-fitted nuisance estimates satisfy
	  \[
	    \|\hat f-f_0\|_{L^2}=o_p(N^{-1/4}),
	    \qquad
	    \|\hat{\boldsymbol{\Lambda}}-\boldsymbol{\Lambda}_0\|=o_p(1),
	  \]
	  and the resulting respondent-averaged orthogonal-score remainder is $o_p(N^{-1/2})$. (For the average-parameter score, this high-level remainder condition can be replaced by the usual DML product-rate and stability conditions that make products of nuisance-estimation errors asymptotically negligible. We state the condition at this level because we do not verify primitive neural-network entropy or approximation conditions here.)
  \item For the respondent-level MAP results, $\hat{\boldsymbol{\Sigma}}_\eta\to_p\boldsymbol{\Sigma}_{\eta,0}$ with eigenvalues bounded away from zero and infinity, $\hat f(\bfZ_i)=O_p(1)$ for respondent sequences under consideration, the relevant penalized objective has a unique well-separated maximizer, and the numerical procedure returns the global maximizer.  For the fixed-$T$ result below, additionally $\hat f(\bfZ_i)\to_p f_0(\bfZ_i)$ for the respondent under consideration.
  \item For the large-$T$ results, along a sequence with $T_i\to\infty$, within-respondent tasks satisfy a uniform law of large numbers, the expected respondent-level log-likelihood is uniquely maximized at the true $\bfbeta_i$, and the Fisher information at $\bfbeta_i$ is positive definite.
\end{enumerate}

Our first result concerns the population-average estimand $\theta_{0k}$, where pooling across the $N$ respondents drives identification and within-respondent information enters only through the cluster structure of the variance.

\begin{proposition}
\label{prop:app_avg}
Under (H1)--(H4), for each fixed coordinate $k$,
\[
  \sqrt{N}\,(\hat\theta_k-\theta_{0k})
  =
  \frac{1}{\sqrt{N}}
  \sum_{i=1}^{N}
  \phi_{ik}
  +
  o_p(1),
\]
where
\[
  \phi_{ik}
  =
  \frac{1}{T_i}
  \sum_{t=1}^{T_i}
  \varphi_{ikt}(\theta_{0k},f_0,\boldsymbol{\Lambda}_0),
  \qquad
  \E[\phi_{ik}]=0,
  \qquad
  V_k=\Var(\phi_{ik})\in(0,\infty).
\]
Consequently,
\[
  \sqrt{N}\,(\hat\theta_k-\theta_{0k})
  \;\xrightarrow{d}\;
  N(0,V_k),
\]
where $N$ is the number of respondents.  The respondent-clustered variance estimator based on the centered empirical influence contributions
\[
  \hat\phi_{ik}
  =
  \frac{1}{T_i}
  \sum_{t=1}^{T_i}
  \hat\psi_{ikt}
  -
  \hat\theta_k
\]
yields
\[
  \hat V_k
  =
  \frac{N}{N-1}\cdot\frac{1}{N}
  \sum_{i=1}^{N}\hat\phi_{ik}^2
  \to_p V_k.
\]
Equivalently, $\widehat{\Var}(\hat\theta_k)=\hat V_k/N$.  Under a balanced fixed-$T$ sequence with $T_i=T$ for all $i$, we have $n=NT$ and the equivalent $\sqrt{n}$ normalization follows by rescaling the asymptotic variance by the constant factor $T$.
\end{proposition}

\begin{proof}
This is the standard DML conclusion under cross-fitting and Neyman orthogonality.  Given respondent-level randomization, bounded cluster size, and the $o_p(N^{-1/4})$ first-stage rate, the mean-zero orthogonal moment $\varphi$ yields first-order insensitivity to estimation error in $\hat f$, giving the asymptotically linear representation above.  The respondent-clustered central limit theorem then applies to the independent respondent-level contributions $\phi_{ik}$; see \citet{chernozhukov2018double} and \citet{farrell2025deep}.
\end{proof}

We now turn from the average to the respondent-level estimator $\hat{\bfbeta}_i^{\mathrm{MAP}}$, where within-respondent variation becomes the central object.  Two questions arise: with a fixed number of tasks per respondent, what does the MAP step recover; and what changes as within-respondent information accumulates?  Our next result addresses the fixed-$T$ case.

\begin{proposition}
\label{prop:app_fixedT}
Under (H1), (H2), and (H5)---including the additional pointwise convergence requirement $\hat f(\bfZ_i)\to_p f_0(\bfZ_i)$ in (H5)---for any respondent $i$ with fixed $T_i$,
\[
  \hat{\bfbeta}_i^{\mathrm{MAP}} - \bfbeta_{i,T_i}^\star = o_p(1)
  \qquad \text{as } N \to \infty.
\]
Thus, with a fixed number of tasks per respondent, the feasible hybrid estimator converges to the oracle penalized mode induced by the limiting prior mean and limiting working covariance, not necessarily to the true $\bfbeta_i$ itself.  The respondent-level prior penalty continues to exert asymptotically non-negligible influence when $T_i$ remains fixed.
\end{proposition}

\begin{proof}
The feasible respondent-level objective is a plug-in version of the oracle objective.  Under (H5), the plug-in prior mean at $\bfZ_i$ and the working covariance converge in probability to their population limits.  The statement is conditional on the respondent's finite task sequence, with cross-fitting ensuring that the nuisance estimates used for that respondent are trained on other respondents. Since $T_i$ is fixed, the respondent-level log-likelihood is a finite continuous function of $\bfbeta$.  The unique well-separated maximizer condition permits the argmax continuous mapping theorem, yielding $\hat{\bfbeta}_i^{\mathrm{MAP}} \to_p \bfbeta_{i,T_i}^\star$.
\end{proof}

Proposition~\ref{prop:app_fixedT} shows that with $T_i$ fixed, the respondent-level MAP estimator remains anchored to the prior even as $N\to\infty$: the limit is the oracle penalized mode, not the respondent's true coefficient $\bfbeta_i$.  The next result turns this around by letting within-respondent information accumulate.  Once $T_i$ grows, the respondent-level log-likelihood scales linearly in $T_i$ while the prior penalty does not, so the likelihood eventually dominates and the feasible MAP estimator becomes consistent for $\bfbeta_i$.

\begin{proposition}
\label{prop:app_largeT}
Under (H1), (H5), and (H6), along a sequence with $T_i\to\infty$, define the feasible hybrid estimator by
\[
  \hat{\bfbeta}_i^{\mathrm{MAP}}
  =
  \operatorname*{arg\,max}_{\bfbeta \in \mathbb{R}^p}
  \left\{
    \ell_i(\bfbeta)
    -
    \frac{1}{2}
    (\bfbeta-\hat f(\bfZ_i))^\top
    \hat{\boldsymbol{\Sigma}}_\eta^{-1}
    (\bfbeta-\hat f(\bfZ_i))
  \right\}.
\]
Then
\[
  \hat{\bfbeta}_i^{\mathrm{MAP}} - \bfbeta_i = o_p(1).
\]
\end{proposition}

\begin{proof}
Divide the penalized objective by $T_i$:
\[
  \frac{1}{T_i}\ell_i(\bfbeta)
  -
  \frac{1}{2T_i}
  (\bfbeta-\hat f(\bfZ_i))^\top
  \hat{\boldsymbol{\Sigma}}_\eta^{-1}
  (\bfbeta-\hat f(\bfZ_i)).
\]
On compact sets, the penalty term is $O_p(T_i^{-1})$ uniformly by (H5), whereas the average log-likelihood converges uniformly to its population limit by (H6).  The normalized objective therefore has the same asymptotic limit as the respondent-level likelihood objective, whose unique maximizer is $\bfbeta_i$.  The argmax theorem gives $\hat{\bfbeta}_i^{\mathrm{MAP}}\to_p\bfbeta_i$.
\end{proof}

Propositions~\ref{prop:app_fixedT} and~\ref{prop:app_largeT} together describe what the respondent-level estimator can recover when both observables in $\bfZ_i$ and respondent-specific choice data are available.  The final proposition isolates a complementary question: without using respondent-specific choices, how informative is $\bfZ_i$ on its own?  Any predictor of $\beta_{ik}$ built from $\bfZ_i$ alone is bounded in its correlation with the latent coefficient by a fixed population quantity, with the bound saturated by the conditional mean $m_{0k}(\bfZ_i)$.

\begin{proposition}
\label{prop:app_corr_ceiling}
Let $m_{0k}(\bfZ_i)=\E[\beta_{ik}\mid \bfZ_i]$, and let $g_k(\bfZ_i)$ be any square-integrable predictor measurable with respect to $\bfZ_i$ with $\Var(g_k(\bfZ_i))>0$.  Assume $\Var(\beta_{ik})>0$.  Then
\[
  \operatorname{Corr}\!\bigl(g_k(\bfZ_i),\beta_{ik}\bigr)^2
  \leq
  R_{Z,k}.
\]
If $\Var(m_{0k}(\bfZ_i))>0$, then
\[
  \operatorname{Corr}\!\bigl(m_{0k}(\bfZ_i),\beta_{ik}\bigr)^2
  =
  \frac{\Var(m_{0k}(\bfZ_i))}{\Var(\beta_{ik})}
  =
  R_{Z,k}.
\]
In that case, equality holds if and only if $g_k(\bfZ_i)=a+b\,m_{0k}(\bfZ_i)$ almost surely for some $a\in\mathbb{R}$ and $b\neq 0$.  If $\Var(m_{0k}(\bfZ_i))=0$, then $R_{Z,k}=0$ and every square-integrable $Z$-measurable predictor has zero covariance with $\beta_{ik}$.
\end{proposition}

\begin{proof}
Since $g_k(\bfZ_i)$ is measurable with respect to $\bfZ_i$, the law of iterated expectations implies
\[
  \Cov(g_k(\bfZ_i),\beta_{ik})
  =
  \Cov(g_k(\bfZ_i),m_{0k}(\bfZ_i)).
\]
Therefore,
\[
  \operatorname{Corr}\!\bigl(g_k(\bfZ_i),\beta_{ik}\bigr)^2
  =
  \frac{
    \Cov(g_k(\bfZ_i),m_{0k}(\bfZ_i))^2
  }{
    \Var(g_k(\bfZ_i))\Var(\beta_{ik})
  }.
\]
Applying the Cauchy--Schwarz inequality,
\[
  \Cov(g_k(\bfZ_i),m_{0k}(\bfZ_i))^2
  \leq
  \Var(g_k(\bfZ_i))\,\Var(m_{0k}(\bfZ_i)).
\]
This yields
\[
  \operatorname{Corr}\!\bigl(g_k(\bfZ_i),\beta_{ik}\bigr)^2
  \leq
  \frac{
    \Var(m_{0k}(\bfZ_i))
  }{
    \Var(\beta_{ik})
  }.
\]
If $\Var(m_{0k}(\bfZ_i))>0$, then
\[
  \Cov(m_{0k}(\bfZ_i),\beta_{ik})
  =
  \Var(m_{0k}(\bfZ_i)),
\]
so
\[
  \operatorname{Corr}\!\bigl(m_{0k}(\bfZ_i),\beta_{ik}\bigr)^2
  =
  \frac{
    \Var(m_{0k}(\bfZ_i))
  }{
    \Var(\beta_{ik})
  }.
\]
When $\Var(m_{0k}(\bfZ_i))>0$, equality in Cauchy--Schwarz holds if and only if $g_k(\bfZ_i)$ is almost surely affine in $m_{0k}(\bfZ_i)$.  When $\Var(m_{0k}(\bfZ_i))=0$, the same covariance identity gives $\Cov(g_k(\bfZ_i),\beta_{ik})=0$ for every square-integrable $Z$-measurable predictor.
\end{proof}

These propositions formalize the central asymptotic distinction in the paper. Large $N$ identifies and estimates the population-average structural mean $\E[m_{0k}(\bfZ_i)]=\E[\beta_{ik}]$ under the maintained mean-logit specification; in the more general working-model notation of the propositions, the same result is stated for the population-average conditional-logit projection $\E[f_{0k}(\bfZ_i)]$. With fixed $T$, the respondent-level estimator converges only to an oracle penalized summary. Large $T$ is what turns the respondent-level update into a consistent estimator of the individual preference vector itself because the respondent-level likelihood dominates the prior penalty.

The $R_{Z,k}$ result describes a different object: the best possible correlation between the latent coefficient and any predictor based only on observed respondent characteristics.  It is a fixed superpopulation explained-variance ratio under the $N$ and $T$ asymptotic sequences above, not a quantity that generically converges to one.  A sequence with $R_{Z,k}\to 1$ would require a separate assumption, such as $\E[\Var(\beta_{ik}\mid\bfZ_i)]/\Var(\beta_{ik})\to0$, that residual heterogeneity in $\beta_{ik}$ conditional on $\bfZ_i$ vanishes.  By contrast, the respondent-level MAP propositions concern estimators that also use respondent-specific choice data.  In particular, Proposition~\ref{prop:app_largeT} is not a $Z$-only prediction result and does not imply $R_{Z,k}\to1$.

\FloatBarrier

\section{Debiased Inference}\label{app:other_estimands}

This appendix develops orthogonal-score inference for the population-average preference parameters and for the broader set of smooth functionals reported in the applications.  We first present the orthogonal score and respondent-clustered variance estimator for the average parameters $\theta_k$, then extend the same cross-fitted construction to nonlinear functionals such as average marginal effects, counterfactual win probabilities, attribute-importance shares, and marginal rates of substitution.

\subsection{Orthogonal score and clustered inference}

For population-average parameters, the maintained structural target is $\E[f_k(\bfZ_i)]=\E[\beta_{ik}]$; the high-level orthogonal-score result writes the target in its working-model form,
\[
  \theta_k = \E[f_{0k}(\bfZ_i)],
\]
which equals the structural target when the maintained mean-logit model holds.
The population local information matrix is
\[
  \Lambda_{0,jk}(\bfZ)
  =
  \E\!\left[
    G'\!\left(\bfDelta\bfX^\top f_0(\bfZ)\right)\Delta X_j \Delta X_k
    \;\middle|\;
    \bfZ
  \right],
\]
with $G'(v)=G(v)(1-G(v))$.  The orthogonal score $\psi_{ikt}$ and the debiased estimator $\hat\theta_k=N^{-1}\sum_i T_i^{-1}\sum_t\psi_{ikt}$ are the feasible quantities of~\eqref{eq:psi} and~\eqref{eq:theta_hat}: the plug-in DNN mean $\hat f_k(\bfZ_i)$ plus the influence-function correction $[\hat{\bfLambda}^{-1}(\bfZ_i)\,\bfDelta\bfX_{it}\,(Y_{it}-\hat G_{it})]_k$ formed from the logit residual and the inverse local information matrix.

The key orthogonality property is
\[
  \frac{\partial}{\partial r}
  \E\!\left[\psi_{ikt}(\theta_k^0, f_0 + r\,\mathbf{h})\right]\bigg|_{r=0}
  =
  0
  \qquad
  \text{for all perturbation directions } \mathbf{h},
\]
which implies that first-order errors in the DNN mean stage do not propagate into $\hat\theta_k$.  We estimate the variance by clustering at the respondent level, using the centered respondent-level contributions $\hat\phi_{ik} = T_i^{-1}\sum_{t=1}^{T_i}\psi_{ikt} - \hat\theta_k$:
\[
  \widehat{\Var}(\hat{\theta}_k)
  =
  \frac{1}{N(N-1)}
  \sum_{i=1}^{N}
  \hat\phi_{ik}^{\,2}
  =
  \frac{\hat V_k}{N},
  \qquad
  \hat V_k = \frac{N}{N-1}\cdot\frac{1}{N}\sum_{i=1}^{N}\hat\phi_{ik}^{\,2}.
\]
Centering each respondent's averaged score by $\hat\theta_k$ is required because $\psi_{ikt}$ has mean $\theta_k \neq 0$; this matches the respondent-clustered estimator of Proposition~\ref{prop:app_avg}.  Under the regularity conditions in \citet{farrell2025deep} and \citet{chernozhukov2018double}, this yields $\sqrt{N}$-consistent asymptotically normal inference for the average preference parameters, where $N$ is the number of respondents (the independent sampling units).

\subsection{Debiased inference for additional quantities of interest}

The orthogonal score in~\eqref{eq:psi} debiases the linear mean-stage functional, written in the general notation as
$\theta_k=\E[f_{0k}(\bfZ_i)]$; under the maintained mean-logit model this is $\E[f_k(\bfZ_i)]=\E[\beta_{ik}]$.  Many quantities of interest are smooth
\emph{nonlinear} functionals of the same mean stage $f_0$, and they inherit
$\sqrt{N}$-consistent, asymptotically normal inference from the \emph{same}
cross-fitted nuisances by composing~\eqref{eq:psi} with a gradient.  Let $W_i$
collect the randomized design data for respondent $i$ and let
$H(\cdot,W):\mathbb{R}^p\to\mathbb{R}$ be continuously differentiable in its
first argument.

\begin{proposition}
\label{prop:app_master}
Let $\theta_H=\E[H(f_0(\bfZ_i),W_i)]$ and define the task-level signal
\[
  m^{H}_{it}(f,\bfLambda)
  = H\!\big(f(\bfZ_i),W_i\big)
  + \nabla_f H\!\big(f(\bfZ_i),W_i\big)^{\top}
    \bfLambda^{-1}(\bfZ_i)\,\bfDelta\bfX_{it}\,\big(Y_{it}-G_{it}\big),
\]
with $G_{it}=G(\bfDelta\bfX_{it}^{\top}f(\bfZ_i))$, and define the centered score
\[
  \psi^{H}_{it}(\theta,f,\bfLambda)
  =
  m^{H}_{it}(f,\bfLambda)-\theta.
\]
The population information matrix is $\bfLambda_0(\bfZ)=\E[G'(\bfDelta\bfX^{\top}f_0(\bfZ))\,\bfDelta\bfX\bfDelta\bfX^{\top}\mid\bfZ]$.
Suppose (H1)--(H4) hold, $H$ is $C^1$ with locally Lipschitz gradient and finite
envelope, and \textup{(HG)}
$$\Cov\big(\nabla_f H(f_0(\bfZ_i),W_i),\,G'_{it}\bfDelta\bfX_{it}\bfDelta\bfX_{it}^{\top}\mid\bfZ_i\big)=\mathbf{0}.$$
Then $\E[m^{H}_{it}(f_0,\bfLambda_0)]=\theta_H$, $\E[\psi^{H}_{it}(\theta_H,f_0,\bfLambda_0)]=0$, the centered moment is Neyman-orthogonal in
$(f,\bfLambda)$, and
\[
  \hat\theta_H=\frac{1}{N}\sum_{i=1}^{N}\frac{1}{T_i}\sum_{t=1}^{T_i}
  \Big[H(\hat f(\bfZ_i),W_i)+\nabla_f H(\hat f(\bfZ_i),W_i)^{\top}
  \hat\bfLambda^{-1}(\bfZ_i)\bfDelta\bfX_{it}(Y_{it}-\hat G_{it})\Big]
\]
satisfies $\sqrt{N}(\hat\theta_H-\theta_H)\xrightarrow{d}N(0,V_H)$,
$V_H=\Var(\phi_i)$, $\phi_i=T_i^{-1}\sum_t\psi^{H}_{it}(\theta_H,f_0,\bfLambda_0)$,
estimated consistently by the respondent-clustered
$\hat V_H=\tfrac{N}{N-1}\tfrac1N\sum_i\hat\phi_i^{\,2}$, with $\hat\phi_i=T_i^{-1}\sum_t\psi^{H}_{it}(\hat\theta_H,\hat f,\hat\bfLambda)$.
\end{proposition}

Note that setting $H(f,W)=f_k$ (so $\nabla_f H=\mathbf{e}_k$, and
(HG) holds trivially) returns exactly~\eqref{eq:psi} and
Proposition~\ref{prop:app_avg}.  When $H$ is nonlinear, the second-order
$f$--$\bfLambda$ remainder carries a non-residual cross term not removed by the
defining moment of $f_0$, so $\sqrt{N}$-normality additionally requires the
\emph{product} rate $\|\hat f-f_0\|_{L_2}\,\|\hat\bfLambda-\bfLambda_0\|=o_p(N^{-1/2})$,
i.e.\ $\|\hat\bfLambda-\bfLambda_0\|=o_p(N^{-1/4})$.  (The same cross term is
present for $\theta_k$ and is absorbed in the catch-all remainder clause of condition (H4).)

Next, we make a series of observations that follow from the result above, which are all instances of the same estimator.  For a target $\theta_H=\E[H(f_0(\bfZ_i),W_i)]$, with $W_i$ collecting the relevant design data (the profile randomization law, a fixed contest contrast, or a design second-moment block), the feasible estimator $\hat\theta_H$ averages over respondents and tasks the plug-in $H(\hat f(\bfZ_i),W_i)$ together with the one-step correction $\nabla_f H(\hat f(\bfZ_i),W_i)^{\top}\hat\bfLambda^{-1}(\bfZ_i)\bfDelta\bfX_{it}(Y_{it}-\hat G_{it})$, which re-weights the average-parameter influence term of~\eqref{eq:psi} by $\nabla_f H$; its standard error is the respondent-clustered $\widehat{\mathrm{SE}}(\hat\theta_H)=\sqrt{\hat V_H/N}$, with $\hat V_H=\tfrac{N}{N-1}\tfrac1N\sum_i\hat\phi_i^{2}$ and $\hat\phi_i=T_i^{-1}\sum_t\psi^{H}_{it}(\hat\theta_H,\hat f,\hat\bfLambda)$.  The cross-fitted nuisances $(\hat f,\hat\bfLambda)$ are estimated once and reused; only the functional $H$ and its gradient $\nabla_f H$ change across the quantities below, which the observations below take up in turn.

The four observations correspond to the smooth mean-stage quantities used in the paper. AMEs connect the structural logit scale back to the AMCE tradition; counterfactual win probabilities are the electoral-comparison quantities used in the applications; attribute-importance shares quantify which blocks of attributes drive systematic utility variation; and population MRS/WTP are smooth tradeoff ratios for which the average-parameter scores deliver standard delta-method inference. This is why they admit the orthogonal-score treatment below, while polarization fractions, total importance shares, and compensating-differential fractions are separated as distributional functionals.
\begin{itemize}[leftmargin=1.5em,noitemsep]
  \item \emph{Average marginal effect}: let $\mathcal Q_k$ be the actual profile-pair randomization law for the remaining admissible contrast vector $\mathbf{D}_{-k}$ after fixing the level-$k$ contrast to switch from the reference/off condition to level $k$, respecting mutually exclusive levels, reference categories, and any design restrictions. Then
  \[
    H_k(f)=\int\!\left[
      G\!\left(f_k+\mathbf{d}_{-k}^{\top}f_{-k}\right)
      -
      G\!\left(\mathbf{d}_{-k}^{\top}f_{-k}\right)
    \right]d\mathcal Q_k(\mathbf{d}_{-k});
  \]
  the gradient has $a_k=\int G'_{\mathrm{on}}\,d\mathcal Q_k$ and $a_j=\int d_j(G'_{\mathrm{on}}-G'_{\mathrm{off}})\,d\mathcal Q_k$ for $j\neq k$.
  \item \emph{Counterfactual win probability}: for a fixed contrast $\bfc=\bfX_A-\bfX_B$, $H(f)=G(\bfc^{\top}f)$ with gradient $\nabla_f H=G'(\bfc^{\top}f)\,\bfc$.
  \item \emph{Attribute-importance share}: the systematic numerator $N_a^{\mathrm{sys}}$ from a quadratic score $\psi^{N_a}_{it}$, with the share $N_a^{\mathrm{sys}}/\sum_b N_b^{\mathrm{sys}}$ formed by the delta method.
  \item \emph{Population MRS and WTP}: a function of the average parameters, $g(\theta)=-\theta_j/\theta_k$ with gradient $\nabla g=(-1/\theta_k,\theta_j/\theta_k^{2})^{\top}$, via the delta method on the $\hat\theta$ scores.
\end{itemize}

More precisely, we have: \medskip

\noindent \textbf{Average marginal effect on the probability scale.} 
The estimand is the AME of \S\ref{sec:quantities}, $\theta_H=\mathrm{AME}_k^{f_0}$, with $H=H_k$ defined by the design integral over $\mathcal Q_k$ above and gradient $a(\bfZ)=\nabla_f H_k(f_0(\bfZ))$. Here $G'_{\mathrm{on}}=G'(f_{0k}+\mathbf{d}_{-k}^{\top}f_{0,-k})$ and $G'_{\mathrm{off}}=G'(\mathbf{d}_{-k}^{\top}f_{0,-k})$ inside the $\mathcal Q_k$ integral. Because $a$ is a $Y$-free functional of the known randomization law, it is $\bfZ$-measurable and satisfies (HG) by construction; under (H1)--(H4), the product rate, and a profile pool that grows proportionally to $N$, Proposition~\ref{prop:app_master} gives $\sqrt{N}$-normal inference for $\mathrm{AME}_k^{f_0}$, which equals the nonparametric AMCE under correct specification and the latent-coefficient AME under the maintained mean-logit specification. \medskip

\noindent \textbf{Counterfactual win probability and comparative statics.} 
The estimand is the counterfactual win probability of \S\ref{sec:quantities}, $\theta_H=\pi_{AB}=\E_{\bfZ}[G(\bfc^{\top}f_0(\bfZ_i))]$, with $H$ and gradient as listed above for a fixed contrast $\bfc=\bfX_A-\bfX_B$ independent of the randomized $\bfDelta\bfX_{it}$.  Condition (HG) holds because $\bfc$ is fixed; Proposition~\ref{prop:app_master} (with the product rate) yields $\sqrt{N}$-normal inference for $\pi_{AB}$.  The comparative static $\Delta=\pi_{A'B}-\pi_{AB}$ has the differenced influence $\psi^{\pi'}_{it}-\psi^{\pi}_{it}$, whose respondent-clustered variance nets out the within-respondent covariance of the two contests.  The $\eta$-integrated electorate share $\E[G(\bfc^{\top}\bfbeta_i)]$ differs from $\pi_{AB}$ by a Jensen gap and is not point-identified at fixed $T$. \medskip

\noindent \textbf{Attribute importance: mean-stage share.} 
The estimand is the between-respondent attribute-importance share of \S\ref{sec:quantities}.  Let $\bfS_a=\Var_X(\bfDelta\bfX_{\cdot,a})$ be the (known) design second-moment block of attribute $a$ and $N_a^{\mathrm{sys}}=\E[f_{0,a}(\bfZ_i)^{\top}\bfS_a f_{0,a}(\bfZ_i)]$, with quadratic-functional score
\[
  \psi^{N_a}_{it}=f_a(\bfZ_i)^{\top}\bfS_a f_a(\bfZ_i)
  +2\big(\bfS_a f_a(\bfZ_i)\big)^{\top}
    \big[\bfLambda^{-1}(\bfZ_i)\bfDelta\bfX_{it}(Y_{it}-G_{it})\big]_{g_a}-N_a^{\mathrm{sys}}.
\]
This score is Neyman-orthogonal in $(f,\bfLambda)$ under (H1)--(H4), a strict $o_p(N^{-1/4})$ mean rate, and $\E\|f_0\|^4<\infty$.  The share $\theta^{\mathrm{share,sys}}_a=N_a^{\mathrm{sys}}/\sum_b N_b^{\mathrm{sys}}$ then admits $\sqrt{N}$-normal inference by the delta method, with simplex Jacobian $\bfJ_{ab}=(\mathbf{1}\{a=b\}-\theta^{\mathrm{share,sys}}_a)/\sum_b N_b^{\mathrm{sys}}$ and respondent-clustered standard error $\sqrt{\bfJ^{\top}\hat\bfSigma_N\bfJ/N}$.  The total share, which adds $\operatorname{tr}(\bfS_a\,\E[\Var(\bfeta_{i,a}\mid\bfZ_i)])$, is a functional of the residual law and is not point-identified at fixed $T$ (Proposition~\ref{prop:app_fixedT}). \medskip

\noindent \textbf{Population MRS and willingness to pay at the mean.} 
The estimand is the population MRS and WTP of \S\ref{sec:quantities}, a function of the average parameters (with $g$ and gradient $\nabla g$ as listed above, and $H_{\mathrm{WTP}}(\theta)=-\theta_j/\theta_{\mathrm{money}}$), handled by the delta method on the average-parameter scores.  Stacking the scores~\eqref{eq:psi} for the two coordinates gives the orthogonal influence $\IF_i=\nabla g(\theta_0)^{\top}(\phi_{ij},\phi_{ik})^{\top}$, which is mean-zero because $(\phi_{ij},\phi_{ik})$ is mean-zero and the delta-method gradient is evaluated at the fixed population value $\theta_0$.  Under (H1)--(H4), joint asymptotic normality of $(\hat\theta_j,\hat\theta_k)$, a denominator bounded away from zero, and $V=\nabla g^{\top}\bfSigma\nabla g>0$, $\sqrt{N}\{(-\hat\theta_j/\hat\theta_k)-(-\theta_j/\theta_k)\}\xrightarrow{d}N(0,V)$, with $V$ estimated by $\nabla g^{\top}\hat\bfSigma_{\{j,k\}}\nabla g$ from the respondent-clustered $\hat\bfSigma$.  Near a vanishing denominator a Fieller interval is reported instead.  These ratios of average coefficients are distinct from the average of individual ratios. \medskip

In addition, we have the following remarks: \medskip

\noindent \textbf{Distributional functionals.}  The polarization fraction
$\pi_k=\Pr(\beta_{ik}>0)$, the total importance share, and threshold fractions
such as the compensating differentials are distributional
functionals of the latent law of $\bfbeta_i$.  At fixed $T$ the residual law
$F_{\eta_k\mid\bfZ}$ is not point-identified (Proposition~\ref{prop:app_fixedT}),
and the indicator $\mathbf{1}\{\beta_{ik}>0\}$ is non-smooth, so no
$\sqrt{N}$-orthogonal score exists for the exact target.  A smoothed-CDF score
yields $\sqrt{N}$ inference only for a pre-smoothed surrogate, and only under
the mean-logit specification, an independently estimated residual law, and a
coupled bandwidth/first-stage rate.  We therefore recommend a respondent-cluster
multiplier (wild) bootstrap for these quantities, noting that the fixed-$T$
shrinkage of $\hat\bfbeta_i$ (Proposition~\ref{prop:app_fixedT}) biases the
plug-in toward consensus and is a centering bias the bootstrap does not remove. \medskip

\noindent \textbf{Finite-sample coverage.}  To confirm that the debiased intervals of Proposition~\ref{prop:app_master} and its associated observations above attain their nominal level in finite samples, we simulated coverage for each quantity of interest under the DNN first stage used in the applications.  On a data-generating process whose mean preference vector varies smoothly with three respondent moderators ($N=1{,}000$ respondents, $T=8$ tasks, $300$ replications), the nominal $95\%$ respondent-clustered intervals are well calibrated for every quantity (DNN column of Table~\ref{tab:learner_coverage}). The average parameter, the average marginal effect, the counterfactual vote share and its comparative static, the attribute-importance share, and willingness to pay all lie between $0.92$ and $0.99$, and only the marginal rate of substitution---a ratio of average parameters, and the hardest case in the hierarchy of Supplementary Materials~\ref{sec:asymptotics}---runs a little low at $0.85$.  The same table shows the suite holds across the elastic-net and forest first stages, and Supplementary Materials~\ref{app:thetak_mc} traces the average-parameter coverage across sample sizes and mean-surface shapes.

\subsection{Monte-Carlo coverage for various first-stage learners}\label{app:thetak_mc}

Because the debiased estimator uses the first stage only as a plug-in, the same orthogonal-score machinery accommodates any sufficiently accurate learner; this study verifies that claim and identifies which learner is best calibrated in finite samples.  We ran a dedicated Monte-Carlo study for $\theta_k$.  The displayed grid crosses $N \in \{500,1000,2000\}$ respondents, $T \in \{5,15\}$ tasks per respondent, three shapes for the conditional mean surface $f_0(\bfZ)$ (linear, mild nonlinear, and strong nonlinear), and three flexible first-stage learners: the DNN used in the applications, the spline-expanded elastic net, and the GRF.  Each cell uses $p=5$ preference coordinates, three respondent moderators, $\sigma_\eta=0.4$, five-fold respondent-level cross-fitting, no stage-2 update, and 500 replications.  The estimand is the first average preference coordinate, $\theta_1$; all learners feed the same orthogonal-score and respondent-clustered variance calculation.

Table~\ref{tab:thetak_mc_headline} reports a representative cell at $N=1{,}000$ and $T=15$, and Figure~\ref{fig:thetak_mc_coverage} shows the full coverage surface over $N$.  All three learners deliver broadly valid intervals---confirming that the framework is not tied to any single first stage---but they are not equally calibrated.  The DNN is closest to nominal across the grid and essentially on target in the larger cells, which is why we adopt it as the default throughout the applications.  The spline-expanded elastic net and GRF run somewhat anti-conservative in several cells, especially at smaller or intermediate $N$.  Orthogonalization removes first-order first-stage bias, but the residual finite-sample calibration still depends on how well the learner tracks the conditional-logit projection and on the size of the respondent-level sample.

\begin{table}[H]
\centering
\footnotesize
\caption{Monte-Carlo Coverage for $\theta_1$ at $N=1{,}000$ and $T=15$, by DGP Shape and First-stage Learner}
\label{tab:thetak_mc_headline}
\begin{tabular}{@{}llrrrrr@{}}
\toprule
DGP shape & Learner & Coverage & MCSE & Mean CI width & Mean bias & Bias/SE \\
\midrule
Linear & DNN       & 0.944 & 0.010 & 0.123 &  0.0078 &  0.248 \\
Linear & EN-spline & 0.912 & 0.013 & 0.102 &  0.0144 &  0.550 \\
Linear & GRF       & 0.918 & 0.012 & 0.088 &  0.0117 &  0.522 \\
Mild nonlinear & DNN       & 0.944 & 0.010 & 0.215 &  0.0023 &  0.041 \\
Mild nonlinear & EN-spline & 0.898 & 0.014 & 0.348 & -0.0305 & -0.344 \\
Mild nonlinear & GRF       & 0.924 & 0.012 & 0.103 &  0.0081 &  0.310 \\
Strong nonlinear & DNN       & 0.964 & 0.008 & 0.292 &  0.0110 &  0.148 \\
Strong nonlinear & EN-spline & 0.938 & 0.011 & 0.465 &  0.0157 &  0.132 \\
Strong nonlinear & GRF       & 0.914 & 0.013 & 0.208 & -0.0021 & -0.039 \\
\bottomrule
\end{tabular}
\par\vspace{3pt}
\begin{minipage}[]{1\textwidth}\footnotesize
    \textit{Note:} Each entry is based on 500 Monte-Carlo replications.  Coverage is the fraction of nominal 95\% respondent-clustered Wald intervals containing the known truth for $\theta_1$.  MCSE is the Monte-Carlo standard error of the coverage estimate.  EN-spline denotes the elastic-net first stage with the package's default spline-expanded moderator basis.
\end{minipage}
\end{table}

\begin{figure}[!ht]
  \centering
  \caption{Monte-Carlo Coverage of Nominal 95\% Debiased Intervals for $\theta_1$}
  \label{fig:thetak_mc_coverage}
  \includegraphics[width=\textwidth]{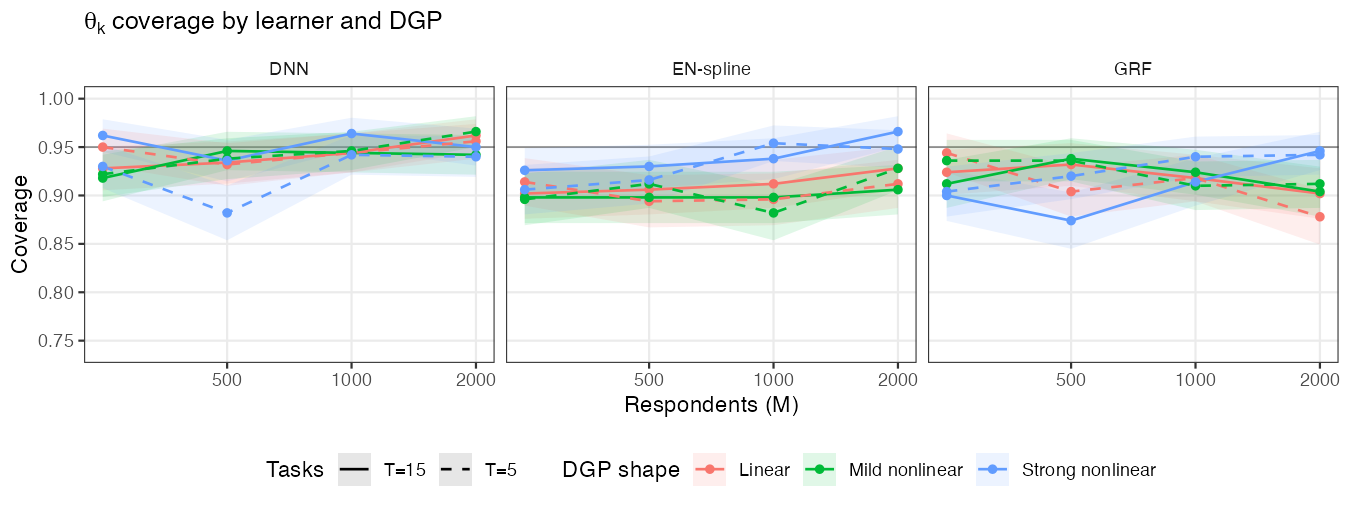}
  \begin{minipage}[]{1\textwidth}\footnotesize
        \textit{Note:} Each point is a 500-replication cell, with the very small $N=250$ cells omitted from the display.  Panels vary the first-stage learner; colors vary the DGP shape for the conditional mean surface $f_0(\bfZ)$; line type varies the number of tasks per respondent.  Shaded bands are $\pm 1.96$ Monte-Carlo standard errors and the horizontal line marks nominal 95\% coverage.  The DNN approaches nominal coverage in the larger cells.  The spline-expanded elastic net and GRF remain somewhat anti-conservative in several cells, indicating finite-sample first-stage or variance-estimation sensitivity rather than a universal learner-invariance result.
    \end{minipage}
\end{figure}

\FloatBarrier

\subsection{Scale identification and a diagnostic}\label{app:scale}

Forced-choice data identify the preference vector only up to a respondent-specific scale: if utility carries a scale $\sigma_i$, the choice likelihood depends on $\bfbeta_i/\sigma_i$ alone, so taste heterogeneity and choice-consistency heterogeneity are not separately identified.  Table~\ref{tab:invariance} records what this normalization leaves intact.  Signs and the direction-based shares built from them, individual marginal rates of substitution, compensating-differential thresholds, and within-respondent importance shares are sign comparisons or ratios internal to a respondent, and are invariant to $\sigma_i$.  Cross-respondent comparisons of coefficient \emph{magnitudes}---the population MRS and willingness to pay in level units, and displays of preference intensity---are not.

\begin{table}[t]
\centering\small
\caption{Scale Invariance of Reported Quantities under a Respondent-specific Utility Scale $\sigma_i$}
\label{tab:invariance}
\begin{tabular}{@{}p{0.62\textwidth} c@{}}
\toprule
Quantity & Scale-invariant? \\
\midrule
Signs of $\hat\beta_{ik}$; direction shares; polarization fractions $\Pr(\beta_{ik}>0)$ & yes \\
Individual MRS and compensating-differential thresholds & yes \\
Within-respondent importance shares & yes \\
Counterfactual win probabilities and vote shares on the normalized choice scale & yes \\
Population-average sign and direction of $\theta_k$ & yes \\
$\theta_k$ magnitudes compared across respondent groups & no \\
Population MRS / WTP in level units; intensity displays & no \\
\bottomrule
\end{tabular}
\begin{minipage}[]{1\textwidth}\footnotesize
    \textit{Note:} Within-respondent ratios and sign-based quantities survive; cross-respondent magnitude comparisons do not. For counterfactual win probabilities and vote shares, the ``yes'' refers to the already normalized logit choice scale; applying an arbitrary additional positive rescaling to the reported coefficient vector changes $G(\bfc^\top\bfbeta)$.
\end{minipage}
\end{table}

As a diagnostic for whether the implied scale varies systematically with observables, we regress a scale proxy, $\log\|\hat\bfbeta_i\|$, on the respondent covariates available in the democracy application.  The proxy depends weakly on ideology and party identification (each significant at the $5\%$ level), but the covariates jointly explain only $2.7\%$ of its variance, so scale heterogeneity along these observed dimensions is modest.  We accordingly state the cross-group comparisons in the main text in scale-invariant terms---signs, fractions, and within-respondent ratios---and read the magnitude displays as descriptive.

\subsection{Three aggregations of a counterfactual contest}\label{app:three_agg}

The counterfactual win probabilities in the main text aggregate the recovered $\hat\bfbeta_i$ through the logit link, $\E[G(\bfDelta\bfX^\top\hat\bfbeta_i)]$. Two alternative aggregations are useful benchmarks. The first is the debiased mean-stage share, $\E[G(\bfDelta\bfX^\top \hat f(\bfZ_i))]$, which replaces each respondent's updated vector with its design-identified mean stage. The second is the majority-preference function, $\Pr_i(\bfDelta\bfX^\top\hat\bfbeta_i > 0)$ \citep{abramson2022voter}, which counts respondents whose latent utility favors the candidate rather than averaging choice probabilities.

For the headline contest---a co-partisan candidate who endorses prosecuting journalists---all three aggregations imply that the candidate falls below a majority. The plug-in share is $0.48$, the debiased mean-stage share is $0.47$, and the majority-preference function is $0.44$. Thus, the one-point wedge between the plug-in and mean-stage shares lies within the range acknowledged in the main text, and the substantive conclusion does not depend on the aggregation rule for this contest. We use the plug-in share in the main text because it is the closest analog to a predicted vote share under the fitted structural model: it averages each respondent's model-implied choice probability and includes the residual heterogeneity recovered by the empirical-Bayes update. The debiased mean-stage share is the cleaner inferential benchmark, while the majority-preference function is a scale-invariant directional benchmark that asks who has positive latent utility rather than who probabilistically votes for the candidate. We report the plug-in share in the main text and flag any contest for which the sign of the conclusion changes across aggregation rules.

\section{Validation Against Reduced-Form Benchmarks}\label{app:validation}

A useful sanity check for the hybrid estimator is to compare its mean-stage averages with standard pooled and subgroup homogeneous-logit estimates. These reduced-form logits are not, in general, exact targets for $\theta_k=\E[f_k(\bfZ)]$ in a heterogeneous logit model: because the logit link is nonlinear, the pooled coefficient solves its own projection problem and need not equal the average of heterogeneous mean-stage coefficients, even when $\Delta\bfX\perp\bfZ$. Equality requires no relevant heterogeneity or special symmetry/linearization conditions. We therefore use the pooled and subgroup logits as familiar reduced-form benchmarks rather than as quantities the structural estimator must reproduce exactly. This validation covers only the \emph{mean} of preferences---the individual-level heterogeneity recovered by the empirical-Bayes update is precisely what reduced-form estimators cannot access, and therefore cannot be validated against them.

We use the \citet{bansak2016europeans} European immigration conjoint for this check because of its large sample size ($NT = 74{,}090$ tasks from 14{,}818 respondents across 15 countries), which makes the within-country logit benchmarks themselves precisely estimated.  For the overall comparison, we fit a pooled homogeneous logit of $Y$ on the 28 attribute-level dummies and compare each coefficient to the DNN mean stage's $\hat\theta_k = N^{-1}\sum_i \hat f_k(\bfZ_i)$.  For the subgroup comparison, we refit the logit separately within each of the 15 country subsamples ($\approx 5{,}000$ tasks each) and compare each country-specific coefficient to the DNN's within-country conditional mean $\E_n[\hat f_k(\bfZ_i) \mid \text{country}_i = c]$.

\begin{figure}[t!]
  \centering
  \caption{Validation of Hybrid-estimator Averages against Reduced-form Homogeneous Logit Estimates on the Bansak et al. (2016) Immigration Conjoint}
  \label{fig:validation_amce}
  \IfFileExists{figs/fig_validation_amce.pdf}{
    \includegraphics[width=0.9\textwidth]{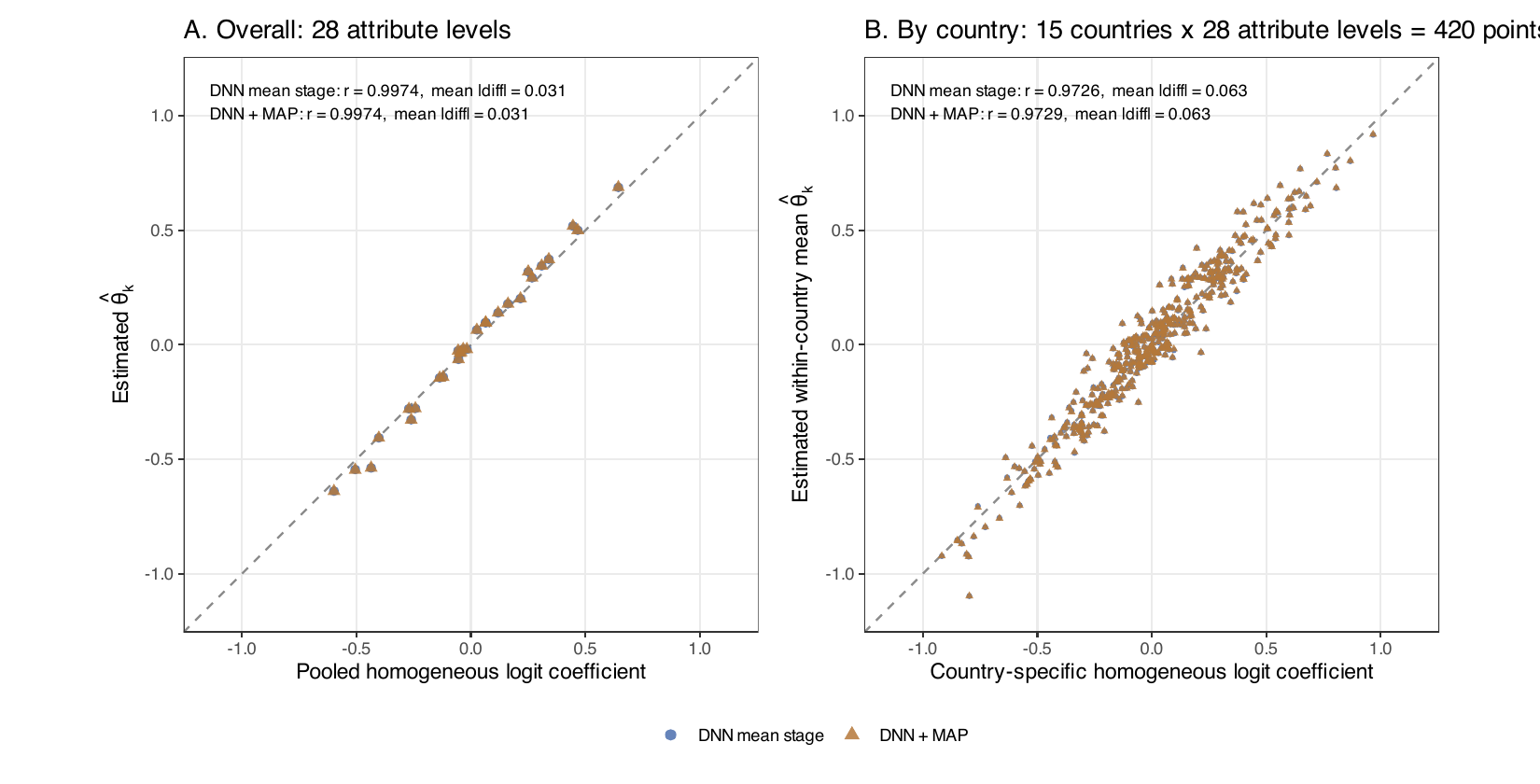}
  }{
    \fbox{\parbox{0.8\textwidth}{\centering\textit{[Validation figure pending: re-run \texttt{code/09\_run\_all.R} + \texttt{code/70\_validation\_amce.R} + \texttt{code/71\_figures\_validation\_amce.R} to regenerate \texttt{fig\_validation\_amce.pdf}.]}}}
  }
  \begin{minipage}[]{1\textwidth}\footnotesize
      \textit{Note:} Estimates use the \citet{bansak2016europeans} immigration conjoint ($NT = 74{,}090$).  Each panel overlays two flavors of the average estimator: the pure DNN mean stage (blue circles) and the MAP-overwritten version that adds per-respondent residuals $\hat\eta_i$ (orange triangles).  \textbf{A}: 28 attribute-level coefficients; each point compares $\hat\theta_k$ to the pooled homogeneous-logit coefficient.  \textbf{B}: $15 \times 28 = 420$ (country $\times$ attribute) pairs; each point compares the within-country conditional mean to the logit coefficient from a country-specific regression.  Dashed lines show $y = x$.  Both flavors track the overall and within-country reduced-form benchmarks closely and are nearly indistinguishable, consistent with the mean-stage averages lining up with familiar homogeneous-logit summaries and with the empirical-Bayes residuals adding negligible finite-sample noise to the population average.
  \end{minipage}
\end{figure}

Figure~\ref{fig:validation_amce} displays the two scatters, comparing both the pure DNN mean stage and the MAP-overwritten version ($\hat\theta_k = N^{-1}\sum_i (\hat f_k(\bfZ_i) + \hat\eta_{ik})$) against the same reduced-form benchmarks.  Because the homogeneous logits are benchmarks rather than exact estimands under heterogeneous preferences, close agreement is reassuring but not mechanically required; large discrepancies would be a warning sign about scale, coding, or misspecification.  The overall comparison (Panel~A, 28 points) lies almost perfectly on the 45-degree line for both estimators, with $r = 0.997$ and a mean absolute difference of $0.031$ log-odds units---well within the pooled logit's own sampling variability.  The subgroup comparison (Panel~B, 15 countries $\times$ 28 attributes $= 420$ points) also tracks the 45-degree line closely for both estimators, with $r = 0.973$ and a mean absolute difference of $0.063$, essentially identical for the DNN mean stage and the MAP version.  The slight attenuation of the subgroup correlation relative to the overall correlation is expected and has two sources: the country-specific logits are noisier because each uses only $\approx 5{,}000$ observations, and the DNN's within-country averages smooth slightly across observationally similar respondents from different countries.  In both panels, the hybrid estimator closely tracks the reduced-form benchmarks without special calibration. We view this as a diagnostic consistency check rather than a substantive contribution: it supports that the average mean-stage estimates are on the same scale and sign as familiar reduced-form logit summaries, while the individual-level heterogeneity results rest on the additional structural and empirical-Bayes assumptions developed above.

\FloatBarrier

\section{Design-Guidance Simulation}\label{app:factorial}

This appendix describes the simulation behind the design guidance presented in the conclusion.  We lay out the setup, summarize the main findings, translate them into guidance by quantity of interest, and report the supporting figures and one capacity dose-response table.

\subsection{Setup}\label{app:factorial_setup}

Each respondent's preference vector decomposes as
\[
  \beta_{ij} = \mu_j + f_j(\bfZ_i) + \eta_{ij}, \qquad \eta_{ij} \sim N(0, \sigma_\eta^2),\ \sigma_\eta = 0.3,
\]
where $\mu_j$ is the population-mean coefficient on attribute level~$j$, $f_j(\bfZ)$ is the systematic Z-driven component, and $\eta_{ij}$ is residual heterogeneity not explained by observables.  Choices follow a binary logit on a paired conjoint task with two profiles drawn uniformly over a dummy-coded design.  The vector $\boldsymbol{\mu}$ of population-mean coefficients is set by hand to span the magnitudes typically seen in applied conjoint AMCEs ($|\mu_j|$ in the logit range $0.02$--$0.40$, mixed signs), guided by the AMCE point estimates reported in \citet{bansak2023amce} and \citet{saha2022ambitious}.  Covariates $\bfZ$ have $P_Z = 10$ columns: three are continuous (standard normal) and the remaining seven are Bernoulli$(0.3)$; all are standardized.  The systematic component is rescaled coordinate by coordinate: for each coefficient $j$, $f_j(\bfZ)$ is rescaled so that $\Var(f_j(\bfZ))=\tau_z^2$, where $\tau_z=\sigma_\eta\sqrt{R_Z^2/(1-R_Z^2)}$. Thus the design factor $R_Z^2$ denotes the common target per-coordinate explained-variance ratio $\Var(f_j(\bfZ_i))/\Var(\beta_{ij})$, not a vector-level average across coefficients.

We cross five design factors---the number of respondents $N$, tasks per respondent $T$, the number of preference coordinates $p$, covariate informativeness $R_Z^2$, and the linear share $\lambda$ of the systematic component (defined below)---over the grid
\[
  N\in\{1{,}000,\ 2{,}500,\ 5{,}000,\ 10{,}000\},\
  T\in\{3,\ 5,\ 10,\ 15\},\
  p\in\{20,\ 30,\ 40\},
\]
\[
  R_Z^2\in\{0.35,\ 0.55,\ 0.75\},\
  \lambda\in\{0,\ 0.30,\ 0.50,\ 0.70,\ 1\}.
\]
The functional form by which $\bfZ$ drives preferences is governed by the linear-share parameter $\lambda$.  Let $L_{ij}=\bfZ_i^\top\boldsymbol{\gamma}_j$ denote the linear component, centered and rescaled to unit variance as $L_{ij}^{\circ}$.  Let
\[
  A_i^0 = 0.3(Z_{i1}^2-1) + 0.3 Z_{i1}Z_{i2} + 0.2(Z_{i3}^2-1),
  \qquad
  I_i^0 = S_i\{0.3 Z_{i1}Z_{i2}+0.3 Z_{i1}Z_{i3}+0.2 Z_{i2}Z_{i3}\},
\]
where $S_i=2\mathbf{1}\{Z_{i1}>0\}-1$ gives the sign-flipping component.  For any raw component $q_i$, write $q_i^{\perp L}$ for its centered residual after population projection on $L_{ij}^{\circ}$, rescaled to unit variance; write $I_i^{\perp A}$ for the centered, unit-variance residual from projecting $I_i^{\perp L}$ on $A_i^{\perp L}$.  Before the final coordinate-specific rescaling to $\Var(f_j(\bfZ_i))=\tau_z^2$, the systematic component is
\[
  \widetilde f_j(\bfZ_i;\lambda,q) =
  \sqrt{\lambda}\,L_{ij}^{\circ}+\sqrt{1-\lambda}\,q_{ij}^{\circ}.
\]
\begin{center}
\small
\begin{tabular}{@{}ll@{}}
\toprule
Family & Nonlinear component $q_{ij}^{\circ}$ \\
\midrule
Additive-nonlinear (AN) & $A_i^{\perp L}$ \\
Interactive-polarized (IN) & $I_i^{\perp L}$ \\
Mixed (MN) & $\operatorname{std}\{(A_i^{\perp L}+I_i^{\perp A})/\sqrt{2}\}$ \\
\bottomrule
\end{tabular}
\end{center}
Thus MN is a 50/50 orthogonalized mix of the AN and IN components.
Sweeping $\lambda$ from $0$ to $1$ thus traces a continuous path from a fully nonlinear DGP to a purely linear one, replacing a discrete class label with a continuous functional-form dimension.  Each $\lambda < 1$ is crossed with the three nonlinear families, so the functional-form factor has $4 \times 3 + 1 = 13$ configurations---four $\lambda < 1$ values times three families, plus the purely linear $\lambda = 1$---and crossing these with the $4 \times 4 \times 3 \times 3 = 144$ combinations of $(N, T, p, R_Z^2)$ yields $1{,}872$ cells, each replicated $R = 10$ times ($18{,}720$ runs in total).  Each run uses the empirical-Bayes hybrid DNN of Section~\ref{sec:estimation}: the DNN ensemble mean stage followed by the MAP update under the \texttt{EnsC5} prior-precision calibration ($c_\eta = 5$).

\subsection{Main findings}\label{app:factorial_findings}

\paragraph{Variance decomposition.} An ANOVA of the cell-mean individual-$\bfbeta$ correlation $\rho_\beta$ across the $1{,}872$ cells attributes the following shares of its variance to the design factors:
\[
  R_Z^2:\ 56\%,\quad N:\ 18\%,\quad T:\ 15\%,\quad \lambda:\ 4\%,\quad p:\ 2\%,\quad \text{DGP}:\ 1\%.
\]
Covariate informativeness dominates.  The number of respondents ($N$) and tasks per respondent ($T$) are nearly tied for second place, with $N$ slightly ahead because the grid reaches $N = 10{,}000$.  The linear-share parameter $\lambda$ and the remaining factors contribute little. \medskip

\noindent \textbf{$NT$ is not a sufficient statistic.}  At fixed $NT = 15{,}000$, $T = 3$ (with $N = 5{,}000$) yields cell-mean $\rho_\beta = 0.443$ and $T = 15$ (with $N = 1{,}000$) yields $0.481$ ($+8.6\%$).  Many tasks per respondent help even when total observations are held constant. \medskip

\noindent \textbf{Functional-form robustness.}  Cell-mean $\rho_\beta$ averages $0.552$ on the linear DGP and $0.46$--$0.50$ on the three nonlinear families (additive-nonlinear, interactive-polarized, mixed)---a spread of about $0.09$ across functional forms that span linear projection, polynomial curvature, sign-flipping interactions, and their orthogonalized mixture.  Even with an explicit linear-share parameter, the choice of nonlinear family contributes under $1\%$ of ANOVA variance once $\lambda$ is held fixed.  The binding constraint on individual-level recovery is therefore covariate informativeness, not functional form. \medskip

\noindent \textbf{Population-level point recovery is strong throughout.}  Across cells, mean $|\hat\theta - \theta|$ averages $0.046$ on the logit scale (range $[0.011, 0.168]$), falling steeply with $T$ (from about $0.07$ at $T = 3$ to $0.025$ at $T = 15$).  Population-mean $|\hat{P}(\text{A wins}) - P^\star|$ averages $2.1$ percentage points (range $[0.5, 7.6]$ pp).  The dependence on $R_Z^2$ is small because population averages integrate over the residual heterogeneity $\boldsymbol{\eta}_i$.  The same grid provides only coarse coverage summaries because each cell has $R=10$ replications. Averaged over all cells, debiased confidence-interval coverage is $92\%$ against the $95\%$ target; apparent cell-level or monotone trends, including the mild decline with larger $T$ in the grid average, should be read as descriptive diagnostics rather than precise coverage evidence. \medskip

\begin{table}[b!]
\centering
\caption{Sensitivity of Debiased-interval Coverage to the Mean-stage Network Architecture}
\label{tab:arch_coverage}
\begin{tabular}{lccc}
\toprule
Mean-stage architecture & Coverage & Bias coherence & $\rho_\beta$ \\
\midrule
$c(128, 64, 32)$ (default)    & $71.0\%$           & $0.80$ & $0.60$ \\
$c(128, 64, 64)$              & $79.7\%$           & $0.69$ & $0.60$ \\
$c(256, 128, 64)$             & $84.0\%$           & $0.62$ & $0.59$ \\
$c(256, 128, 128)$            & $86.8\%$           & $0.57$ & $0.59$ \\
$c(512, 256, 256)$            & $\mathbf{92.0\%}$  & $0.42$ & $0.56$ \\
\bottomrule
\end{tabular}
\end{table}

\noindent \textbf{Dedicated coverage check in the strongly nonlinear regime.}  The coarse grid flags its lowest-coverage cell---the interactive-polarized DGP at $\lambda = 0$, $N = 10{,}000$, $T = 15$, $p = 30$, $R_Z^2 = 0.55$---which combines the largest simulated design with the strongest mean-function nonlinearity. This is an intentionally adversarial stress test: $\lambda=0$ removes the linear component of the systematic preference signal, making it an unlikely description of typical conjoint applications. At the default mean-stage architecture $c(128, 64, 32)$ this cell records $71\%$ coverage in the grid. Because that number comes from the coarse $R=10$ screen, we treat it as a signal for a targeted diagnostic rather than as a precise cell-level estimate. A capacity dose-response on this cell, with $R = 200$ replications and all other design settings held fixed, traces how widening the mean-stage network affects coverage. In the table, \emph{bias coherence} is the share of the residual error that is directionally systematic rather than randomly varying across coefficients. The table shows that at $c(512, 256, 256)$ coverage rises to $92.0\%$ ($95\%$ Wilson interval $[87.4\%, 95.0\%]$ at $R = 200$), statistically indistinguishable from the nominal $95\%$ target.  Bias coherence falls from $0.80$ to $0.42$ along the dose-response, confirming that the coverage gain comes from a genuine reduction in systematic bias rather than from wider standard errors.  Individual-$\bfbeta$ recovery is essentially flat across the dose-response (a $0.04$ drop in $\rho_\beta$ at the widest network).  In practice, we widen the mean-stage network to $c(512, 256, 256)$ in regimes that combine the largest designs with the most nonlinear preference functions.  For the more typical regimes covered by the rest of the grid, the default $c(128,64,32)$ is supported by the coarse grid summaries, though not by a separate high-replication coverage study for every cell.  Further widening beyond $c(512, 256, 256)$ does not help: our diagnostics on $c(1024, 512, 512)$ show no coverage gain and signs of overfitting the structured residual, so $c(512, 256, 256)$ is the widest network we recommend. \medskip

\noindent \textbf{Estimator choice.}  The MAP update with the \texttt{EnsC5} calibration of Section~\ref{sec:estimation} (labeled MAP-c5 in the figures) improves on the plain DNN ensemble on average in every $R_Z^2$ regime (see Figure~\ref{fig:factorial_map_vs_ens}).  Mean MAP gain in individual-$\bfbeta$ correlation is $+0.055$ at $R_Z^2 = 0.35$, $+0.031$ at $0.55$, and $+0.011$ at $0.75$ where the DNN mean already captures most heterogeneity.  MAP attains the higher individual-$\bfbeta$ correlation in $98\%$ of cells at $R_Z^2 \in \{0.35, 0.55\}$, $92\%$ at $R_Z^2 = 0.75$, and $95.7\%$ of all cells.  See Supplementary Materials~\ref{app:estimation_details} for the diagnostic study that motivates MAP-c5 as the default.

\subsection{Design guidance by quantity of interest}\label{app:factorial_guidance}

Applied claims rarely require exacting precision---they need direction, magnitude, and ranking reliable enough to support the substantive conclusion.  Combining the variance decomposition with cross-cell results from our simulations yields the following heuristics:
\begin{itemize}[leftmargin=1.5em,noitemsep]
  \item \textbf{Mean-type quantities.}  AMCEs and AMEs are recovered well in any design that supports a credible reduced-form analysis, and aggregate counterfactuals are close behind.  Counterfactual vote shares and win probabilities---the building blocks of empirical models of electoral competition---are typically within $1$--$3$ percentage points in designs with $N=1{,}000$ respondents and either $T=5$ tasks at $R_Z^2 \geq 0.35$ or $T=10$ tasks at any simulated $R_Z^2$.
  \item \textbf{Distributional claims.}  Claims about how preferences are spread across respondents---the share who prefer candidates with policy~$X$, or which attribute carries the most weight on average---are well recovered with $T \geq 10$ tasks and reasonably informative covariates ($R_Z^2 \geq 0.35$).  Substantively clear conclusions (``a clear majority prefer policy~$X$,'' ``policy is the most important attribute'') are dependable in this regime; closer calls---a 52--48 split, or two attributes nearly tied in importance---require more tasks per respondent or stronger covariates.
  \item \textbf{Individual-respondent claims.}  Saying something about a particular respondent---their personal attribute ranking, or counterfactual choice probabilities computed from their own coefficient vector that improve on the population-mean prediction---requires that $\hat{\bfbeta}_i$ track the true preference vector at the individual level.  This requirement is more demanding.  It is met at $R_Z^2 \approx 0.75$ in modest designs, plausibly met around $R_Z^2 \approx 0.5$ given the steep slope of $\rho_\beta$ in $R_Z^2$, but not met at the bottom of our grid ($R_Z^2 = 0.35$), where even the largest simulated designs ($N = 10{,}000$, $T = 15$) reach only $\rho_\beta \approx 0.5$.  Aggregate and distributional claims (preceding bullets) succeed at lower $R_Z^2$, so researchers planning individual-level claims should invest in covariate quality first.
  \item \textbf{Tradeoff quantities.}  These split into ratios and sums.  MRS and WTP are negative ratios of coefficients---dollars per hour of travel time saved (WTP), or how much extra unemployment a respondent would tolerate for a point less inflation (MRS).  When the true ratio is clearly nonzero, modest coefficient noise leaves it largely intact.  When it is near zero, the small-denominator problem swings the ratio wildly, and we recommend caution.  In that regime, population-level summaries---the negative ratio of average coefficients, or a trimmed mean of the individual ratios---are far more stable than individual ratios, and the negative ratio of averages carries valid debiased confidence intervals (Section~\ref{sec:quantities}).  Compensating differentials, by contrast, are threshold conditions on sums of coefficients (does benefit~$k$ outweigh cost~$j$ for respondent~$i$?), so they avoid the small-denominator issue and are recovered as well as the coefficients themselves.  Projects centered on tradeoff quantities should invest in rich covariates, a large $N$ ($\geq 5{,}000$), and a high $T$; drawing the relevant attributes from a continuous range (e.g., continuously varying rates or prices, as in our tax-policy application in Section~\ref{sec:ballard_rosa}) rather than $3$--$4$ discrete levels also sharpens the coefficients and the ratios built from them.
\end{itemize}

\subsection{Main figures}\label{app:factorial_figures}

We present six figures that summarize the simulation, organized by the quantity-of-interest hierarchy of Section~\ref{sec:quantities}: easiest first, hardest last.  Most plots use $R_Z^2$ on the $x$-axis and $T$ as the line color so the dominant factors are visible immediately.

Figure~\ref{fig:factorial_theta} shows that the population mean $\theta_k$ is recovered accurately in every design.

\begin{figure}[H]
\centering
\caption{Recovery of Population-mean $\theta_k$ across the Simulation Design Grid of $R_Z^2$ and $T$}
\label{fig:factorial_theta}
\includegraphics[width=0.6\textwidth]{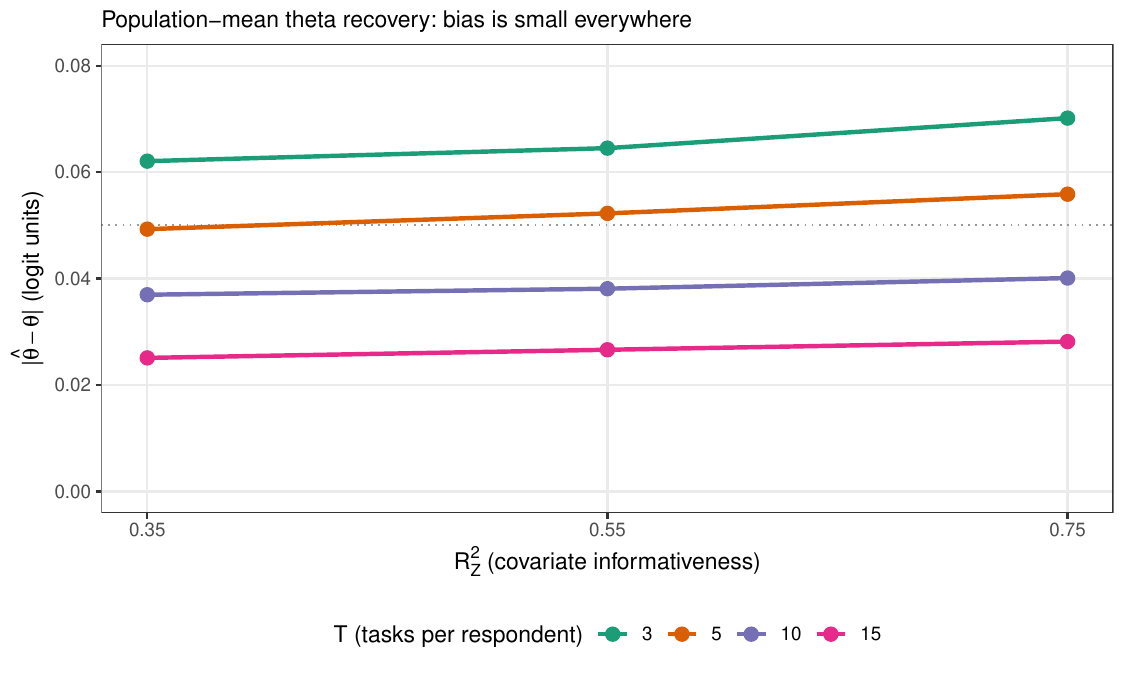}
\begin{minipage}[]{1\textwidth}\footnotesize
    \textit{Note:} Mean $|\hat\theta - \theta|$ on the logit scale, as a function of $R_Z^2$ and $T$.  $T$ is the dominant factor: bias falls from about $0.07$ logit units at $T=3$ to about $0.025$ at $T=15$ (the latter roughly $0.6$ pp on the probability scale at a baseline probability of $0.5$).  $R_Z^2$ matters very little because $\theta_k$ is a population mean.
\end{minipage}
\end{figure}

\begin{figure}[H]
\centering
\caption{Mean Absolute Error of Counterfactual Choice-probability Recovery across the Simulation Design Grid}
\label{fig:factorial_profmae}
\includegraphics[width=0.6\textwidth]{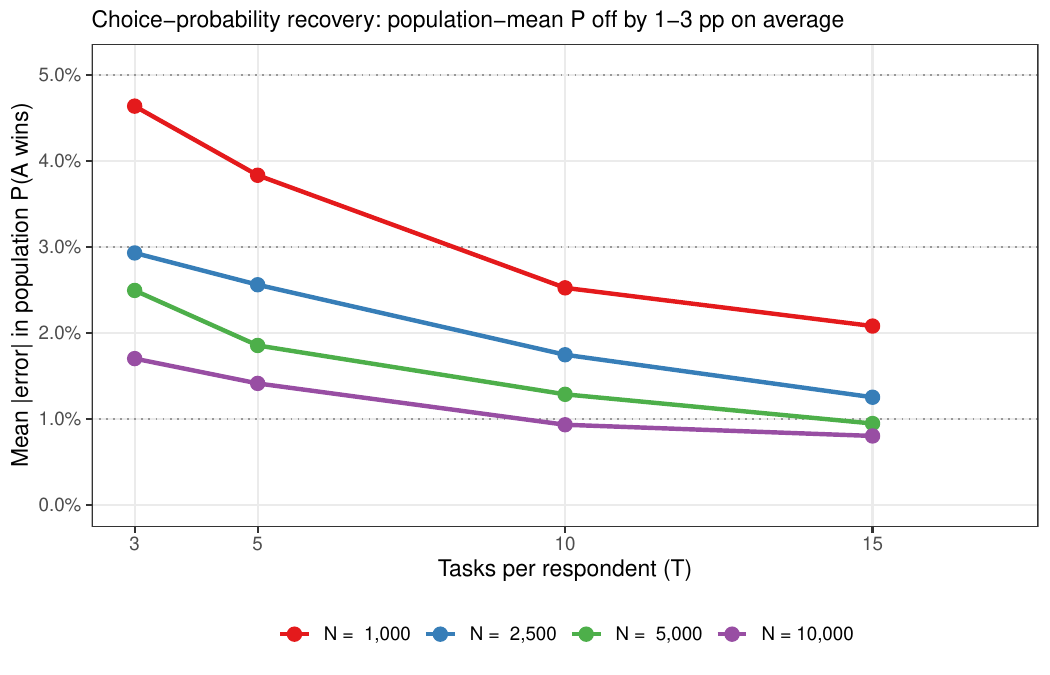}
\begin{minipage}[]{1\textwidth}\footnotesize
    \textit{Note:} Population-mean $|\hat{P}(\text{A wins}) - P^\star|$ across three representative profile matchups, as a function of $T$ and $N$.  Errors are 1--3 percentage points across most of the grid, reaching about $4.5$ pp only at the smallest design ($N{=}1{,}000$, $T{=}3$) and falling below $1$ pp at $N{=}10{,}000$.  $R_Z^2$ has essentially no effect (not shown), because aggregating over respondents averages out the residual heterogeneity $\boldsymbol{\eta}_i$.
\end{minipage}
\end{figure}

Figure~\ref{fig:factorial_profmae} shows that counterfactual choice probabilities come next in the hierarchy, recovered to within a few percentage points.

Figures~\ref{fig:factorial_polarization} and~\ref{fig:factorial_importance} turn to the distributional quantities---polarization fractions and importance shares---which are recovered well once $T \geq 10$.

\begin{figure}[H]
\centering
\caption{Recovery of Population Preference Fractions across the Simulation Design Grid}
\label{fig:factorial_polarization}
\includegraphics[width=0.7\textwidth]{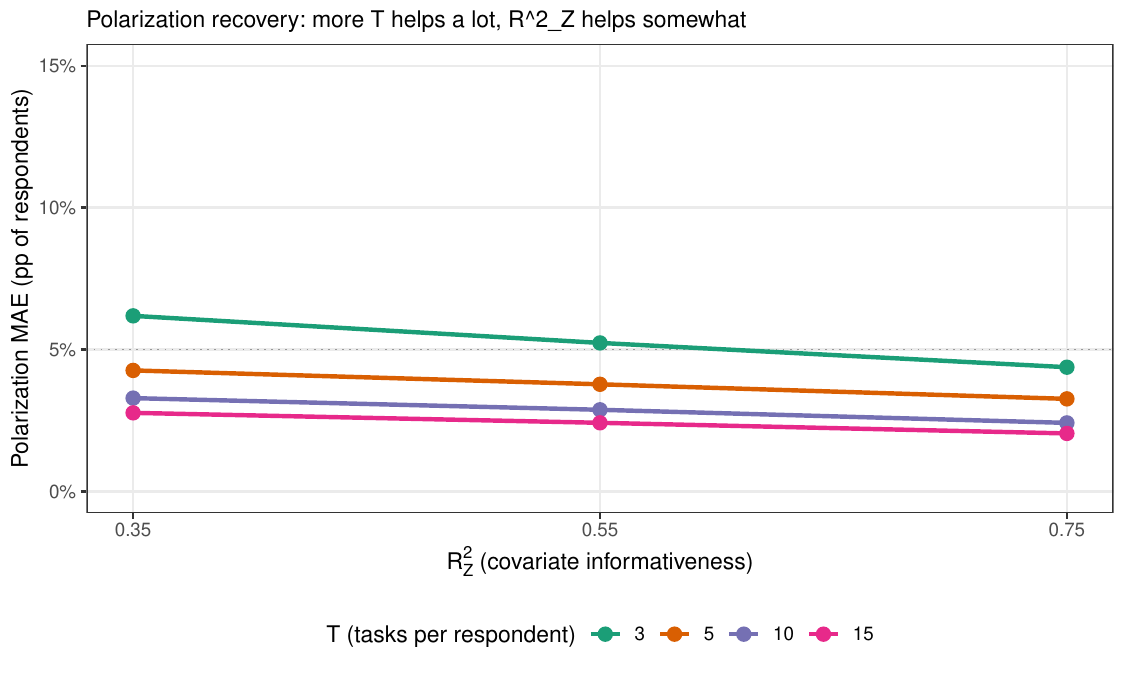}
\begin{minipage}[]{1\textwidth}\footnotesize
    \textit{Note:} MAE between estimated and true population fraction of respondents preferring each attribute level.  Both $T$ and $R_Z^2$ help: averaging over the grid, MAE runs about $4$--$6$ pp at $T{=}3$ and falls to about $2$ pp at $T{=}15$; $T$ does most of the work and $R_Z^2$ helps somewhat.
\end{minipage}
\end{figure}

\begin{figure}[H]
\centering
\caption{Correlation between Estimated and True Importance Shares across the Simulation Design Grid}
\label{fig:factorial_importance}
\includegraphics[width=0.7\textwidth]{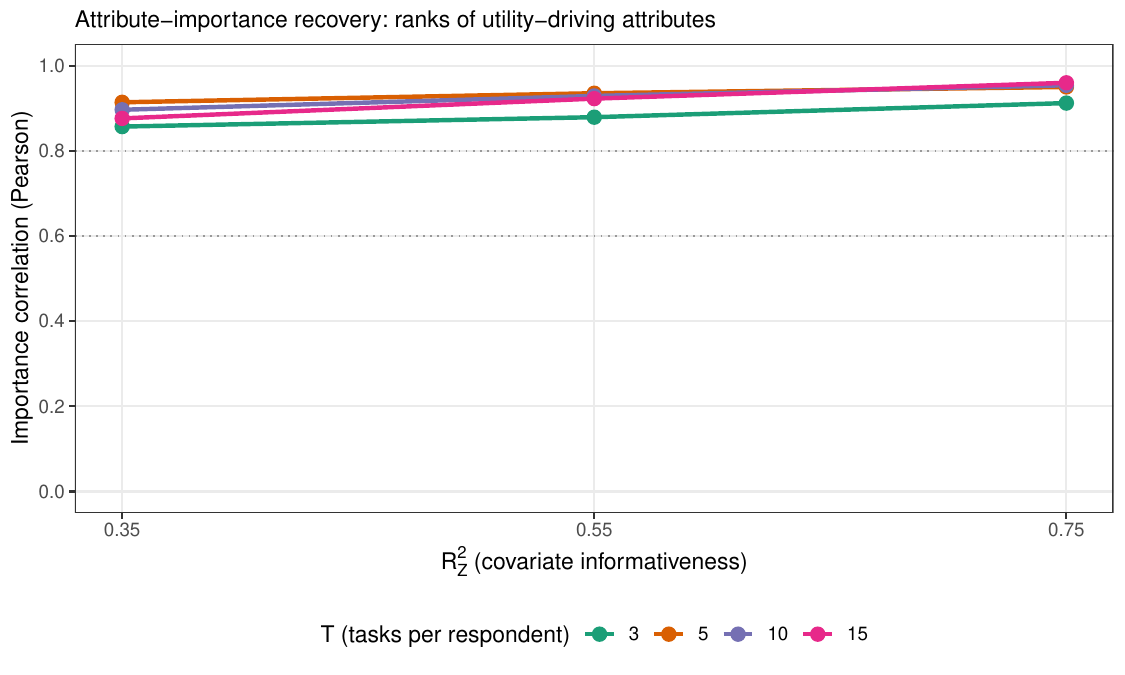}
\begin{minipage}[]{1\textwidth}\footnotesize
    \textit{Note:} Pearson correlation between estimated and true population-mean importance shares (importance$_g \propto \sum_{k \in g} \beta_k^2 \mathrm{Var}(X_k)$).  Correlation is high throughout---about $0.85$ at the weakest designs and $0.95$--$0.96$ at $R_Z^2 = 0.75$---so the ranking of utility-driving attributes is recovered well even in modest designs.
\end{minipage}
\end{figure}

\begin{figure}[H]
\centering
\caption{Individual-$\bfbeta$ Recovery Correlation $\rho_\beta$ from the MAP Estimator across the Simulation Design Grid}
\label{fig:factorial_grid}
\includegraphics[width=0.825\textwidth]{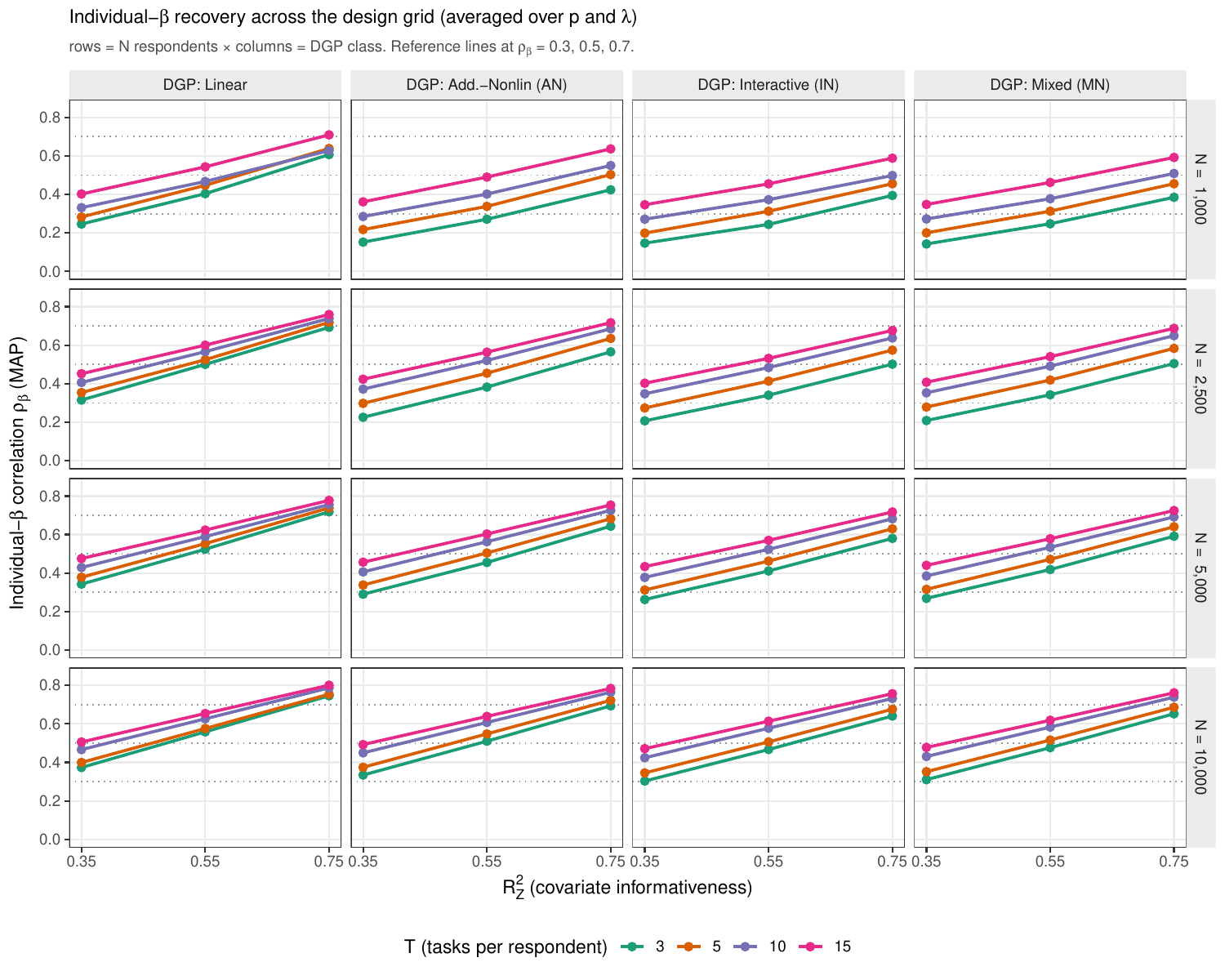}
\begin{minipage}[]{1\textwidth}\footnotesize
    \textit{Note:} Cell-mean $\rho_\beta$ from the MAP estimator, plotted against $R_Z^2$ with line color indicating $T$.  Columns: DGP class (linear and the three heterogeneous families).  Rows: $N$ respondents, increasing top to bottom.  Averaged over $p$ and the linear-share $\lambda$.  The within-panel slope across $R_Z^2$ is steep (the dominant factor), the four columns look essentially identical (DGP-class robustness), and reading down rows shows the $N$ effect.  Reference lines at $\rho_\beta = 0.3, 0.5, 0.7$.
\end{minipage}
\end{figure}

\begin{figure}[H]
\centering
\caption{Predictive Performance of MAP-c5 vs.~DNN Ensemble across $R_Z^2$}
\label{fig:factorial_map_vs_ens}
\includegraphics[width=0.825\textwidth]{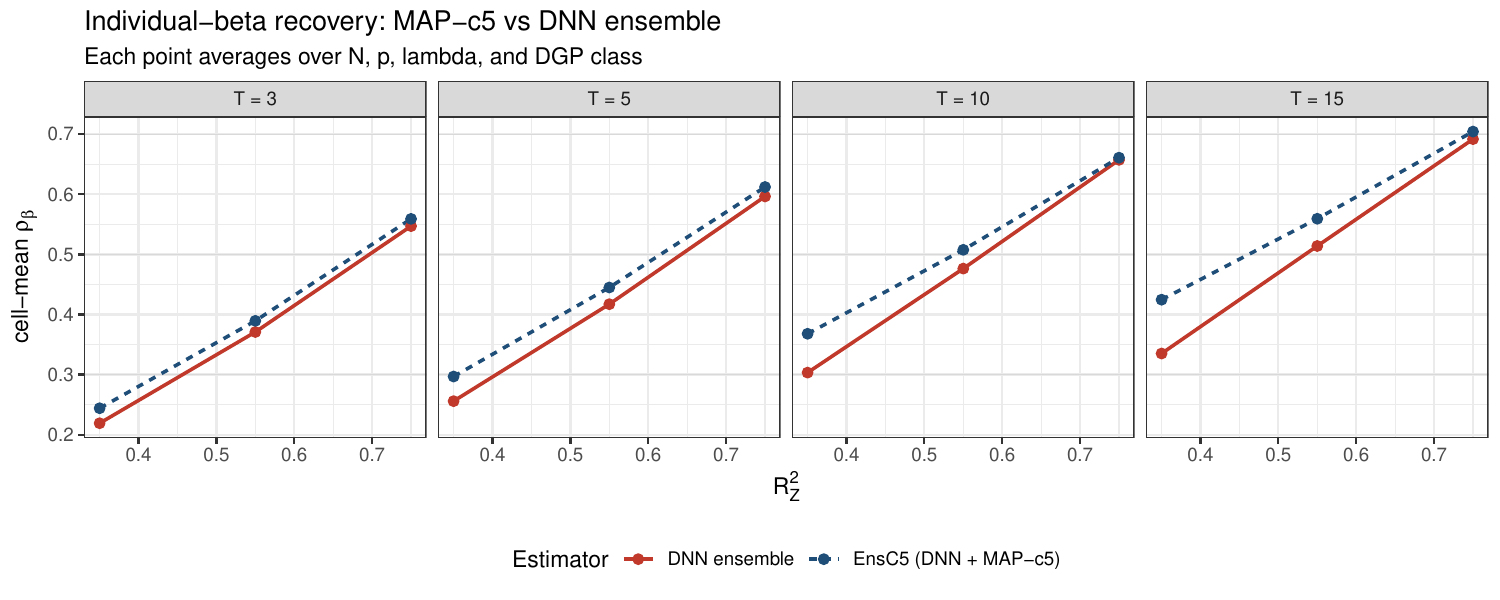}
\begin{minipage}[]{1\textwidth}\footnotesize
    \textit{Note:} MAP-c5 (blue) outperforms the DNN ensemble (red) most clearly at low $R_Z^2$, where the within-respondent update compensates for an uninformative cross-sectional prior; the gain fades gradually as $R_Z^2$ rises and the two nearly converge by $R_Z^2 = 0.75$.  One panel per $T$.
\end{minipage}
\end{figure}

Figure~\ref{fig:factorial_grid} is the full individual-$\bfbeta$ recovery surface across $N$, $T$, $R_Z^2$, and DGP class---the most demanding target in the hierarchy.

Figure~\ref{fig:factorial_map_vs_ens} isolates the estimator choice, showing where the MAP-c5 update improves on the plain DNN ensemble.

\FloatBarrier

\section{Additional Results for the Applications}\label{app:more_figures}

This section collects per-application supplements referenced from Section~\ref{sec:applications}, in the order of the main text: several robustness and validation supplements for the democracy application, the party-level importance decomposition for the tax application, and the distributional and counterfactual displays for the candidate application.

\subsection{The democracy application}\label{app:gs_covariates}

Here we collect four supplements to the democracy application of Section~\ref{sec:graham_svolik}: a comparison of the covariate set together with an external validation of the recovered preferences, survey-weighted versions of the headline quantities, a check that the partisan asymmetry is not a shrinkage artifact, and a benchmark against a regularized hierarchical model. \medskip

\noindent \textbf{Covariate-set comparison and external validation.}  The \citet{graham2020democracy} survey collected respondents' direct, pre-conjoint ratings of how undemocratic each practice is.  Including these in $\bfZ_i$ would let the first stage predict a respondent's conjoint preference over a practice partly from their own direct rating of that same practice, blurring the revealed-preference interpretation and rendering the validation in Section~\ref{sec:graham_svolik} circular.  Our main specification therefore excludes the six democracy-attitude items---the four action-specific ``how undemocratic is $X$'' ratings plus two general democracy-support items---retaining the demographic, partisan, ideological, and dispositional covariates ($|\bfZ| = 16$).  Table~\ref{tab:gs_covariates} and Figure~\ref{fig:gs_spec_comparison} show that the population-average estimates are nearly identical either way (maximum $|\Delta\hat{\theta}_k| = 0.05$ across the seven undemocratic actions), echoing \citeauthor{graham2020democracy}'s own Table~2, columns 5--6.  The \emph{individual-level} validation correlation, by contrast, behaves exactly as the leakage concern predicts: with the direct items held out it is $0.37$ (pooled Spearman), whereas placing them inside $\bfZ_i$ inflates it to roughly $0.65$---the mechanical consequence of correlating a quantity partly with itself, and the reason we hold them out.

\begin{table}[H]
  \centering
  \small
  \caption{Graham \& Svolik Estimates with vs.\ without the Six Direct Democracy-attitude Items in $\bfZ_i$}
  \label{tab:gs_covariates}
  \begin{tabular}{@{}lcc@{}}
    \toprule
     & With direct items & Without (main) \\
    \midrule
    \multicolumn{3}{@{}l}{\emph{Average preferences $\hat{\theta}_k$ (logit scale)}} \\
    \quad Co-partisan            & $+0.74$ & $+0.72$ \\
    \quad Ban protests           & $-0.60$ & $-0.61$ \\
    \quad Ignore courts          & $-0.75$ & $-0.75$ \\
    \quad Executive order        & $-0.58$ & $-0.58$ \\
    \quad Gerrymander (2 seats)  & $-0.52$ & $-0.57$ \\
    \quad Gerrymander (10 seats) & $-0.68$ & $-0.69$ \\
    \quad Prosecute journalists  & $-0.94$ & $-0.93$ \\
    \quad Close polling stations & $-0.81$ & $-0.77$ \\
    \midrule
    \multicolumn{3}{@{}l}{\emph{Importance share (\%)}} \\
    \quad Policy        & 26 & 26 \\
    \quad Valence       & 30 & 29 \\
    \quad Party         & 24 & 23 \\
    \quad Undemocratic  & 9  & 9  \\
    \quad Other         & 12 & 12 \\
    \midrule
    \multicolumn{3}{@{}l}{\emph{Median individual MRS vs.\ co-partisanship (overall)}} \\
    \quad Prosecute journalists  & 0.81 & 0.80 \\
    \quad Ignore courts          & 0.66 & 0.64 \\
    \quad Gerrymander (10 seats) & 0.60 & 0.55 \\
    \bottomrule
  \end{tabular}
  \begin{minipage}[]{1\textwidth}\footnotesize
      \textit{Note:} ``With'' uses the original's 22-covariate set; ``Without'' (our main specification) drops the direct ratings ($|\bfZ| = 16$).  The maximum absolute change in any average preference is $0.05$.
  \end{minipage}
\end{table}

\begin{figure}[!h]
  \centering
  \caption{Average Preferences $\hat{\theta}_k$ for the Seven Undemocratic Actions: w/ vs. w/o the Direct Democracy-attitude Items}
  \label{fig:gs_spec_comparison}
  \includegraphics[width=0.8\textwidth]{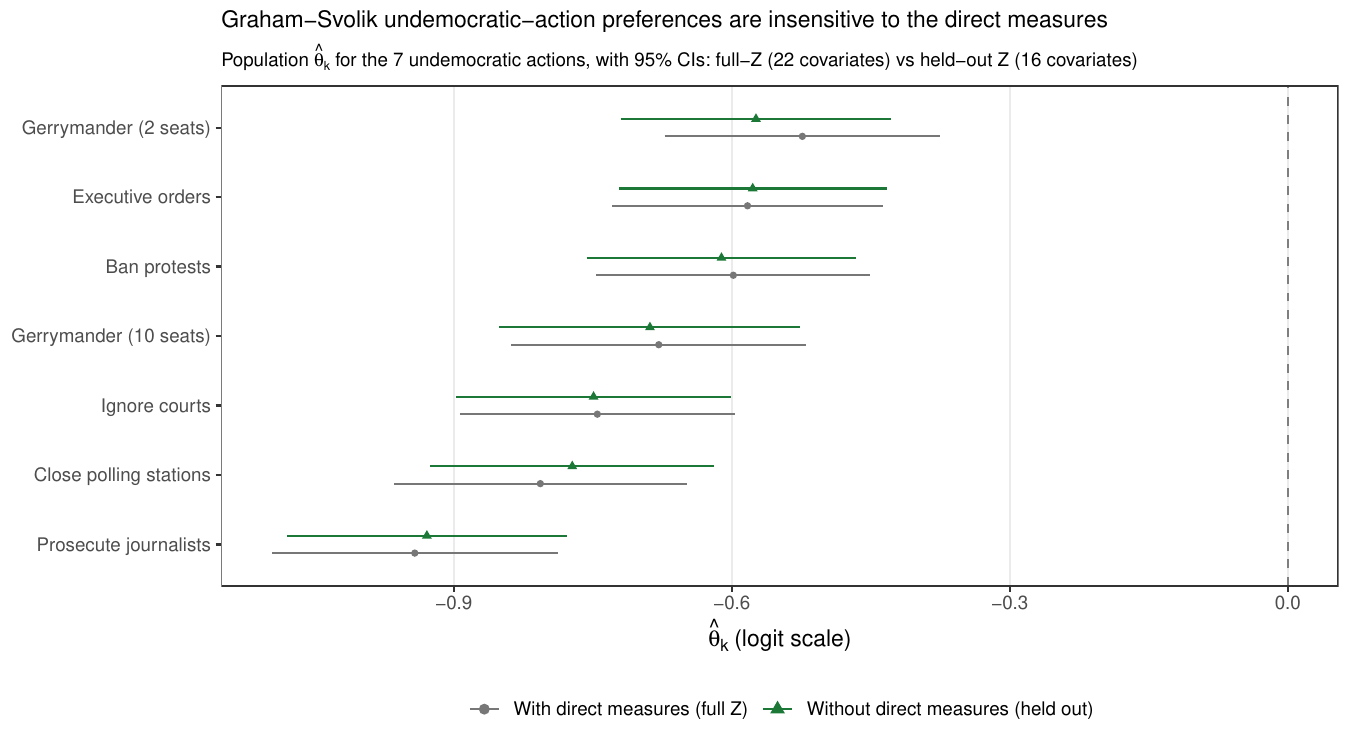}\vspace{0em}
  \begin{minipage}[]{1\textwidth}\footnotesize
        \textit{Note:} With 95\% DML Confidence Intervals. 
    \end{minipage}
\end{figure}\medskip

\FloatBarrier

\begin{figure}[ht!]
  \centering
  \caption{External Validation in the Democracy Conjoint}
  \label{fig:gs_validation}
  \includegraphics[width=\textwidth]{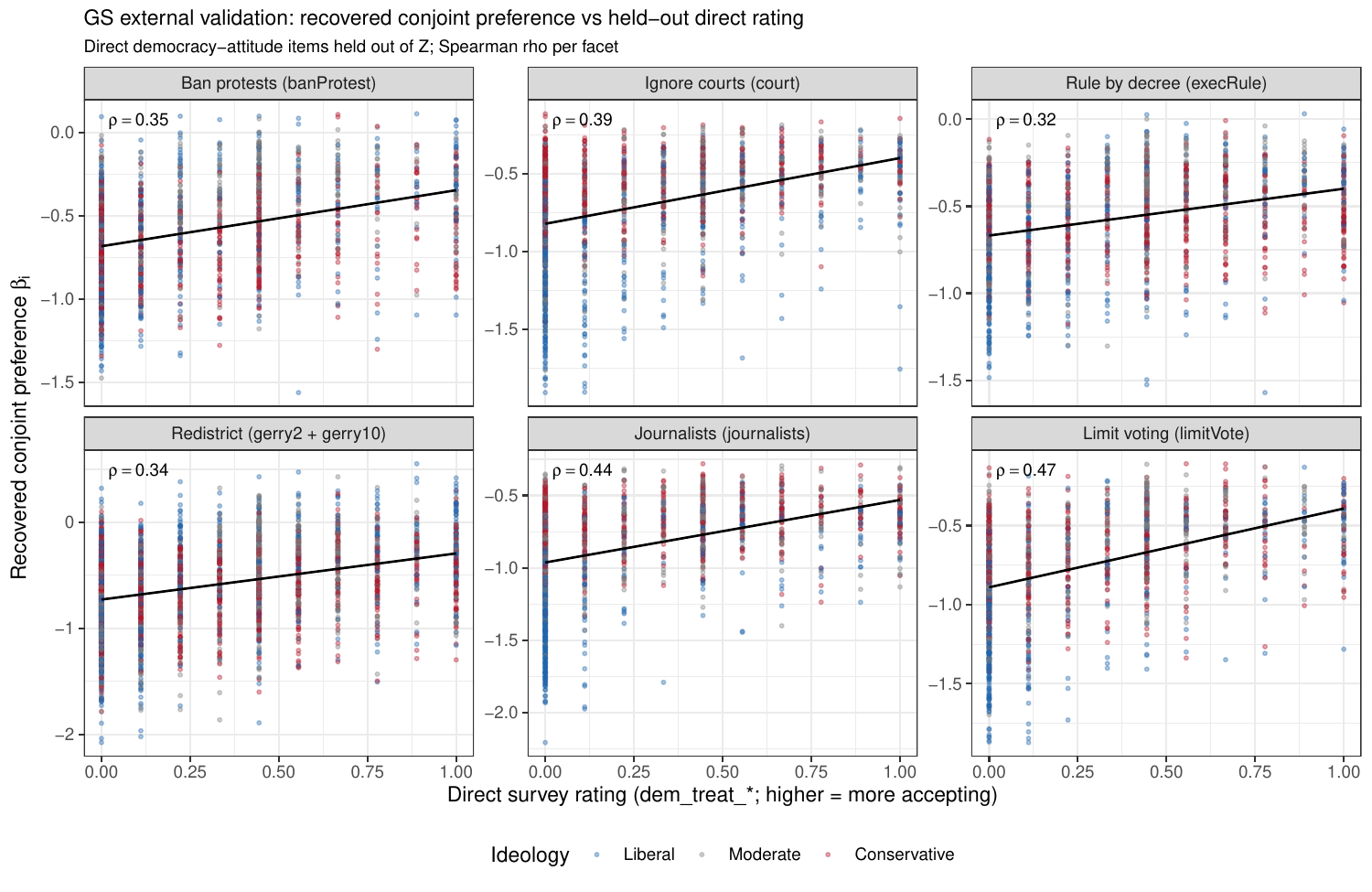}
  \begin{minipage}[]{1\textwidth}\footnotesize
      \textit{Note:} A respondent's recovered preference $\hat{\beta}_{i,k}$ for an undemocratic action vs.\ their direct, pre-conjoint rating of that practice, by ideology tercile, with a fitted line and the within-panel spearman correlation. Direct ratings are held out of estimation; both gerrymander levels share the one redistricting item.
  \end{minipage}
\end{figure}

\FloatBarrier

\noindent \textbf{Positive-tail diagnostic for gerrymandering.}  The two gerrymander variants---a co-partisan redistricting that nets two versus ten seats---separate the average penalty from the positive tail.  On average the ten-seat variant is more opposed than the two-seat variant ($\hat\theta = -0.69$ vs.\ $-0.57$), matching the ordering in \citeauthor{graham2020democracy}'s own reduced-form estimates, so favorability decreases in severity at the level the design identifies.  But the ten-seat variant carries the widest recovered $\hat\beta_i$ distribution of any action (standard deviation $0.44$ vs.\ $0.32$ for the two-seat variant), and hence the largest positive-coefficient minority.  The two variants are randomized equally, so this tail is not a sparsity artifact; the wider recovered distribution is consistent with sharper polarization over the larger prize, registered by the additive index as dispersion rather than interaction.  We therefore treat the tail share as less precisely identified than the corresponding average penalty. \medskip

\noindent \textbf{Survey-weighted estimates.}  The democracy application carries respondent survey weights that our main estimates do not use.  Table~\ref{tab:weighted} reweights the same held-out fit's respondent-level aggregation by those survey weights.  The average co-partisan benefit and the violation penalty are essentially unchanged ($\hat\theta_{\text{party}}$ moves from $0.72$ to $0.70$; $\hat\theta_{\text{journalists}}$ moves from $-0.93$ to $-0.89$), and the distributional conclusions survive: opposition to prosecuting journalists remains universal, and the partisan gradient in compensating differentials holds, with the liberal acceptance fraction moving from 38\% to 34\% and the conservative fraction stable at 66\%.  Survey weighting does not change the substantive conclusion.  The candidate and tax applications' public files do not carry usable respondent weights, so we report unweighted estimates there. \medskip

\begin{table}[!h]
\centering\small
\caption{Unweighted versus Survey-weighted Estimates in the Democracy Application}
\label{tab:weighted}
\begin{tabular}{@{}lcc@{}}
\toprule
Democracy application quantity & Unweighted & Survey-weighted \\
\midrule
$\hat\theta_{\text{co-partisan}}$ & $0.72$ & $0.70$ \\
$\hat\theta_{\text{prosecute journalists}}$ & $-0.93$ & $-0.89$ \\
Fraction opposing prosecuting journalists & 100\% & 100\% \\
Accept journalists for co-partisan --- Liberal & 38\% & 34\% \\
Accept journalists for co-partisan --- Conservative & 66\% & 66\% \\
\bottomrule
\end{tabular}
\begin{minipage}[]{1\textwidth}\footnotesize
    \textit{Note:} The survey-weighted column reweights respondent-level aggregation using the original respondent survey weights; the first-stage fit is unchanged.
\end{minipage}
\end{table}

\noindent \textbf{Differential shrinkage and the partisan asymmetry.}  Because the held-out validation correlation is higher for liberals than conservatives, one might worry that the asymmetry in compensating differentials---conservatives more willing than liberals to accept a co-partisan violation---is an artifact of conservatives' coefficients being more heavily shrunk.  The prior-sensitivity grid of Supplementary Materials~\ref{app:prior_sensitivity} speaks directly to this: the asymmetry is present in the unshrunk stage-one deep-network preferences (liberal acceptance $35\%$, conservative $65\%$) just as in the default empirical-Bayes estimates (liberal $38\%$, conservative $66\%$), so it is not produced by the degree of shrinkage.  The ordering is also stable across the full $c_\eta$ range. \medskip

\noindent \textbf{Benchmark against a regularized hierarchical model.}  A natural question is whether the flexible mean stage and the empirical-Bayes update earn their keep over a standard regularized hierarchical model.  Table~\ref{tab:benchmark} compares three estimators on the democracy application: the hybrid, a properly regularized Bayesian mixed logit with a covariate-modeled mean and respondent-level random slopes (fit with \texttt{glmmTMB}), and the pure mean stage with no respondent update.  We score each on the paper's own external-validation metric---the pooled correlation between recovered preferences and the held-out direct democracy ratings---and on in-sample choice fit.  On the validation metric the hybrid recovers preferences that track the independent ratings best, and it does so for six of the seven undemocratic actions (Table~\ref{tab:benchmark_action}), the ten-seat gerrymander being the lone near-tie.  The mixed logit attains a lower \emph{in-sample} choice log-loss, as expected from its larger set of freely estimated covariance parameters, but its recovered preferences validate less well against the independent benchmark---the comparison that matters for the individual-level quantities we report.  The empirical-Bayes update likewise improves on the mean stage alone on both metrics.  The flexible mean stage and the respondent-level update each contribute.

\begin{table}[!h]
\centering\small
\caption{Out-of-sample Validation and In-sample Fit across Estimators in the Democracy Application}
\label{tab:benchmark}
\begin{tabular}{@{}lccc@{}}
\toprule
Estimator & Validation $\rho$ & In-sample log-loss & Accuracy \\
\midrule
Hybrid (DNN $+$ empirical Bayes) & $0.37$ & $0.56$ & $0.86$ \\
Regularized mixed logit (\texttt{glmmTMB}) & $0.31$ & $0.50$ & $0.77$ \\
Mean stage only ($\hat f(\bfZ)$, no update) & $0.36$ & $0.66$ & $0.64$ \\
\bottomrule
\end{tabular}
\begin{minipage}[]{1\textwidth}\footnotesize
    \textit{Note:} The democracy application comprises $N = 1{,}605$ respondents and $20{,}657$ tasks.  Validation $\rho$ is the pooled Spearman correlation between the recovered preferences and the held-out direct democracy ratings---the metric that bears on the individual-level quantities we report; log-loss and accuracy are in-sample on the forced choices.  The hybrid validates best; the mixed logit fits in-sample best but validates worst, and the mean stage alone validates nearly as well as the hybrid but predicts choices far less accurately.
\end{minipage}
\end{table}

\begin{table}[t]
\centering\small
\caption{Validation Correlations between Recovered Preferences and Held-out Direct Ratings, by Undemocratic Action}
\label{tab:benchmark_action}
\begin{tabular}{@{}lccc@{}}
\toprule
Undemocratic action & Hybrid & Mixed logit & Mean stage \\
\midrule
Ban protests           & $0.354$ & $0.279$ & $0.336$ \\
Ignore courts          & $0.390$ & $0.273$ & $0.381$ \\
Executive order        & $0.323$ & $0.260$ & $0.305$ \\
Gerrymander (2 seats)  & $0.383$ & $0.299$ & $0.359$ \\
Gerrymander (10 seats) & $0.324$ & $0.326$ & $0.304$ \\
Prosecute journalists  & $0.437$ & $0.399$ & $0.421$ \\
Close polling stations & $0.471$ & $0.364$ & $0.448$ \\
\bottomrule
\end{tabular}

\begin{minipage}[]{1\textwidth}\footnotesize
    \textit{Note:} Each $\rho$ correlates the recovered preference against the held-out direct rating. The hybrid is highest for six of the seven actions; the ten-seat gerrymander is the single exception, a near-tie with the mixed logit.
\end{minipage}
\end{table}

\subsection{The tax application}

Subgroup AMCEs average within groups, so they can show that parties differ but not which attributes drive each party's choices; the variance decomposition answers that question.  Figure~\ref{fig:br_importance_party} reports, for each party, the share of plan-choice variance attributable to each of the six bracket rates and the revenue indicator---the importance shares of Section~\ref{sec:graham_svolik}, computed by party.  These are the shares behind the reframing of the partisan account in Section~\ref{sec:ballard_rosa}: Democrats' choices are driven most by the top bracket ($29.2\%$) and the very bottom bracket, while Republicans weight the working- and middle-class brackets most heavily ($26.1\%$ and $17.6\%$, against $15.9\%$ and $8.5\%$ for Democrats) and, of the three groups, place the least weight on the top bracket ($18.4\%$).  The full display adds two facts the main text does not quote: Independents fall between the parties on each of the four largest-share brackets, and the revenue indicator contributes a small share ($4\%$) in every group, so the partisan disagreement plays out across the brackets rather than over revenue.

\begin{figure}[t!]
  \centering
  \caption{Variance-decomposition Attribute-importance Shares by Party in the Ballard-Rosa et al. (2017) Conjoint}\label{fig:br_importance_party}
  \includegraphics[width=0.85\textwidth]{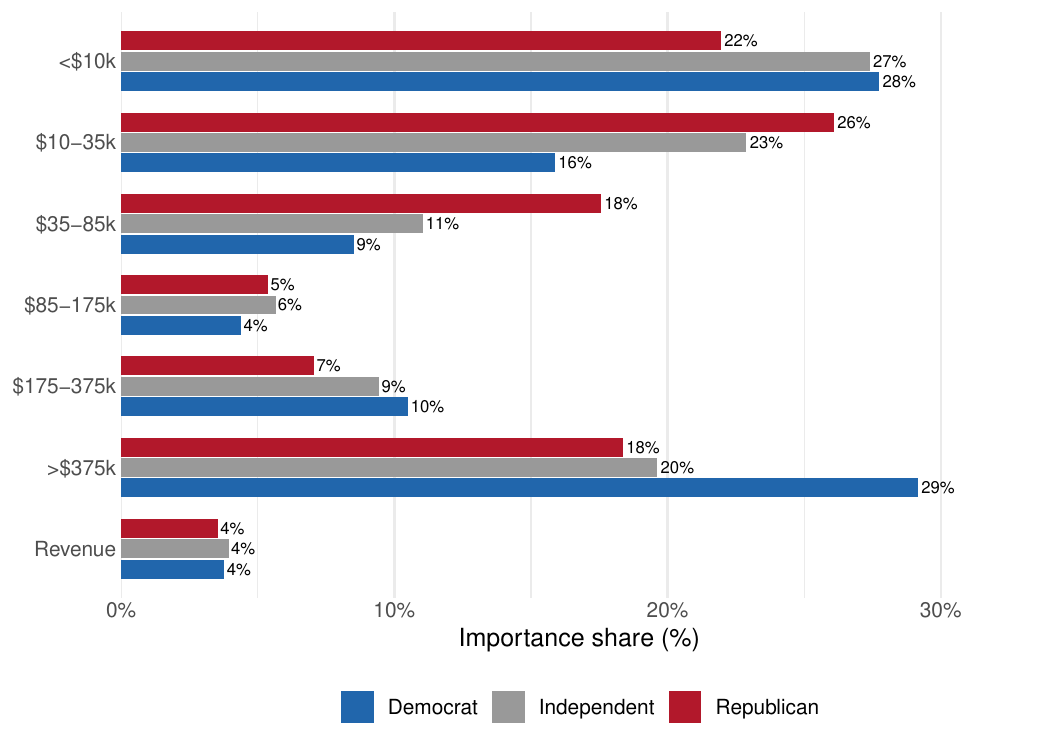}
  \begin{minipage}[]{1\textwidth}\footnotesize
      \textit{Note:} Shares are from a variance decomposition. Democrats weight the top bracket more; Republicans weight the low and middle brackets more.
  \end{minipage}
\end{figure}

\subsection{The candidate application}

Section~\ref{sec:saha_weeks} examines what the model can still recover in the sparsest design ($T = 3$) and argues that the near-zero average gender effect masks offsetting partisan camps.  The three figures here show the distributions and the contest behind that argument.

Figure~\ref{fig:sw_beta_ridgelines} shows the full densities of $\hat{\beta}_{ik}$ across respondents for all 13 attribute levels, ordered by variance---the distributions that the favor-versus-oppose fractions of Figure~\ref{fig:sw_amce_fraction} compress to signs.  For most levels the disagreement is about intensity, not direction: nearly all respondents favor both agenda levels, with effects from near zero to more than $1.5$ logit units.  Gender is the exception---wide, centered near zero, and split in direction---and that split is what Figure~\ref{fig:sw_partisan_gender} in the main text decomposes by party.

\begin{figure}[t!]
  \centering
  \caption{Ridgeline Densities of $\hat{\beta}_{ik}$ across Respondents, by Attribute Level (Saha \& Weeks 2022), Ordered by Variance}
  \label{fig:sw_beta_ridgelines}
  \includegraphics[width=\textwidth]{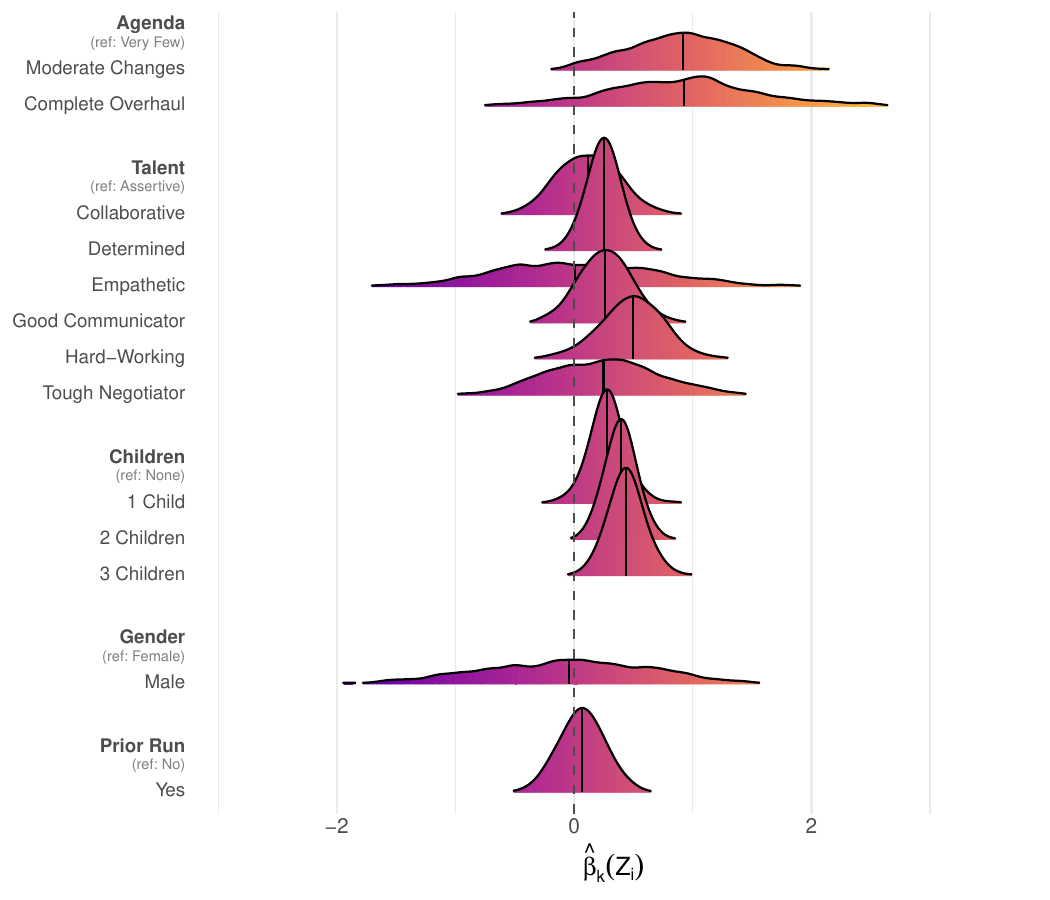}
  \begin{minipage}[]{1\textwidth}\footnotesize
      \textit{Note:} Dashed line at zero, solid at the median.
  \end{minipage}
\end{figure}

Figure~\ref{fig:sw_importance} addresses how much each attribute drives the vote.  It shows the distribution of the individual-level importance shares $\text{Imp}_{i,g}$ across respondents, with dashed lines at the means reported in Section~\ref{sec:saha_weeks}: policy agenda $52\%$, talent $21\%$, gender $16.5\%$, children $8\%$, and progressive ambition $2\%$.  The spread around those means is itself informative: voters differ not only in which way they lean but in which attributes their decisions turn on, heterogeneity that is invisible to average-effect analysis.

\begin{figure}[t!]
  \centering
  \caption{Distribution of Individual-level Attribute-importance Shares (Saha \& Weeks 2022)}
  \label{fig:sw_importance}
  \includegraphics[width=0.7\textwidth]{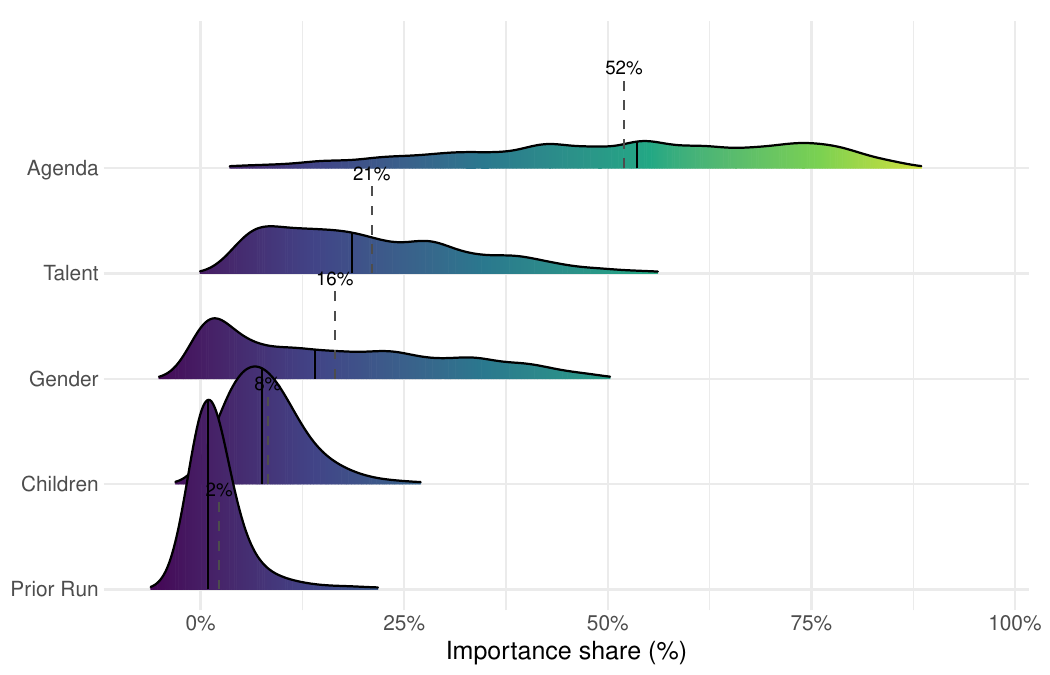}
  \begin{minipage}[]{1\textwidth}\footnotesize
      \textit{Note:} Dashed lines mark means.  Policy agenda is the largest average share ($52\%$), with wide individual heterogeneity.
  \end{minipage}
\end{figure}

Counterfactual contests are the payoff of recovering the full preference vector: once each respondent carries $\hat{\bfbeta}_i$, any head-to-head matchup becomes computable, including contests in which multiple attributes move together, which attribute-by-attribute averages cannot deliver.  Figure~\ref{fig:sw_agenda_counterfactual} stages the matchup of Section~\ref{sec:saha_weeks}, chosen to stack the two attributes voters split on most sharply: Candidate A, an Empathetic Female, faces a Tough Negotiator Male fixed at Complete Overhaul, and A varies her policy agenda across its three levels---Very Few Changes (preserve the status quo), Moderate Changes (incremental reform), and Complete Overhaul (wholesale transformation).  At Complete Overhaul the talent-and-gender contrast already pulls A below $50\%$ overall (Democrats $59.5\%$, Independents $47.9\%$, Republicans $34.5\%$); Moderate Changes leaves the partisan gap intact (Democrats $60.3\%$, Republicans $33.6\%$) near $48\%$ overall; and Very Few Changes drops her to $35.9\%$ overall and to $24.1\%$ among Republicans.  Reading across the three positions shows the two mechanisms of Section~\ref{sec:saha_weeks} operating together: scaling back the agenda costs about 12 points in every group---a penalty no partisan benefit offsets---while the $\sim$25-point partisan gap persists at every position.

\begin{figure}[t!]
  \centering
  \caption{Predicted Vote Share in Saha \& Weeks (2022) for Candidate A vs.\ Candidate B as A Scales Back Her Agenda}
  \label{fig:sw_agenda_counterfactual}
  \includegraphics[width=0.8\textwidth]{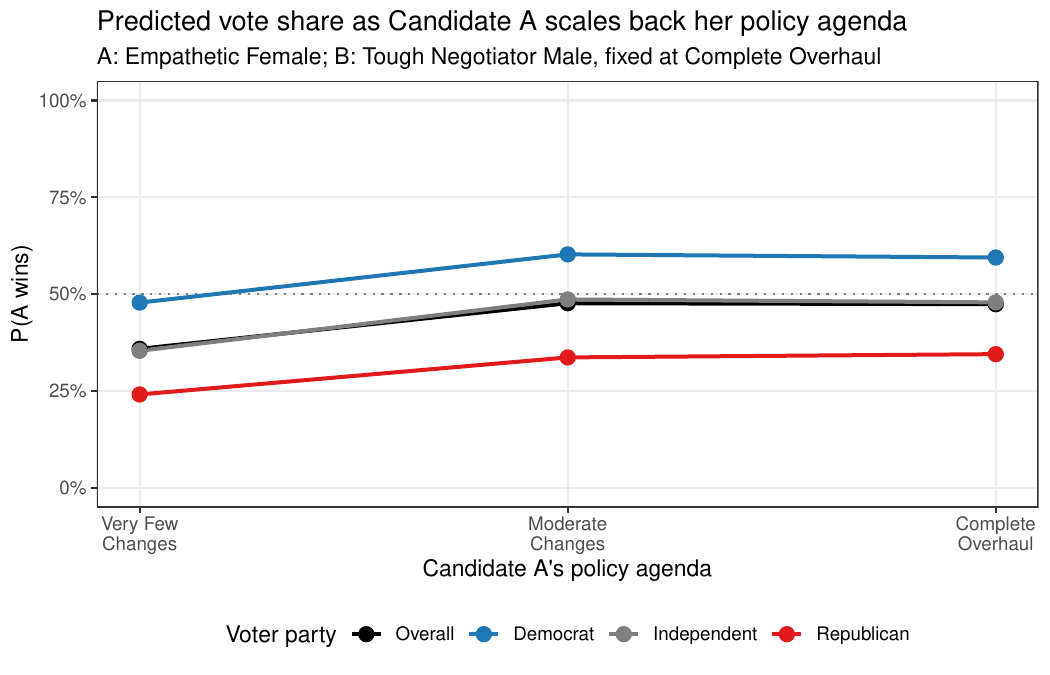}
  \begin{minipage}[]{1\textwidth}\footnotesize
        \textit{Note:} Candidate A is the Empathetic Female; Candidate B is the Tough Negotiator Male, fixed at Complete Overhaul.  Each line aggregates $\Pr(A \text{ wins} \mid \hat{\bfbeta}_i)$ within a party subgroup.  Two mechanisms are visible at once: the agenda drop between Moderate and Very Few Changes (about 12 points, near-uniform across parties), and a roughly 25-point partisan gap at every agenda position.
    \end{minipage}
\end{figure}

\noindent \textbf{Robustness to attribute interactions.}  The analysis of the candidate application in Section~\ref{sec:saha_weeks} maintains the additive index of Section~\ref{sec:framework}.  Because this is the sparsest of the three designs ($T = 3$), it is the one most exposed to undetected interactions, so we refit it under the two forms of Supplementary Materials~\ref{app:interactions}, holding the production configuration fixed and varying only whether and how interactions enter.  Table~\ref{tab:sw_interactions} reports the comparison.  The plug-in quantities that carry the substantive findings barely move: the gender importance share rises from $16.1\%$ to at most $17.5\%$, the partisan preference fractions shift by at most a few points, the plug-in win probability in the ambitious-woman contest is stable to the third decimal, and cross-fitted held-out log-loss changes by less than $0.01$ in either direction (the low-rank fit lowers it by $0.006$, the explicit fit raises it by $0.009$).  Allowing interactions does not change these conclusions.

\begin{table}[h!]
  \centering
  \small
  \caption{Robustness of the Candidate-application Findings to Attribute Interactions}
  \label{tab:sw_interactions}
  \begin{tabular}{lccc}
    \toprule
    & Additive & Explicit & Low-rank ($r=2$) \\
    \midrule
    \multicolumn{4}{l}{\emph{Panel A. Plug-in quantities}} \\
    Gender importance share (\%)                    & $16.1$  & $16.6$  & $17.5$ \\
    Democrats preferring female (\%)                & $66.5$  & $67.5$  & $69.0$ \\
    Republicans preferring male (\%)                & $65.7$  & $65.0$  & $63.2$ \\
    Plug-in $\Pr(\text{ambitious woman beats man})$ & $0.515$ & $0.516$ & $0.528$ \\
    Held-out log-loss                               & $0.712$ & $0.722$ & $0.706$ \\
    \midrule
    \multicolumn{4}{l}{\emph{Panel B. Debiased inference under the expanded interaction specification}} \\
    Debiased mean gender effect $\hat\theta$        & $0.10$  & $-13.8$ & $9.5$ \\
    \quad standard error                            & $0.07$  & $18.7$  & $17.4$ \\
    Interactions with $|z| > 2$ (of $59$)           & ---     & $0$     & $0$ \\
    \bottomrule
  \end{tabular}
  \begin{minipage}[]{1\textwidth}\footnotesize
      \textit{Note:} All three columns are fresh refits under the identical production configuration; the additive column reproduces the canonical numbers within fit noise (gender importance $16.1\%$ against the $16.5\%$ of Section~\ref{sec:saha_weeks}, with two independent seeds of the additive fit differing by $0.03$ points on that share).  Panel~A: the plug-in quantities that carry the substantive findings are stable across specifications.  Panel~B: at $T = 3$ the $59$ interaction terms cannot be estimated with any precision, so the debiased gender effect has a standard error above $17$ and no coefficient separates from zero.
  \end{minipage}
\end{table}

Inference on the interactions themselves is another matter.  With only three tasks per respondent, there is too little within-respondent information to estimate the $59$ extra cross-attribute parameters.  The debiased average gender effect, precisely estimated under the additive index ($\hat\theta = 0.10$, standard error $0.07$), carries a standard error above $17$ once the interaction terms are added, and none of the $59$ interaction coefficients separates from zero ($0$ reach $|z| > 2$; the usual ``about $3$ by chance'' benchmark would require independent, unregularized, well-calibrated $z$-statistics and is only a rough reference here).  This is the breakdown anticipated by the expanded information matrix in Supplementary Materials~\ref{app:interactions}, and it leaves the debiased interaction intervals uninformative.  We therefore treat the additive index as a reasonable working specification for these data, with the narrower claim that the attempted interaction refits leave the plug-in substantive quantities stable; formal interaction inference requires richer designs ($T \geq 10$, larger $N$).

\FloatBarrier

\subsection{Sensitivity to the empirical-Bayes prior}\label{app:prior_sensitivity}

We assess how sensitive the distributional findings are to the empirical-Bayes prior constant $c_\eta$, which scales the per-respondent shrinkage in the stage-two MAP update (our default is $c_\eta = 5$, the \texttt{EnsC5} estimator).  We re-estimate the respondent-level preferences at $c_\eta \in \{1, 5, 25\}$ and, as a limiting case, report the stage-one deep-network preferences with no empirical-Bayes update.  The qualitative conclusions are robust: in the democracy application, large majorities oppose every undemocratic action (all between $94\%$ and $100\%$) and the gerrymandering positive-tail share stays near $6\%$ across all settings; in the candidate application, the partisan gender split (roughly two-thirds of Democrats preferring the woman and two-thirds of Republicans the man) and gender's importance share (about $16\%$) move by at most one to two percentage points.  Signs, orderings, and overall magnitudes are therefore not artifacts of the prior.  The one quantity that is materially prior-dependent is the compensating-differential fraction---the share willing to accept a co-partisan who would prosecute journalists---among ideological moderates, which ranges from $33\%$ to $17\%$ as the prior tightens from $c_\eta = 1$ to $25$.  This is expected: that fraction is a threshold-crossing count and the moderate group sits near the indifference boundary, so it is the most exposed to the degree of individual-level shrinkage.  We report results at the default $c_\eta = 5$, and note that the substantive claim---that co-partisanship buys more tolerance among conservatives than among moderates or liberals---holds at every prior we examined.

\FloatBarrier

\section{Extensions to Other Outcome Formats}\label{app:extensions}

Recall the decomposition $\bfbeta_i = f(\bfZ_i) + \boldsymbol{\eta}_i$ of~\eqref{eq:beta_decomp}, with preferences entering the linear utility index $u_{it} = \bfDelta\bfX_{it}^\top \bfbeta_i$.  For the DNN mean stage and respondent-level empirical-Bayes update, the preference decomposition is unchanged across outcome formats; the observation model supplies the likelihood.  In every case below, $f(\bfZ)$ is estimated by the same cross-fitted DNN mean stage of Supplementary Materials~\ref{app:estimation_details}, with only the training loss swapped to the negative log-likelihood of the new observation model (squared error for ratings, multinomial cross-entropy for choice and ranking).  This appendix sketches the extension to ratings, multinomial choice, and rankings, describing for each (i)~the likelihood and (ii)~the empirical-Bayes update for $\boldsymbol{\eta}_i$.  The same orthogonal-score logic extends after replacing the binary-logit score and information matrix with the corresponding score and Hessian for the outcome model; we do not develop those formulas here.

\subsection{Ratings}\label{app:ext_ratings}

When respondents rate each profile on a continuous scale (e.g., feeling thermometers, willingness-to-vote scores), write the profile-level observation model as
\begin{equation*}
  Y_{ijt} = \bfX_{ijt}^\top \bfbeta_i + \varepsilon_{ijt}, \qquad
  \varepsilon_{ijt} \sim N(0,\sigma_\varepsilon^2), \quad j=1,2,
\end{equation*}
with independent profile-level errors. In differenced form, $\Delta Y_{it}=\bfDelta\bfX_{it}^{\top}\bfbeta_i+\Delta\varepsilon_{it}$ and $\Var(\Delta\varepsilon_{it})=2\sigma_\varepsilon^2$. The respondent-level update therefore has the closed-form Gaussian solution
\begin{equation*}
  \hat{\boldsymbol{\eta}}_i
  =
  \left(
    \bfDelta\bfX_i^\top \bfDelta\bfX_i /(2\sigma_\varepsilon^2)
    +
    \boldsymbol{\Sigma}_\eta^{-1}
  \right)^{-1}
  \bfDelta\bfX_i^\top
  \left(\Delta\mathbf{Y}_i-\bfDelta\bfX_i\hat f(\bfZ_i)\right)
  /(2\sigma_\varepsilon^2),
\end{equation*}
where $\bfDelta\bfX_i$ stacks the respondent's task contrasts and $\boldsymbol{\Sigma}_\eta$ is the working prior covariance.  Ratings are typically more informative per task than binary choice because they reveal cardinal intensity rather than just the sign of the latent comparison.

\subsection{Multinomial choice}\label{app:ext_multinomial}

When each task offers $J \geq 3$ alternatives and the respondent picks one, the binary logit is replaced by a softmax:
\begin{equation*}
  P(Y_{it} = j \mid \bfX_{it}, \bfbeta_i) = \frac{\exp(\bfX_{itj}^\top \bfbeta_i)}{\sum_{j'=1}^{J} \exp(\bfX_{itj'}^\top \bfbeta_i)}.
\end{equation*}
The respondent-level update uses the multinomial score and Fisher information,
\begin{equation*}
  \hat{\boldsymbol{\eta}}_i = \arg\max_{\boldsymbol{\eta}} \sum_{t} \log P(Y_{it} \mid \bfX_{it}, \hat f(\bfZ_i) + \boldsymbol{\eta}) - \tfrac{1}{2} \boldsymbol{\eta}^\top \boldsymbol{\Sigma}_\eta^{-1} \boldsymbol{\eta},
\end{equation*}
solved by Newton iteration with the multinomial Hessian $-\sum_t \bfX_{it}^\top (\operatorname{diag}(\mathbf{p}_{it}) - \mathbf{p}_{it}\mathbf{p}_{it}^\top) \bfX_{it} - \boldsymbol{\Sigma}_\eta^{-1}$.  A task with $J$ alternatives can be more informative than a binary forced choice because the selected alternative is compared with several others at once, but the information gain is not a fixed multiple: it depends on the design matrix, the choice probabilities, and the utility region in which the task falls.

\subsection{Rankings}\label{app:ext_rankings}

When respondents rank $J$ alternatives from most to least preferred, the natural likelihood is Plackett--Luce (also known as rank-ordered or exploded logit):
\begin{equation*}
  P(\pi_{i} \mid \bfX_{i}, \bfbeta_i) = \prod_{r=1}^{J-1} \frac{\exp(\bfX_{i\pi(r)}^\top \bfbeta_i)}{\sum_{j' = r}^{J} \exp(\bfX_{i\pi(j')}^\top \bfbeta_i)},
\end{equation*}
where $\pi(r)$ is the alternative placed in position $r$.  Because each factor is a conditional multinomial choice over the still-unranked alternatives, the mean stage and empirical-Bayes update of Supplementary Materials~\ref{app:ext_multinomial} apply after summing the log-likelihood and derivatives over the $J-1$ rank stages. A ranking often contains more preference information than a single binary or multinomial choice because it reveals a full ordering, but its Fisher information again depends on the design matrix, choice probabilities, and utility scale rather than equaling a fixed number of multinomial choices.  Rankings are attractive when the full ordering is substantively meaningful, but cognitive burden grows quickly with $J$ and respondent fatigue can dominate beyond $J = 4$ or $5$.

\end{document}